\documentclass[10pt,twocolumn,english]{IEEEtran}
\usepackage[T1]{fontenc}
\usepackage[latin9]{inputenc}
\usepackage{geometry}
\usepackage{color}
\usepackage{verbatim}
\usepackage{amsmath}
\usepackage{amsthm}
\usepackage{graphicx}

\makeatletter
\theoremstyle{plain}
\newtheorem{thm}{\protect\theoremname}
\theoremstyle{plain}
\newtheorem{lem}[thm]{\protect\lemmaname}
\theoremstyle{plain}
\newtheorem{cor}[thm]{\protect\corollaryname}

\date{}
\usepackage{amsmath,kbordermatrix,blkarray}
\providecommand{\keywords}[1]{\textbf{\textit{Keywords---}} #1}
\usepackage{xcolor}
\usepackage{widetext}

\author{Jin Xu and
Natarajan Gautam, \IEEEmembership{Senior Member,~IEEE}
\thanks{Jin Xu is with School of Science and Engineering, the Chinese University of Hong Kong, Shenzhen, China
(Email: xujin@cuhk.edu.cn)}
\thanks{Natarajan Gautam (Corresponding author) is with Department of Industrial and Systems Engineering,
Texas A\&M University, College Station, TX, USA (Email:gautam@tamu.edu)}}

\@ifundefined{showcaptionsetup}{}{%
 \PassOptionsToPackage{caption=false}{subfig}}
\usepackage{subfig}
\makeatother

\usepackage{babel}
\providecommand{\corollaryname}{Corollary}
\providecommand{\lemmaname}{Lemma}
\providecommand{\theoremname}{Theorem}

\begin{document}
\title{Peak Age of Information in Priority Queueing Systems}
\maketitle
\begin{abstract}
We consider a priority queueing system where a single processor serves
$k$ classes of packets that are generated randomly following Poisson
processes. Our objective is to compute the expected Peak Age of Information
(PAoI) under various scenarios. In particular, we consider two situations
where the buffer size at each queue is one and infinite, and in the
infinite buffer size case we consider First Come First Serve (FCFS)
and Last Come First Serve (LCFS) as service disciplines. For the system
with buffer size one at each queue, we derive PAoI exactly for the
case of exponential service time and bounds (which are excellent approximations)
for the case of general service time, with small $k$. For the system
with infinite buffer size, we provide closed-form expressions of PAoI
for both FCFS and LCFS where service time is general and $k$ could
be large. Using those results we investigated the effect of ordering
of priorities and service disciplines for the various scenarios. We
perform extensive numerical studies to validate our results and develop
insights.

\end{abstract}

\section{Introduction\label{sec:Introduction}}

In the recent years, the notion of Age of Information (AoI) has garnered
attention from many researchers. The main applications that have been
cited include sensor networks, wireless networks, and autonomous vehicle
systems \cite{kaul2012real}, as it is important to know the freshness
of information in all those cases. Our research has been motivated
by an application in smart manufacturing where edge devices, sensors
in particular, with limited processing capabilities, would monitor
the health of various tools, condition of components, and quality
of work pieces in machines. This sensed information would be used
for timely decisions such as tool changes, re-calibration and rework,
thereby improving overall quality of the manufactured products. In
such a scenario, decisions may be made based on delayed information
due to the discontinuous sampling and long information processing
time. Hence it is crucial to consider the freshness of information
in decision making, for some type of which AoI is an ideal choice.

AoI is a metric defined and used by researchers such as Kaul et al
\cite{kaul2012real} to describe the freshness of data. We consider
a system where a data source (sensor or resource) from time to time
sends updates or files (in this paper we call each update or file
a ``packet'') to the processor (also called server). The time when
a packet is generated by a data source is called its arrival time
(also called release time) into the system. The server processes packets
in a non-preemptive way. Unprocessed packets are queued due to the
limited processing capacity of the server. AoI at an arbitrary time
point $t$ is defined as the length of period between time $t$ and
the most recent release time among all the packets that have been
processed. Mathematically, the AoI at time $t$ is defined as $\triangle(t)=t-\max\{r_{l}:C_{l}\leq t\}$,
where $r_{l}$ is the release time of the $l^{th}$ packet that is
generated and $C_{l}$ is the time when it is processed by the server
(also called its completion time). While the time-average AoI could
be a metric to measure data freshness, many researchers consider Peak
Age of Information (PAoI) as a more tractable metric \cite{huang2015optimizing,inoue2019general}.
We let the $n^{th}$ peak value of $\triangle(t)$ be $A_{n}$, which
is a random variable and it is shown in Figure \ref{fig:Age-of-Information}.
By assuming ergodicity, the limit expectation of this peak value,
i.e., $\boldsymbol{E}[A]=\lim_{n\rightarrow\infty}\boldsymbol{E}[A_{n}]$,
is then defined as PAoI for this data source. We will later extend
this notion to multiple sources and formulate our model.

\begin{figure}[h]
\centering{}\includegraphics[scale=0.3]{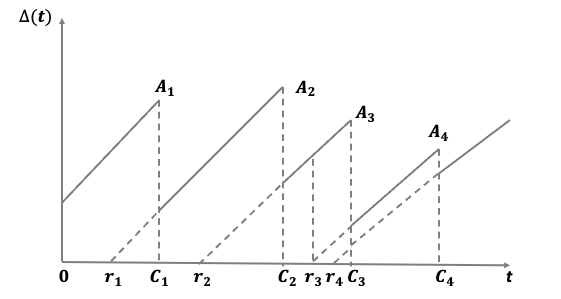}\caption{Age of Information for a Single Queue\label{fig:Age-of-Information}}
\end{figure}

It has been well documented and accepted that monitoring and sensing
according to a Poisson process is effective \cite{masry1978poisson}.
In that light we consider multiple data sources (sensors) that monitor
according to a Poisson process with potentially different rates due
to the difficulty in sensing (recall our motivation example of a manufacturing
setting).\textbf{ }We consider a setting where there are $k$ data
sources prioritized from 1 (highest) to $k$ (lowest). There is a
single server that ``serves'' the $k$ packet streams based on a
static priority mechanism. We study the static queue priorities mainly
because in many cases, there are some data sources whose packets contain
age sensitive information or emergency information such as high temperature,
high pressure (see \cite{najm2019content}). These data packets need
to be transmitted as soon as possible, thus high static priorities
for these data sources are needed. Another example is given by Maatouk
et al \cite{maatouk2019age}, which says that in the vehicle network,
the safety related data should be allocated higher priorities over
the other non-safety related data to improve the traveling experience.
All these real applications motivate us to consider a multi-queue
system with static priorities.

\begin{figure}[h]
\centering{}\includegraphics[scale=0.35]{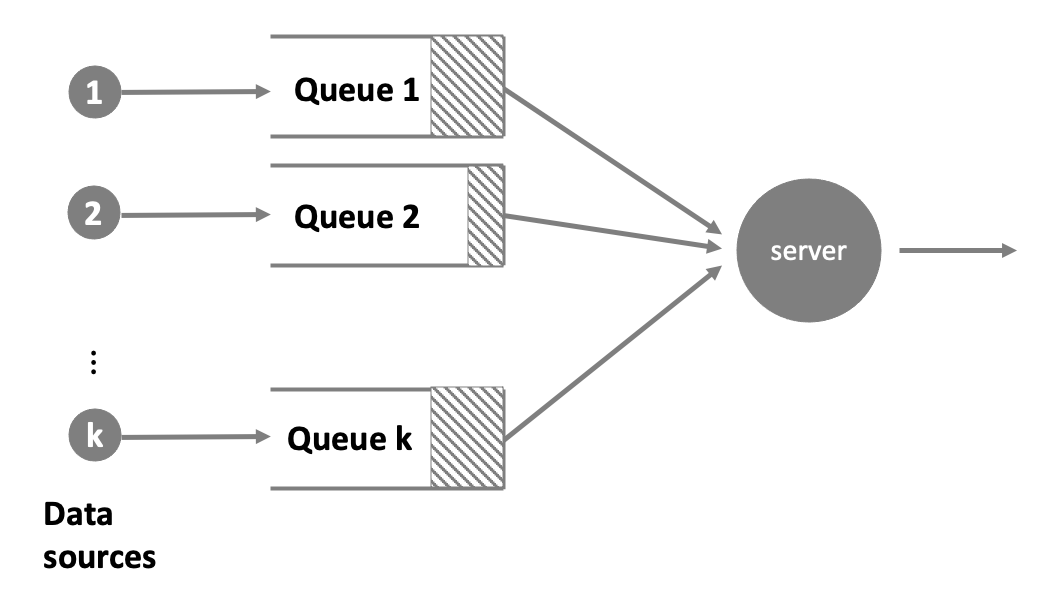}\caption{System Model \label{fig:System-Model}}
\end{figure}

The system model is provided in Figure \ref{fig:System-Model}. We
consider two settings in this paper: one in which there is a buffer
at each queue that can hold at most one packet at a time; the other
where each buffer can hold infinitely many packets. For the first
setting we discuss the system M/G/1/1$+\sum1^{*}$ and M/G/1/1$+\sum1$
for general service times. The notation $1+\sum1^{*}$ means besides
the processing area at the server, each data source has a buffer with
size one. The asterisk means that the packet waiting in the buffer
is replaced by the newest arrival, the same as the notation used in
\cite{costa2016age,zou2019benefis}. If there is no asterisk, i.e.,
M/G/1/1$+\sum1$, then it means that packet that enters the buffer
will not be replaced by new arrivals. For the setting of infinite
buffer size, we discuss the M/G/1 type queues with First Come First
Serve (FCFS) and Last Come First Serve (LCFS) service disciplines
respectively. It is still unknown which setting and which discipline
will result in the smallest PAoI. So our objective is to derive the
PAoI for each setting, and then find the optimal setting and discipline.
The main contributions of our paper are listed as follows:
\begin{enumerate}
\item We first provide a novel modeling method to evaluate PAoI of multi-queue
systems by focusing on the busy period of server and buffer status.
Using this method we provide the exact PAoI for the prioritized system
M/M/1/$1+\sum1^{*}$ and system M/G/1/$1+\sum1$ with small number
of queues $k$. And we further provide the bounds (which are also
excellent approximations) for PAoI of the system M/G/1/$1+\sum1^{*}$.
\item We provide the exact PAoI for the infinite buffer prioritized system
M/G/1 with FCFS and LCFS for general number of queues $k$, and our
analysis for deriving PAoI under LCFS can be applied to other systems
with LCFS as well.
\item By providing the exact PAoI of systems above, we show a surprising
result that LCFS is not the optimal service discipline for minimizing
PAoI among all the non-preemptive work-conserving disciplines in the
system where buffer size of each queue is infinite. We also show a
counter-intuitive finding that having a single buffer at each queue
does not always provide lower PAoI than the having a buffer with infinite
size. We further reveal that it is because the special definition
of the metric PAoI. 
\item We reveal the fact that PAoI of queues with low priorities are sensitive
to the traffic intensity of queues with high priorities, so queues
that contain important or time-sensitive information should be given
high priorities. Also, if the PAoI averaged across queues is to be
minimized, we show that it is beneficial to assign low priorities
to high traffic queues.
\end{enumerate}
The rest of this paper is organized as follows. A summary of the literature
is provided in Section \ref{sec:Related-Work}. Then, in Section \ref{sec:Buffer-Size-One}
we provide the PAoI analysis for M/G/1/1+$\sum1^{*}$ and M/G/1/1$+\sum1$
type queues. In Section \ref{sec:Infinite-Buffer-Size} we provide
the PAoI analysis for queues with infinite buffer size, under both
FCFS and LCFS disciplines within each queue. We perform numerical
studies in Section \ref{sec:Numerical-Study}, and make concluding
remarks as well as discuss the future work in Section \ref{sec:Conclusions-and-Future}.

\section{\label{sec:Related-Work}Related Work}

The idea of data age, information freshness, and timeliness for data
warehouses are introduced and discussed in \cite{theodoratos1999data,bouzeghoub2004framework}.
In recent years, data freshness has drawn much more attention because
of the development of Internet of Things, fog computing and edge data
storage \cite{al2010hedera,vaquero2014finding}. Kaul et al \cite{kaul2012real}
firstly provided average AoI for M/M/1, M/D/1 and D/M/1 type queues.
Costa et al \cite{costa2016age} then obtained analytical results
of average AoI and PAoI under FCFS for M/M/1/1, M/M/1/2 (which allows
drop of new arrivals), as well as M/M/1/2{*} (which allows update
for the waiting packet) queues. The performance of LCFS policy for
the single buffer case where service times are gamma distributed was
provided by Najm and Nasser \cite{najm2016age}. Soysal and Ulukus
\cite{soysal2018age} considered G/G/1/1 type queues and provided
bounds of AoI for different arrival and service processes. Zou et
al \cite{zou2019benefis} discussed PAoI and AoI under the waiting
procedure in M/G/1/1 and M/G/1/2{*} cases. Kosta et al \cite{kosta2019queue}
discussed the performance of AoI and PAoI for the single-queue slotted-time
system with and without packets management. Inoue et al \cite{inoue2019general}
discussed the relationship between PAoI and AoI for the single queue
case, and they provided the AoI/PAoI analysis for different single
queue models including M/G/1 and G/M/1 systems with FCFS, preemptive
LCFS, and non-preemptive LCFS. Some recent works have considered AoI
for the system with single server and multiple queues. Huang and Modiano
\cite{huang2015optimizing} provided the PAoI for multi-class M/G/1
and M/G/1/1 queues where all packets flow into a combined queue. Najm
and Telatar \cite{najm2018status} considered the M/G/1/1 system with
multiple sources updating while allowing preemption. Kosta et al \cite{kosta2019age}
considered a slotted-time system and discussed the performance of
round-robin, working-conserving and random policy. Jiang et al \cite{jiang2018can}
considered an AoI minimization problem with Bernoulli arrivals in
a slotted-time system and modeled it as a MDP problem to determine which
data source to serve next. They showed that Whittle's index policy
is a near optimal policy and they also provided a decentralized policy
which achieves nearly identical performance as the Whittle's index
policy. The optimality of Whittle's index policy was further discussed
by Maatouk et al in \cite{maatouk2020optimality}. Kadota et al \cite{kadota2019scheduling}
considered an AoI minimization problem in a slotted-time system with
throughput constraint considerations. Talak et al \cite{talak2019optimizing}
considered weighted AoI and PAoI minimization problem in a discrete
timed system with channel errors. It is also pointed out by \cite{talak2019optimizing}
that the PAoI/AoI for the discrete time queues may differ significantly
from their continuous time counterpart. Other AoI/PAoI minimization
problems for discrete time systems can be found in \cite{he2017optimal,hsu2017age,jiang2018can,jiang2018timely,jiang2019timely,kadota2018scheduling}.
The multi-class queues with FCFS and LCFS across queues are discussed
in Yates and Kaul \cite{yates2019age}. A detailed review for the
current literature for AoI was also provided in \cite{yates2019age}.

However, we notice that if all the packets arriving into the multi-queue
system are served following FCFS or LCFS regardless of their classes
(queues), queues with high arrival rates will be served more frequently.
It is not always the case that the queues with high arrival rates
are important. Queues with low traffic intensities may also be important,
and their packets may need to be processed as soon as the server becomes
available. Besides, spending too much time processing a certain data
source is a waste of service resource. We thus want to consider a
service policy which gives certain queues higher priorities. Such
a multiple-queue system with queue priorities has been studied for
a long time, however most of previous works focused on different metrics
such as queue lengths and waiting time distributions \cite{jaiswal1968priority,adan2001queueing}.
AoI and PAoI are metrics introduced in recent years, and their performance
under queue priorities are not well understood. Najm et al \cite{najm2019content}
considered a system with two streams of different priorities and discussed
different service disciplines for the low priority stream. Recently,
Kaul and Yates \cite{kaul2018age} modeled the AoI of M/M/1 priority
queues as a hybrid system by assuming the waiting room (buffer size)
for the system is either null or one. However, their model is restrictive
since there is at most one buffer for all queues. Maatouk et al \cite{maatouk2019age}
discussed the model where each queue has an individual buffer, and
provided a closed-form expression for AoI using a hybrid system analysis.
However, it is assumed in \cite{maatouk2019age} that arrival and
service rates for all queues are exponential and identical. It is
still unknown if having finite buffer can help reduce PAoI for each
queue in the multi-queue scenario, especially when arrival and service
rates differ from one queue to another. In our work, we for the first
time provide the exact PAoI for the system where queues are prioritized
and each queue has its own waiting room (buffer), arrival rate, and
service rate. In this paper we only consider static queue priorities
because in many applications, some data streams have more important
information which need to be transmitted as soon as possible. Moreover,
the queue performance is more tractable when queue priorities are
fixed, but the behavior of PAoI in this case is still not fully understood.
In this paper, we provide a new modeling approach of calculating PAoI
by focusing on the buffer status. We derive the exact PAoI for M/M/$1+\sum1^{*}$
system and M/G/$1+\sum1$ system, as well as bounds for M/G/$1+\sum1^{*}$
queues with priorities. Also, for the case of infinite buffer size,
we derive the exact PAoI for M/G/1 system with FCFS and LCFS service
disciplines within each queue. We seek to find a priority order and
service discipline that would result in low average PAoI across queues.
We also seek to understand the effect of arrival rates and service
times on the PAoI for systems under different settings for buffer
size and service disciplines.

\section{Queues with Buffer Size One \label{sec:Buffer-Size-One}}

In this section we discuss the system M/G/1/1+$\sum1^{*}$ with the
arrival process for each queue $i$ being a Poisson process with rate
$\lambda_{i}$. The service time (processing time) $P_{i}$ for packets
from queue $i$ is an i.i.d. random variable with cdf $F_{i}(x)$
and mean $\frac{1}{\mu_{i}}.$ A new arrival will replace the packet
waiting in the queue (if there is one) since the newest packet contains
the most recent information of the source. Note that this model is
different from the M/G/1/1 model introduced in \cite{huang2015optimizing,najm2018status}.
In their model, there is no buffer for each queue, so whenever a packet
arrives and sees the server being busy, the packet is either rejected
or preempts the packet in service. In our model, the assumption of
buffer at each queue allows the server to serve queues in a prioritized
way. Moreover, keeping only the most recent packet in the buffer can
reduce the server's load, and also guarantee that the freshest information
is stored in the buffer.

The difficulty in analyzing such a system with buffer at each queue
is that packets served by the server are only a subset of packets
generated by the data source, due to some getting replaced. Focusing
on how each packet goes through the system often makes modeling more
complicated \cite{kaul2018age}. Instead, in this paper we introduce
a new modeling approach to derive PAoI, which is to incorporate the
buffer state. We can also use this idea to derive PAoI for other systems
with buffer size more than one, as we will see in Section \ref{sec:Infinite-Buffer-Size}.
We depict a sample path of the buffer state for queue $i$ in Figure
\ref{fig:Buffer-State} with notations described subsequently. From
Figure \ref{fig:Buffer-State} we can see that buffer state of queue
$i$ is either $0$ or $1$. When the buffer state is 1 (the buffer
is full), we say the buffer is busy. We use $r_{ij}$, $S_{ij}$ and
$C_{ij}$ to denote the release time, starting time of processing
and completion time of $j^{th}$ packet that arrives at queue $i$
(note that $S_{ij}$ and $C_{ij}$ only exist if the packet is processed
by the server). Suppose at time $r_{i1}$, packet $1$ arrives at
queue $i$. It waits until time $S_{i1}$, when the server becomes
available to serve it. Right after time $S_{i1}$, buffer $i$ keeps
empty until packet 2 arrives at time $r_{i2}$. Packet $2$ stays
in the buffer for a while, then gets replaced by packet $3$ at time
$r_{i3}$. Packet 3 is then replaced by packet 4 at time $r_{i4}.$
At time $S_{i4}$ the server becomes available and starts serving
packet 4, and the buffer becomes empty again. The service of packet
4 is completed at time $C_{i4},$ and the peak age of information
upon the completion of packet 4 is given as $C_{i4}-r_{i1}$, which
is equal to 
\begin{eqnarray}
C_{i4}-r_{i1} & = & (C_{i4}-S_{i4})+(S_{i4}-r_{i2})\nonumber \\
 &  & +(r_{i2}-S_{i1})+(S_{i1}-r_{i1}).\label{eq:1}
\end{eqnarray}
The term $(C_{i4}-S_{i4})$ of Equation (\ref{eq:1}) is the processing
time of packet 4, and $(S_{i4}-r_{i2})$ is the time period during
which the buffer has one packet. The third term $(r_{i2}-S_{i1})$
is the time period during which the buffer stays empty, and the last
term $(S_{i1}-r_{i1})$ is the waiting time of packet 1. Recall that
the processing times of packets from the same source are i.i.d., so
the expected value of $(C_{i4}-S_{i4})$ is $\boldsymbol{E}[P_{i}]=\frac{1}{\mu_{i}}.$
The buffer is empty during time $(r_{i2}-S_{i1})$, and from the memoryless
property of exponential inter-arrival times, we know the expected
length of $(r_{i2}-S_{i1})$ is $\boldsymbol{E}[I_{i}]=\frac{1}{\lambda_{i}}$.
Therefore we can write the PAoI for source $i$ as
\begin{eqnarray}
\boldsymbol{E}[A_{i}] & = & \boldsymbol{E}[P_{i}]+\boldsymbol{E}[W_{i}]+\boldsymbol{E}[I_{i}]+\boldsymbol{E}[G_{i}],\label{eq:2}
\end{eqnarray}

where $\boldsymbol{E}[G_{i}]$ is the expected waiting time (in buffer)
of a packet that is eventually processed by the server, and $\boldsymbol{E}[W_{i}]$
is the expected length of time period when the buffer is continuously
busy. Note that Equation (\ref{eq:2}) holds true for every queue
$i$. For M/G/1/1+$sum1^*$ type queues, we have already stated that $\boldsymbol{E}[P_{i}]=\frac{1}{\mu_{i}}$
and $\boldsymbol{E}[I_{i}]=\frac{1}{\lambda_{i}}$. The difficult
part in deriving PAoI remains in calculating $\boldsymbol{E}[W_{i}]$
and $\boldsymbol{E}[G_{i}]$. Notice that $W_{i}$ is the period that
the buffer is full, and it is not determined by which packet we keep
in the buffer. If we reject the new arrivals (instead of the system
that we are analyzing) when the buffer is full, then $W_{i}$ is the
waiting time for the packet that enters the buffer. Using this property,
if we let $p_{i}$ be the probability that buffer $i$ is full, then
from Little's Law \cite{little2011or} we know the average queue length
is $p_{i}=\lambda_{i}(1-p_{i})\boldsymbol{E}[W_{i}]$. So we have
\begin{eqnarray}
\boldsymbol{E}[W_{i}] & = & \frac{p_{i}}{\lambda_{i}(1-p_{i})}.\label{eq:3}
\end{eqnarray}

\begin{figure}[h]
\begin{centering}
\includegraphics[scale=0.28]{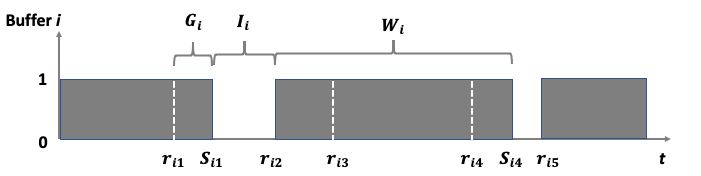}\caption{Buffer State for Queue $i$ \label{fig:Buffer-State}}
\par\end{centering}
\end{figure}

From Equation (\ref{eq:3}), $\boldsymbol{E}[W_{i}]$ can be obtained
once we know $p_{i}.$ We shall discuss how to find $p_{i}$ later
in this section. Now we continue with the system where new arrivals
replace the existing ones in buffer. We first characterize $G_{i}$,
which depends on $W_{i}$, as we will see in Lemma \ref{lem:1}.
\begin{lem}
$\boldsymbol{E}[G_{i}|W_{i}]=\frac{1}{\lambda_{i}}(1-\boldsymbol{E}[e^{-\lambda_{i}W_{i}}]).$\label{lem:1}
\end{lem}
\begin{IEEEproof}
Suppose $N(t)=m$ packets arrive during $W_{i}$, then $G_{i}$ is
the time gap from the arrival time of the $m^{th}$ packet to time
$t$. From Campbell's Theorem (P173, Theorem 5.14 in \cite{kulkarni2016modeling})
we have for $x\leq t$, 

\begin{eqnarray*}
 &  & \boldsymbol{P}(G_{i}<x|N(t)=m,W_{i}=t)\\
 & = & \boldsymbol{P}(R_{m}>t-x|N(t)=m,W_{i}=t)\\
 & = & \int_{t-x}^{t}\frac{m}{t}(\frac{u}{t})^{m-1}du\\
 & = & 1-(\frac{t-x}{t})^{m}.
\end{eqnarray*}

Thus by integrating $\boldsymbol{P}(G_{i}>x|N(t)=m,W_{i}=t)$ for
$x$ from $0$ to $t$, we have
\begin{eqnarray*}
 &  & \boldsymbol{E}[G_{i}|N(t)=m,W_{i}=t]=\frac{t}{m+1}.
\end{eqnarray*}
Then, unconditioning using $\boldsymbol{P}(N(t)=m)=e^{-\lambda_{i}t}\frac{(\lambda_{i}t)^{m}}{m!}$,
we get
\begin{eqnarray*}
\boldsymbol{E}[G_{i}|W_{i}=t] & = & \sum_{m=0}^{\infty}\frac{t}{m+1}e^{-\lambda_{i}t}\frac{(\lambda_{i}t)^{m}}{m!}\\
 & = & \sum_{m=0}^{\infty}e^{-\lambda_{i}t}\frac{(\lambda_{i}t)^{m+1}}{(m+1)!}\frac{1}{\lambda_{i}}\\
 & = & \frac{e^{-\lambda_{i}t}}{\lambda_{i}}(e^{\lambda_{i}t}-1).
\end{eqnarray*}
\end{IEEEproof}
Lemma \ref{lem:1} shows that one needs to know the Laplace\textendash Stieltjes
transform (LST) to get $\boldsymbol{E}[G_{i}]$. The exact LST of
$W_{i}$ can be obtained when service times are exponentially distributed,
as we will see in Section \ref{subsec:Exact-Analysis-for}. If service
times are generally distributed, we provide the bounds for PAoI based
on result of Lemma \ref{lem:1}, which we will see in Section \ref{subsec:Bounds-and-Approximation}.

\subsection{Exact Analysis for M/M/1/1+$\mathbf{\sum1^{*}}$ Type Queues\label{subsec:Exact-Analysis-for}}

In this subsection we consider a special case where the processing
time $P_{i}$ is exponentially distributed with $\boldsymbol{E}[P_{i}]=\frac{1}{\mu_{i}}$
for $i=1,2,...,k$. Knowing the LST of $W_{i}$ can help us obtain
both $\boldsymbol{E}[W_{i}]$ and $\boldsymbol{E}[G_{i}]$, so in
this subsection we focus on calculating LST of $W_{i}.$ Since $W_{i}$
is not affected by which packet we keep in the buffer, in this subsection,
we assume new arrivals are rejected if the buffer is full. We adopt
the method used to characterize the busy period in \cite{conway2003theory}
to derive the LST of $W_{i}$, i.e., $\boldsymbol{E}[e^{-sW_{i}}]$.
Let $B_{i}(t)$ be the number of packets in buffer $i$ at time $t$
with $B_{i}(t)\in\{0,1\}$. Let $J(t)\in\{0,1,...,k\}$ be the packet
class that is in service at time $t$, where $J(t)=0$ means the server
is idling. The vector $S(t)=(J(t),B_{1}(t),...,B_{k}(t))$ thus indicates
the state of the system at time $t$. From PASTA \cite{kulkarni2016modeling}
we know that the time average performance of the system is the same
as that seen by Poisson arrivals. If a packet from class $i$ arriving
at time $t$ sees $B_{i}(t^{-})=0$, it then enters the buffer if
the server is busy, or enters the server directly if the server is
idling. Thus the state observed by a packet that enters buffer $i$
(right before its entering time) is always $B_{i}(t^{-})=0$. We let
$\psi_{j}(s)=\frac{\mu_{j}}{\mu_{j}+s}$ be the LST of service time
for packets from queue $j$. Because the service time is exponential,
$\psi_{j}(s)$ is also the LST of remaining service time of the packet
observed by an entering packet, if a class $j$ packet is in service.
Let $U_{i}$ be the remaining service time observed by a packet entering
queue $i$. If a packet from class 1 enters at time 0, then we have

\begin{eqnarray*}
 &  & \boldsymbol{E}[e^{-sW_{1}}|B_{1}(0^{-})=0]\\
 & = & \boldsymbol{E}[e^{-sU_{1}}|B_{1}(0^{-})=0]\\
 & = & \boldsymbol{P}(J(0^{-})=0|B_{1}(0^{-})=0)\\
 &  & +\sum_{j=1}^{k}\psi_{j}(s)\boldsymbol{P}(J(0^{-})=j|B_{1}(0^{-})=0).
\end{eqnarray*}

Before characterizing $\boldsymbol{E}[e^{-sW_{2}}]$ for buffer 2,
we first introduce the busy period of the server. Let $T_{1}$ be
the time period that the server is continuously busy processing packets
from buffer 1, and $\eta_{1}(s)=\boldsymbol{E}[e^{-sT_{1}}]$. The
busy period $T_{1}$ always starts from processing a packet from buffer
1. Suppose the processing time of this packet is of length $P_{1}=l$
and if there is more than one priority 1 packet arriving during $[0,l]$,
then another busy period will start from time $l$ and the new busy
period is identically distributed as $T_{1}$. Thus we have $\boldsymbol{E}[e^{-s(l+T_{1})}|P_{1}=l,B_{1}(l)=1]=e^{-sl}\eta_{1}(s)$.

If there is no arrival during $P_{1}=l$, then the busy period would
be $l$ only. By unconditioning on $B_{1}(l)$ we have 

$\boldsymbol{E}[e^{-(l+T_{1})}|P_{1}=l]=e^{-sl}\eta_{1}(s)(1-e^{-\lambda_{1}l})+e^{-sl}e^{-\lambda_{1}l}.$ 

Unconditioning on $P_{1}=l$ we have 
\begin{eqnarray*}
\eta_{1}(s) & = & \eta_{1}(s)[\psi_{1}(s)-\psi_{1}(s+\lambda_{1})]+\psi_{1}(s+\lambda_{1}).
\end{eqnarray*}

Thus the LST of $T_{1}$ is given by $\eta_{1}(s)=\frac{\psi_{1}(s+\lambda_{1})}{1-\psi_{1}(s)+\psi_{1}(s+\lambda_{1})}=\frac{\mu_{1}(s+\mu_{1})}{s^{2}+2\mu_{1}s+s\lambda_{1}+\mu_{1}^{2}}$
and the derivative of $\eta_{1}(s)$ at $s=0$ is given by $\eta_{1}^{'}(s)|_{s=0}=\frac{-\lambda_{1}-\mu_{1}}{\mu_{1}^{2}}.$

Now we characterize the LST of $W_{2}$ by the fact that $\boldsymbol{E}[e^{-sW_{2}}]=\boldsymbol{E}[e^{-s(U_{2}+T_{1})}]$
and conditioning on different scenarios observed by a packet that
enters buffer 2 at time 0. If the server is idling when the packet
enters the buffer, then $\boldsymbol{E}[e^{-s(U_{2}+T_{1})}|B_{1}(0^{-})=0,J(0^{-})=0,B_{2}(0^{-})=0]=1.$

If the server is busy processing a packet from buffer $j$ for $j\in\{1,...,k\}$,
and buffer 1 is not empty, then we have 

\begin{eqnarray*}
 &  & \boldsymbol{E}[e^{-s(U_{2}+T_{1})}|B_{1}(0^{-})=1,J(0^{-})=j,B_{2}(0^{-})=0]\\
 &  & =\boldsymbol{E}[e^{-sU_{2}}|J(0^{-})=j]\boldsymbol{E}[e^{-sT_{1}}]=\psi_{j}(s)\eta_{1}(s).
\end{eqnarray*}

If the server is busy processing a packet from buffer $j$ for $j\in\{1,...,k\}$,
and buffer 1 is empty, then we have 
\begin{eqnarray*}
 &  & \boldsymbol{E}[e^{-s(U_{2}+T_{1})}|B_{1}(0^{-})=0,J(0^{-})=j,\\
 &  & U_{2}=u,B_{1}(u)=1,B_{2}(0^{-})=0]=e^{-su}\eta_{1}(s),
\end{eqnarray*}
 and 
\begin{eqnarray*}
 &  & \boldsymbol{E}[e^{-s(U_{2}+T_{1})}|B_{1}(0^{-})=0,J(0^{-})=j,\\
 &  & U_{2}=u,B_{1}(u)=0,B_{2}(0^{-})=0]=e^{-su}.
\end{eqnarray*}
 By unconditioning on $B_{1}(u)$ we have 
\begin{eqnarray*}
 &  & \boldsymbol{E}[e^{-s(U_{2}+T_{1})}|B_{1}(0^{-})=0,J(0^{-})=j,\\
 &  & U_{2}=u,B_{2}(0^{-})=0]\\
 &  & =e^{-su}\eta_{1}(s)(1-e^{-\lambda_{1}u})+e^{-su}e^{-\lambda_{1}u}.
\end{eqnarray*}
 By unconditioning on $U_{2}=u$ we have 
\begin{eqnarray*}
 &  & \boldsymbol{E}[e^{-s(U_{2}+T_{1})}|B_{1}(0^{-})=0,J(0^{-})=j,B_{2}(0^{-})=0]\\
 &  & =\psi_{j}(s)\eta_{1}(s)-\psi_{j}(s+\lambda_{1})\eta_{1}(s)+\psi_{j}(s+\lambda_{1}).
\end{eqnarray*}

So far we have characterized the LST of $W_{2}$ conditioning on different
scenarios observed by entering packets. We only need the probabilities
of $\boldsymbol{P}(B_{1}(0^{-})=\{0,1\},J(0^{-})=j|B_{2}(0^{-})=0)$
to obtain $\boldsymbol{E}[e^{-sW_{2}}]$, which we will discuss at
the end of this subsection. Before doing that, we now consider how
to obtain the LST of $W_{3}$ by conditioning on different scenarios.
For simplicity of analysis we here assume $\lambda_{1}=\lambda_{2}$
and $\mu_{1}=\mu_{2}$. The argument for distinct $\lambda_{1}$ and
$\lambda_{2}$ or $\mu_{1}$ and $\mu_{2}$ are similar, however notationally
cumbersome. We let $T_{12}$ be the busy time during which the server
continuously serves packets from buffer 1 and buffer 2 and let $B_{12}(t)=B_{1}(t)+B_{2}(t).$
We now characterize the LST of $T_{12}$ by letting $\eta_{12,0}(s)=\boldsymbol{E}[e^{-sT_{12}}\mid B_{12}(0^{+})=0]$
and $\eta_{12,1}(s)=\boldsymbol{E}[e^{-sT_{12}}\mid B_{12}(0^{+})=1]$.

We suppose the busy period $T_{12}$ starts at time 0 by processing
a packet from either buffer 1 or buffer 2 with processing time $P_{1}=l$.
Then we have 
\begin{eqnarray*}
 &  & \boldsymbol{E}[e^{-s(l+T_{12})}|B_{12}(0^{+})=0,P_{1}=l,B_{12}(l)=0]\\
 &  & =e^{-sl},
\end{eqnarray*}

\begin{eqnarray*}
 &  & \boldsymbol{E}[e^{-s(l+T_{12})}|B_{12}(0^{+})=0,P_{1}=l,B_{12}(l)=1]\\
 &  & =e^{-sl}\eta_{12,0}(s),
\end{eqnarray*}

\begin{eqnarray*}
 &  & \boldsymbol{E}[e^{-s(l+T_{12})}|B_{12}(0^{+})=0,P_{1}=l,B_{12}(l)=2]\\
 &  & =e^{-sl}\eta_{12,1}(s),
\end{eqnarray*}

\begin{eqnarray*}
 &  & \boldsymbol{E}[e^{-s(l+T_{12})}|B_{12}(0^{+})=1,P_{1}=l,B_{12}(l)=1]\\
 &  & =e^{-sl}\eta_{12,0}(s),
\end{eqnarray*}
and

\begin{eqnarray*}
 &  & \boldsymbol{E}[e^{-s(l+T_{12})}|B_{12}(0^{+})=1,P_{1}=l,B_{12}(l)=2]\\
 &  & =e^{-sl}\eta_{12,1}(s).
\end{eqnarray*}

Note that $B_{12}(0^{+})=2$ has probability 0 since right after time
$0$, a packet from either buffer 1 or 2 is in service. Unconditioning
on $B_{12}(l),$ we have 

\begin{eqnarray*}
 &  & \boldsymbol{E}[e^{-s(l+T_{12})}|B_{12}(0^{+})=0,P_{1}=l]\\
 & = & e^{-sl}e^{-2\lambda_{1}l}+2(1-e^{-\lambda_{1}l})e^{-\lambda_{1}l}e^{-sl}\eta_{12,0}(s)\\
 &  & +e^{-sl}(1-e^{-\lambda_{1}l})^{2}\eta_{12,1}(s),
\end{eqnarray*}
and 
\begin{eqnarray*}
 &  & \boldsymbol{E}[e^{-s(l+T_{12})}|B_{12}(0^{+})=1,P_{1}=l]\\
 & = & e^{-sl}e^{-\lambda_{1}l}\eta_{12,0}(s)+e^{-sl}(1-e^{-\lambda_{1}l})\eta_{12,1}(s).
\end{eqnarray*}

Unconditioning on $P_{1}=l$, we have 
\begin{eqnarray*}
 & \eta_{12,0}(s) & =\psi_{1}(s+2\lambda_{1})\\
 &  & +2[\psi_{1}(s+\lambda_{1})-\psi_{1}(s+2\lambda_{1})]\eta_{12,0}(s)\\
 &  & +[\psi_{1}(s)-2\psi_{1}(s+\lambda_{1})+\psi_{1}(s+2\lambda_{1})]\eta_{12,1}(s),
\end{eqnarray*}
 and 
\begin{eqnarray*}
\eta_{12,1}(s) & = & \psi_{1}(s+\lambda_{1})\eta_{12,0}(s)\\
 &  & +[\psi_{1}(s)-\psi_{1}(s+\lambda_{1})]\eta_{12,1}(s).
\end{eqnarray*}

By solving the two equations above for $\eta_{12,0}(s)$ and $\eta_{12,1}(s)$,
we have \begin{widetext}

\begin{eqnarray*}
\eta_{12,0}(s) & = & \frac{\psi_{1}(s+2\lambda_{1})}{1-2[\psi_{1}(s+\lambda_{1})-\psi_{1}(s+2\lambda_{1})]-\frac{\psi_{1}(s+\lambda_{1})}{1-\psi_{1}(s)+\psi_{1}(s+\lambda_{1})}[\psi_{1}(s)-2\psi_{1}(s+\lambda_{1})+\psi_{1}(s+2\lambda_{1})]},
\end{eqnarray*}

\end{widetext}

and

\begin{eqnarray*}
\eta_{12,1}(s) & = & \frac{\eta_{12,0}(s)\psi_{1}(s+\lambda_{1})}{1-\psi_{1}(s)+\psi_{1}(s+\lambda_{1})}.
\end{eqnarray*}

Recall that $U_{3}$ is the remaining service time observed by a packet
that enters buffer 3 at time 0, we then have the LST of busy period
of buffer 3 as conditioned on various scenarios: 

\begin{eqnarray*}
\boldsymbol{E}[e^{-s(U_{3}+T_{12})}|B_{12}(0^{-})=0,J(0^{-})=0,B_{3}(0^{-})=0]=1,
\end{eqnarray*}

\begin{eqnarray*}
 &  & \boldsymbol{E}[e^{-s(U_{3}+T_{12})}|B_{12}(0^{-})=0,J(0^{-})=j,B_{3}(0^{-})=0]\\
 & = & \psi_{j}(s+2\lambda_{1})+2[\psi_{j}(s+\lambda_{1})-\psi_{j}(s+2\lambda_{1})]\eta_{12,0}(s)\\
 &  & +[\psi_{j}(s)-2\psi_{j}(s+\lambda_{1})+\psi_{j}(s+2\lambda_{1})]\eta_{12,1}(s),
\end{eqnarray*}

\begin{eqnarray*}
 &  & \boldsymbol{E}[e^{-s(U_{3}+T_{12})}|B_{12}(0^{-})=1,J(0^{-})=j,B_{3}(0^{-})=0]\\
 &  & =\psi_{j}(s+\lambda_{1})\eta_{12,0}(s)+[\psi_{j}(s)-\psi_{j}(s+\lambda_{1})]\eta_{12,1}(s),
\end{eqnarray*}

and 
\begin{eqnarray*}
 &  & \boldsymbol{E}[e^{-s(U_{3}+T_{12})}|B_{12}(0^{-})=2,J(0^{-})=j,B_{3}(0^{-})=0]\\
 &  & =\psi_{j}(s)\eta_{12,1}(s).
\end{eqnarray*}

Thus we can characterize the LST of $W_{3}$ once we know the stationary
probability of each scenario. For queues with lower priorities, the
analysis requires more argument, but they are all similar (albeit
cumbersome notationally). To get the stationary probability of each
scenario, we model $S(t)=(J(t),B_{1}(t),B_{2}(t),...,B_{k}(t))$ as
a continuous time Markov chain (CTMC) and obtain the stationary probabilities. Here we only show
the example for the case of $k=2$. For $k>2$ the analysis is similar.
The rate matrix $Q$ of the two-queue case is given as follows:

\begin{widetext}

\footnotesize {%
}

\scriptsize{

\begin{equation*} Q =      \kbordermatrix{  & (0,0,0) & (1,0,0) & (2,0,0) & (1,1,0) & (1,0,1) & (2,1,0) & (2,0,1) & (1,1,1) & (2,1,1)\\ (0,0,0) & -\lambda_{1}-\lambda_{2} & \lambda_{1} & \lambda_{2} & 0 & 0 & 0 & 0 & 0 & 0\\ (1,0,0) & \mu_{1} & -\lambda_{1}-\lambda_{2}-\mu_{1} & 0 & \lambda_{1} & \lambda_{2} & 0 & 0 & 0 & 0\\ (2,0,0) & \mu_{2} & 0 & -\lambda_{1}-\lambda_{2}-\mu_{2} & 0 & 0 & \lambda_{1} & \lambda_{2} & 0 & 0\\ (1,1,0) & 0 & \mu_{1} & 0 & -\lambda_{2}-\mu_{1} & 0 & 0 & 0 & \lambda_{2} & 0\\ (1,0,1) & 0 & 0 & \mu_{1} & 0 & -\lambda_{1}-\mu_{1} & 0 & 0 & \lambda_{1} & 0\\ (2,1,0) & 0 & \mu_{2} & 0 & 0 & 0 & \lambda_{2}-\mu_{2} & 0 & 0 & \lambda_{2}\\ (2,0,1) & 0 & 0 & \mu_{2} & 0 & 0 & 0 & -\lambda_{1}-\mu_{2} & 0 & \lambda_{1}\\ (1,1,1) & 0 & 0 & 0 & 0 & \mu_{1} & 0 & 0 & -\mu_{1} & 0\\ (2,1,1) & 0 & 0 & 0 & 0 & \mu_{2} & 0 & 0 & 0 & -\mu_{2} }.\qquad 
\end{equation*}

}

\normalsize{}

\end{widetext}

The stationary distribution $\hat{\pi}$ (which is a vector) is given
by solving $\hat{\pi}Q=0$ and $\hat{\pi}\mathbf{1}=1$, and we have
\begin{eqnarray*}
p_{1} & = & \hat{\pi}(1,1,0)+\hat{\pi}(2,1,0)+\hat{\pi}(1,1,1)+\hat{\pi}(2,1,1),
\end{eqnarray*}
\begin{eqnarray*}
p_{2} & = & \hat{\pi}(1,0,1)+\hat{\pi}(2,0,1)+\hat{\pi}(1,1,1)+\hat{\pi}(2,1,1),
\end{eqnarray*}
\begin{eqnarray*}
\boldsymbol{P}(J(0^{-})=1|B_{1}(0^{-})=0) & = & \frac{\hat{\pi}(1,0,0)+\hat{\pi}(1,0,1)}{1-p_{1}},
\end{eqnarray*}
and 
\begin{eqnarray*}
\boldsymbol{P}(B_{1}(0^{-})=0,J(0^{-})=1|B_{2}(0^{-})=0) & = & \frac{\hat{\pi}(1,0,0)}{1-p_{2}}.
\end{eqnarray*}
The other conditional probabilities can be calculated similarly. 

In summary, in order to obtain the exact PAoI for queue $i$ in M/M/1/1+$\mathbf{\sum1^{*}}$
type queues, one needs to first have the LST of $W_{i}$ conditioning
on each event of $(B_{1}(0^{-})=\{0,1\},...,B_{i-1}(0^{-})=\{0,1\},B_{i}(0^{-})=0,J(0^{-})=\{0,1,...,k\})$,
then by PASTA and the CTMC analysis to obtain the steady state probability
of each event of $(B_{1}(0^{-})=\{0,1\},...,B_{i-1}(0^{-})=\{0,1\},B_{i}(0^{-})=0,J(0^{-})=\{0,1,...,k\})$.
By further unconditioning on each event one can eventually get the
LST of $W_{i}$. This approach becomes cumbersome when the number
of queues becomes large. However, this modeling method by focusing
on the busy period of the server could be useful in many cases as we
will see in Section \ref{sec:Infinite-Buffer-Size}.

\subsection{Bounds and Approximation for M/G/1/1+$\mathbf{\sum1^{*}}$ Type Queues\label{subsec:Bounds-and-Approximation}}

Here we generalize the analysis in Subsection \ref{subsec:Exact-Analysis-for}
to the M/G/1/1$+\sum1^{*}$ system, where service times are general.
The CTMC analysis used in Subsection \ref{subsec:Exact-Analysis-for}
cannot be applied here. However, since arrivals still follow Poisson
processes, Lemma \ref{lem:1} holds. We can write the PAoI of queue
$i$ as 
\begin{eqnarray}
\boldsymbol{E}[A_{i}] & = & \boldsymbol{E}[P_{i}]+\boldsymbol{E}[W_{i}]+\boldsymbol{E}[I_{i}]+\boldsymbol{E}[G_{i}]\nonumber \\
 & = & \frac{1}{\mu_{i}}+\frac{p_{i}}{\lambda_{i}(1-p_{i})}+\frac{2}{\lambda_{i}}-\frac{1}{\lambda_{i}}\boldsymbol{E}[e^{-\lambda_{i}W_{i}}]\nonumber \\
 & \leq & \frac{1}{\mu_{i}}+\frac{p_{i}}{\lambda_{i}(1-p_{i})}+\frac{2}{\lambda_{i}}-\frac{1}{\lambda_{i}}e^{-\frac{p_{i}}{1-p_{i}}}.\label{eq:5}
\end{eqnarray}

Inequality (\ref{eq:5}) follows from the Jensen's inequality by knowing
that $e^{-\lambda_{i}x}$ is a convex function. Notice that Inequality
(\ref{eq:5}) gives an upper bound of PAoI in terms of probability
$p_{i}$ (which is the steady state probability that buffer $i$ is
full). Takenaka \cite{takenaka1989analysis} considered a multi-queue
M/G/1 system with each queue having a unique buffer size. Our single-buffer
system thus becomes a special case of the model in Takenaka \cite{takenaka1989analysis}.
Takenaka \cite{takenaka1989analysis} introduced the relationship
between $p_{i}$ and the stationary state seen by departures, for
the system in which service times for all queues are identically distributed
with $F_{i}(x)=F(x)$ and $\mu_{i}=\mu$ for all $i$. Thus one can
get the stationary distribution of states by solving an embedded Markov
chain. It is important to note that the result in \cite{takenaka1989analysis}
only works for identically distributed service times. For heterogenous
service times with $k>2$, the results are difficult to obtain \cite{takenaka1984buffer,takenaka1989analysis}.
So till the end of this subsection, we assume that service times for
packets across queues are identically distributed. To use the result
in \cite{takenaka1989analysis} to get $p_{i}$'s, we first introduce
some notations here. Let $\psi(s)$ be the LST of service time. Let
$\mathcal{S}_{k}$ be our original system which has $k$ queues. Say
$\mathcal{S}_{l}$ is the subsystem of $\mathcal{S}_{k}$ which contains
only queue 1 to queue $l$, and packets from queue $l+1$ to $k$
do not arrive in system $\mathcal{S}_{l}$. Let $\pi_{l}(B_{1},B_{2},...,B_{l})$
be the stationary distribution in which the system $\mathcal{S}_{l}$
has $B_{i}\in\{0,1\}$ number of packets in queue $i$ immediately
after the departure of a packet. Now we re-write a theorem from \cite{takenaka1989analysis}
for our model.
\begin{thm}
\label{thm:(Theorem-3-of}(Theorem 3 of \cite{takenaka1989analysis})
The steady state probability of the buffer with size one at queue
$i$ being full is given by $p_{i}=1-\frac{\pi_{i-1}(0,...0)-\pi_{i}(0,...,0)}{\frac{\lambda_{i}}{\mu}+\frac{\lambda_{i}}{\sum_{j=1}^{k}\lambda_{j}}\pi_{k}(0,...,0)}-\frac{\pi_{k}(0,...,0)}{\frac{\sum_{j=1}^{k}\lambda_{j}}{\mu}+\pi_{k}(0,...,0)}$
for all $i\in\{1,2,...,k\}$, where $\pi_{0}(0,...,0)=1$.
\end{thm}
To obtain the probability $p_{i}$, we only need to find the stationary
distribution that is seen by departures. For that, we model the system
state seen by departures as an embedded Markov chain. We only introduce
the case for $k=2$ here. For $k>2$ the analysis is similar but not
presented here for notational and space restrictions. Since the departure
can see at most one packet waiting at each buffer, the transition
matrix for $k=2$ is given as follows:

$$\tilde{P}_2 =      \kbordermatrix{ & (0,0) & (0,1) & (1,0) & (1,1) \\       (0,0) & a_{0} & a_{1} & a_{2} & a_{3}\\ (0,1) & a_{0} & a_{1} & a_{2} & a_{3}\\ (1,0) & a_{0} & a_{1} & a_{2} & a_{3}\\ (1,1) & 0 & b_{0} & 0 & 1-b_{0} },\qquad $$

where $a_{0}=\int_{0}^{\infty}e^{-(\lambda_{1}+\lambda_{2})x}dF(x)$,
$a_{1}=\int_{0}^{\infty}e^{-\lambda_{1}x}(1-e^{-\lambda_{2}x})dF(x)$,
$a_{2}=\int_{0}^{\infty}(1-e^{-\lambda_{1}x})e^{-\lambda_{2}x}dF(x)$,
$a_{3}=\int_{0}^{\infty}(1-e^{-\lambda_{1}x})(1-e^{-\lambda_{2}x})dF(x)$
and $b_{0}=\int_{0}^{\infty}e^{-\lambda_{1}x}dF(x)$. The stationary
distribution $\pi_{2}(0,0)$ can thus be obtained by solving the linear
system $\pi_{2}\tilde{P}_{2}=\pi_{2}$ with $\pi_{2}\boldsymbol{1}=1,$
where $\pi_{2}=(\pi_{2}(0,0),\pi_{2}(0,1),\pi_{2}(1,0),\pi_{2}(1,1))$.
Notice from Theorem \ref{thm:(Theorem-3-of} that we also need $\pi_{1}(0)$
to get $p_{i}$'s. To obtain $\pi_{1}(0)$ we solve the subsystem
$\mathcal{S}_{1}$ with $\pi_{1}\tilde{P}_{1}=\pi_{1}$ and $\pi_{1}\boldsymbol{1}=1,$
where the transition matrix $\tilde{P}_{1}$ of the embedded Markov
chain is given by

$$\tilde{P}_1 =      \kbordermatrix{ & (0) & (1)\\ (0) & b_{0} & 1-b_{0}\\ (1) & b_{0} & 1-b_{0} }.\qquad $$

By solving the embedded Markov chains, we have $\pi_{1}(0)=\psi(\lambda_{1})$
and $\pi_{2}(0,0)=\frac{\psi(\lambda_{1}+\lambda_{2})\psi(\lambda_{1})}{1-\psi(\lambda_{2})+\psi(\lambda_{1}+\lambda_{2})}$.
Then using Theorem \ref{thm:(Theorem-3-of}, we have $p_{1}=1-\frac{1-\psi(\lambda_{1})}{\frac{\lambda_{1}}{\mu}+\frac{\lambda_{1}}{\lambda_{1}+\lambda_{2}}\pi_{2}(0,0)}-\frac{\pi_{2}(0,0)}{\frac{\lambda_{1}+\lambda_{2}}{\mu}+\pi_{2}(0,0)}$
and $p_{2}=1-\frac{\psi(\lambda_{1})-\pi_{2}(0,0)}{\frac{\lambda_{2}}{\mu}+\frac{\lambda_{2}}{\lambda_{1}+\lambda_{2}}\pi_{2}(0,0)}-\frac{\pi_{2}(0,0)}{\frac{\lambda_{1}+\lambda_{2}}{\mu}+\pi_{2}(0,0)}$.
In summary, to obtain the probability $p_{i}$ of a system with $k$
queues, one needs to compute the stationary probability $\pi_{j}(0,...,0)$
for $j=1,...,i-1$ by solving the embedded Markov chain and applying
Theorem \ref{thm:(Theorem-3-of}. For systems with large $k$, solving
all the embedded Markov chains could be tedious. Fast approximations
for $p_{i}$'s are provided in \cite{takenaka1989characteristics}.

So far we have characterized the probability $p_{i}$ for Inequality
(\ref{eq:5}), which we can use to obtain the PAoI upper bound for
each queue. In fact, the upper bounds that we provide in Inequality
(\ref{eq:5}) are decent approximations of PAoI for queues. We will
show it numerically in Section \ref{sec:Numerical-Study}. 

It was found by Costa et al \cite{costa2016age} that for single-queue
systems such as M/M/1/1, M/M/1/2, and M/M/1/2{*}, increasing the arrival
rate can reduce PAoI continuously. However, it is not the case in
our model with multiple queues. We find that by increasing the arrival
rate of a certain queue, its own PAoI will decrease. However, PAoI
for the other queues will not necessarily decrease. We will show the
detail numerically in Section \ref{sec:Numerical-Study}. Besides,
we have the following theorem discussing the scenario when the arrival
rate of a certain queue becomes large. 
\begin{thm}
\label{thm:infty bound}For $1\leq i\leq k$, if $\lambda_{i}\rightarrow\infty$,
then $\boldsymbol{E}[A_{j}]\rightarrow\infty$ for $j>i$, and $\boldsymbol{E}[A_{j}]$
will be bounded for $j\leq i$. 
\end{thm}
\begin{IEEEproof}
We first show that as $\lambda_{i}\rightarrow\infty$, then $\pi_{j}(0,...,0)\rightarrow0$
for $j\geq i$. To show this, we know that in the subsystem $\mathcal{S}_{j}$,
the first element of the transition matrix for the embedded Markov
chain is given by $a_{0}=\int_{0}^{\infty}e^{-(\sum_{l=1}^{j}\lambda_{l})x}dF(x)=\int_{0}^{\infty}(\sum_{l=1}^{j}\lambda_{l})e^{-(\sum_{l=1}^{j}\lambda_{l})x}F(x)dx.$
Since $F(x)\leq1$ for any $x\in[0,\infty)$, by dominated convergence
theorem, we have $\lim_{\lambda_{i}\rightarrow\infty}\int_{0}^{\infty}e^{-(\sum_{l=1}^{j}\lambda_{l})x}dF(x)dx=0.$
By the result from \cite{takenaka1989analysis} that $\pi_{j}(0,...,0)=\sum_{l=1}^{j}a_{0}\pi_{j}(0,...,\stackrel{l}{1},...,0)+a_{0}\pi_{j}(0,...,0)$,
we have $\pi_{j}(0,...,0)\rightarrow0$ for $j\geq i$. From Theorem
\ref{thm:(Theorem-3-of} we have $p_{j}=1-\frac{\pi_{j-1}(0,...0)-\pi_{j}(0,...,0)}{\frac{\lambda_{j}}{\mu}+\frac{\lambda_{j}}{\sum_{l=1}^{k}\lambda_{l}}\pi_{k}(0,...,0)}-\frac{\pi_{k}(0,...,0)}{\frac{\sum_{l=1}^{k}\lambda_{l}}{\mu}+\pi_{k}(0,...,0)}\rightarrow1$
for $j>i$. We then have $\boldsymbol{E}[A_{j}]\rightarrow\infty$
for $j>i$. 

For $j<i$, we have $p_{j}\leq1-\frac{\pi_{j-1}(0,...0)-\pi_{j}(0,...,0)}{\frac{\lambda_{j}}{\mu}+\frac{\lambda_{j}}{\sum_{l=1}^{k}\lambda_{l}}}.$
From the fact that $\pi_{j}(0,...,0)$ for $j<i$ will not be affected
by $\lambda_{i}$ and $\pi_{j}(0,...,0)<\pi_{j-1}(0,...,0)$ if $\lambda_{j}\neq0$
(see Theorem 1 of \cite{takenaka1989analysis}), we then have $p_{j}\leq1-\frac{\pi_{j-1}(0,...0)-\pi_{j}(0,...,0)}{\frac{\lambda_{j}}{\mu}}<1$
as $\lambda_{i}\rightarrow\infty$. Thus $\boldsymbol{E}[A_{j}]$
is bounded by Inequality (\ref{eq:5}). 

For $j=i$, we have 
\begin{eqnarray*}
 &  & \frac{p_{i}}{\lambda_{i}(1-p_{i})}\\
 & = & \frac{1-\frac{\pi_{i-1}(0,...0)-\pi_{i}(0,...,0)}{\frac{\lambda_{i}}{\mu}+\frac{\lambda_{i}}{\sum_{j=1}^{k}\lambda_{j}}\pi_{k}(0,...,0)}-\frac{\pi_{k}(0,...,0)}{\frac{\sum_{j=1}^{k}\lambda_{j}}{\mu}+\pi_{k}(0,...,0)}}{\lambda_{i}(\frac{\pi_{i-1}(0,...0)-\pi_{i}(0,...,0)}{\frac{\lambda_{i}}{\mu}+\frac{\lambda_{i}}{\sum_{j=1}^{k}\lambda_{j}}\pi_{k}(0,...,0)}+\frac{\pi_{k}(0,...,0)}{\frac{\sum_{j=1}^{k}\lambda_{j}}{\mu}+\pi_{k}(0,...,0)})}\\
 & = & \frac{1-\frac{\pi_{i-1}(0,...0)-\pi_{i}(0,...,0)}{\frac{\lambda_{i}}{\mu}+\frac{\lambda_{i}}{\sum_{j=1}^{k}\lambda_{j}}\pi_{k}(0,...,0)}-\frac{\pi_{k}(0,...,0)}{\frac{\sum_{j=1}^{k}\lambda_{j}}{\mu}+\pi_{k}(0,...,0)}}{\frac{\pi_{i-1}(0,...0)-\pi_{i}(0,...,0)}{\frac{1}{\mu}+\frac{1}{\sum_{j=1}^{k}\lambda_{j}}\pi_{k}(0,...,0)}-\frac{\lambda_{i}\pi_{k}(0,...,0)}{\frac{\sum_{j=1}^{k}\lambda_{j}}{\mu}+\pi_{k}(0,...,0)}}\\
 & \leq & \frac{1}{\frac{\pi_{i-1}(0,...0)-\pi_{i}(0,...,0)}{\frac{1}{\mu}+\frac{1}{\sum_{j=1}^{k}\lambda_{j}}\pi_{k}(0,...,0)}-\frac{\pi_{k}(0,...,0)}{\frac{\sum_{j=1}^{k}\lambda_{j}}{\mu\lambda_{i}}+\frac{\pi_{k}(0,...,0)}{\lambda_{i}}}}.
\end{eqnarray*}
 As $\lambda_{i}\rightarrow\infty$, $\frac{p_{i}}{\lambda_{i}(1-p_{i})}\leq\frac{1}{\mu\pi_{i-1}(0,...,0)}.$
By Inequality (\ref{eq:5}) we prove the theorem.
\end{IEEEproof}
Theorem \ref{thm:infty bound} shows that if we increase the arrival
rate of a queue, the PAoI of queues with lower priorities will be
greatly increased, while PAoI of queues with higher priorities will
remain bounded. Like we discussed in Section \ref{sec:Introduction},
some data sources may have information more important or time-sensitive
than other data sources. Theorem \ref{thm:infty bound} implies that
if higher priorities are given to data sources which are more important,
the PAoI of these important data sources can always be guaranteed
at a bounded level. It also implies that if we have queues with traffic
intensity significantly greater than the others, it is better to give
high priorities to those queues with low traffic intensities to guarantee
that all queues have a relatively low PAoI. 

\subsection{Exact Analysis for M/G/1/1+$\mathbf{\sum1}$ Type Queues\label{subsec:Exact-Analysis-for-1}}

Notice that in system M/G/1/1+$\mathbf{\sum1^{*}}$, we keep the most
recent arrival in the buffer. If we instead, keep the first arrival
in the buffer and reject the future arrivals before the buffer becomes
empty, the system becomes M/G/1/1+$\mathbf{\sum}1$. In the single
queue case, the system is denoted as M/G/1/2 which was analyzed in
\cite{costa2016age}. From Equation (\ref{eq:2}), the PAoI of source
$i$ in M/G/1/1+$\mathbf{\sum1}$ is given by 
\begin{eqnarray}
\boldsymbol{E}[A_{i}] & = & \boldsymbol{E}[P_{i}]+\boldsymbol{2E}[W_{i}]+\boldsymbol{E}[I_{i}].\label{eq:4}
\end{eqnarray}

From Subsection \ref{subsec:Bounds-and-Approximation}, we know the
method of calculating the probability $p_{i}$. By Equation (\ref{eq:3})
we can obtain the exact PAoI of M/G/1/1+$\mathbf{\sum1}$ system.
In addition, the following theorem states that the PAoI of each queue
in M/G/1/1+$\mathbf{\sum1}$ system is always larger than or equal
to that of an M/G/1/1+$\mathbf{\sum1}^{*}$ system.
\begin{thm}
\label{thm:PAoI-of-each}PAoI of each queue in M/G/1/1+$\mathbf{\sum1}$
is always greater than or equal to that of M/G/1/1+$\mathbf{\sum1}^{*}$
system, if both systems have the same parameters.
\end{thm}
\begin{IEEEproof}
From Equation (\ref{eq:2}) and (\ref{eq:4}) we only need to show
that $\boldsymbol{E}[G_{i}]\leq\boldsymbol{E}[W_{i}].$ Since $\boldsymbol{E}[G_{i}]=\frac{1}{\lambda_{i}}(1-\boldsymbol{E}[e^{-\lambda_{i}W_{i}}])$
and from the fact that $e^{-\lambda_{i}W_{i}}\geq1-\lambda_{i}W_{i}$,
we have $\boldsymbol{E}[G]\leq\frac{1}{\lambda_{i}}(1-(1-\lambda_{i}W_{i}))=\boldsymbol{E}[W_{i}]$.
Hence proved.
\end{IEEEproof}
Corollary \ref{thm:PAoI-of-each} reveals the fact that if service
times are i.i.d. for each source, then for the system with buffer
size one at each queue, it is always beneficial to keep the most recent
arrival in the buffer for reducing PAoI.

\section{Infinite Buffer Size \label{sec:Infinite-Buffer-Size}}

Although dropping redundant packets such as in system M/G/1+$\sum1^{*}$
can potentially reduce the system traffic, it is not clear if keeping
all the packets can result in a smaller PAoI. More importantly, for
some applications, dropping packets is not an option when the entire
data stream must been obtained for performing offline diagnostics
(also see \cite{kosta2019queue}). In such a scenario, processing
all the generated packets is necessary and for that, buffer size of
each queue needs to be large enough. In this section we discuss a
model in which buffer size of each queue is infinite. This model has
been discussed in \cite{huang2015optimizing,yates2019age}, however
they do not consider queues with priorities. In this section, since
there could be multiple packets waiting in each queue, it is necessary
to ascertain the order of service within a queue. We consider FCFS
and LCFS service disciplines separately when the server serves packets
from the same queue. Still, the server starts serving packets from
high priority queues when it becomes available. Throughout this section,
we assume that $\sum_{j=1}^{k}\frac{\lambda_{j}}{\mu_{j}}<1$ so that
the system is stable.

\subsection{Exact Analysis for M/G/1 Type Queues with FCFS\label{subsec:Exact-Analysis-for-2}}

We first discuss the model in which each queue is served according
to FCFS discipline. From the definition of PAoI we know that when
processing is complete for the $j^{th}$ arrival from queue $i$,
the random variable corresponding to PAoI is equal to $A_{ij}=C_{ij}-r_{i(j-1)}=(C_{ij}-r_{ij})+(r_{ij}-r_{i(j-1)}).$
Since $C_{ij}-r_{ij}$ is the sojourn time of packet $j$ and $r_{ij}-r_{i(j-1)}$
is the inter-arrival time between packet $j-1$ and $j$, the PAoI
for queue $i$ can be written as $\boldsymbol{E}[A_{i}]=\boldsymbol{E}[P_{i}]+\boldsymbol{E}[W_{i}]+\boldsymbol{E}[I_{i}],$
where $\boldsymbol{E}[W_{i}]$ is the expected waiting time in queue
and $\boldsymbol{E}[I_{i}]=\frac{1}{\lambda_{i}}$ is the expected
inter-arrival time. From \cite{gautam2012analysis} we have the exact
expression of $\boldsymbol{E}[W_{i}]$ for M/G/1 type queues with
priority, thus the PAoI of queue $i$ is given by:

\begin{align}
\boldsymbol{E}[A_{i}] & =\boldsymbol{E}[W_{i}]+\boldsymbol{E}[I_{i}]+\boldsymbol{E}[P_{i}]\nonumber \\
 & =\frac{\frac{1}{2}\sum_{j=1}^{k}\lambda_{j}\boldsymbol{E}[P_{j}^{2}]}{(1-\sum_{j=1}^{i}\frac{\lambda_{j}}{\mu_{j}})(1-\sum_{j=1}^{i-1}\frac{\lambda_{j}}{\mu_{j}})}+\frac{1}{\lambda_{i}}+\frac{1}{\mu_{i}}.\label{eq:6}
\end{align}

Interestingly, from the expression of $\boldsymbol{E}[W_{i}]$, we
find that the packets from higher priority queues always have shorter
expected waiting times compared with those from low priority queues.
However, Equation (\ref{eq:6}) shows that higher priority queues
do not always have smaller PAoI because $\frac{1}{\lambda_{i}}$ and
$\frac{1}{\mu_{i}}$ also contribute to PAoI. Another interesting
point from Equation (\ref{eq:6}) is that by increasing arrival rate
$\lambda_{i}$ we can reduce the PAoI for queue $i$ but greatly enlarge
the PAoI for queues with priority lower than $i$. We will also show
this result numerically in Section \ref{sec:Numerical-Study}.

Bedewy et al \cite{bedewy2018age} considered scheduling policies
to minimize the average PAoI across queues, i.e., $\frac{1}{k}\sum_{i=1}^{k}\boldsymbol{E}[A_{i}]$.
If we also consider the same objective and ask the design question
of how to minimize the average PAoI across queues by assigning queue
priorities, the answer is assigning high priorities to queues with
low $\rho_{i}=\frac{\lambda_{i}}{\mu_{i}}$, as we see in Theorem
\ref{thm:Low-priority}. 
\begin{thm}
\label{thm:Low-priority}If the queue priorities satisfy $\rho_{1}\leq\rho_{2}\leq...\leq\rho_{k},$
then the average PAoI across queues given by this priority order is
the smallest among all the priority orders.
\end{thm}
\begin{IEEEproof}
Since
\begin{align}
 & \frac{1}{k}\sum_{i=1}^{k}\boldsymbol{E}[A_{i}]\nonumber \\
= & \frac{1}{k}\sum_{i=1}^{k}\left[\frac{\frac{1}{2}\sum_{j=1}^{k}\lambda_{j}\boldsymbol{E}[P_{j}^{2}]}{(1-\sum_{j=1}^{i}\rho_{j})(1-\sum_{j=1}^{i-1}\rho_{j})}+\frac{1}{\lambda_{i}}+\frac{1}{\mu_{i}}\right],\label{eq:7}
\end{align}
changing priority orders only affects the denominator of the first
term in Equation (\ref{eq:7}). So minimizing the average PAoI across
queues is equivalent to minimizing $\sum_{i=1}^{k}\frac{1}{(1-\sum_{j=1}^{i}\rho_{j})(1-\sum_{j=1}^{i-1}\rho_{j})}$.
If $(\rho_{1},\rho_{2},...,\rho_{k})$ is the optimal priority order
with $\rho_{i}\geq\rho_{i+m}$, by switching the order of $\rho_{i}$
and $\rho_{i+m}$ we have a new priority order $(\rho_{1}^{*},\rho_{2}^{*},...,\rho_{k}^{*})$
with $\rho_{i}^{*}=\rho_{i+m}$, $\rho_{i+m}^{*}=\rho_{i}$ and $\rho_{j}^{*}=\rho_{j}$
for $j\in\{1,...,k\}\backslash\{i,i+m\}$. Then we have $\sum_{l=1}^{j}\rho_{l}=\sum_{l=1}^{j}\rho_{l}^{*}$
for $j<i$, $\sum_{l=1}^{j}\rho_{l}\geq\sum_{l=1}^{j}\rho_{l}^{*}$
for $i\leq j<i+m$, and $\sum_{l=1}^{j}\rho_{l}=\sum_{l=1}^{j}\rho_{l}^{*}$
for $j\geq i+m$. Thus we have 
\begin{eqnarray*}
 &  & \sum_{l=1}^{k}\frac{1}{(1-\sum_{j=1}^{l}\rho_{j})(1-\sum_{j=1}^{l-1}\rho_{j})}\\
 &  & -\sum_{l=1}^{k}\frac{1}{(1-\sum_{j=1}^{l}\rho_{j}^{*})(1-\sum_{j=1}^{l-1}\rho_{j}^{*})}\\
 & = & \sum_{l=i}^{i+m}\bigg[\frac{1}{(1-\sum_{j=1}^{l}\rho_{j})(1-\sum_{j=1}^{l-1}\rho_{j})}\\
 &  & -\frac{1}{(1-\sum_{j=1}^{l}\rho_{j}^{*})(1-\sum_{j=1}^{l-1}\rho_{j}^{*})}\bigg]\\
 & \geq & 0,
\end{eqnarray*}

which contradicts to the assumption that $(\rho_{1},\rho_{2},...,\rho_{k})$
is the optimal priority order. Therefore we prove the theorem.
\end{IEEEproof}
From Theorem \ref{thm:Low-priority} we see that for M/G/1 type queues
with FCFS discipline, it is always better to give high priorities
to queues with small traffic intensities when the objective is to
minimize the average PAoI across all queues. In fact, this observation
is also true for M/G/1/1+$\sum1^{*}$ queues that we discussed in
Section \ref{sec:Buffer-Size-One}. The intuitive reason for this
is if we do the opposite, i.e., allowing high traffic queues to have
high priority, the server would be busy serving high traffic intensity
queues and barely have chance to serve low priority queues. Packets
from low priority queues therefore would suffer a large waiting time.
We will show this numerically in Section \ref{sec:Numerical-Study}. 

\subsection{Exact Analysis for M/G/1 Type Queues with LCFS\label{subsec:M/G/1-LCFS}}

In this subsection, we derive the PAoI for priority queues with LCFS
within each queue. The server chooses the highest priority queue when
it becomes available, and from each queue it serves the last arrived
packet first. We now introduce a new service scheme which has the
same PAoI as LCFS. We first divide each queue into two virtual parts:
initial buffer and main queue. The initial buffer can hold only one
packet. Whenever a new arrival occurs, we send this new arrival into
the initial buffer, and move the stale packet (if there is one) from
the initial buffer into the main queue. When the server starts serving
a queue, it serves the packet from initial buffer first if it is not
empty, then serves packets from main queue in an arbitrary order with
the understanding that service times are i.i.d.. A demonstrative graph
of the idea of initial buffer and main queue is shown in Figure \ref{fig:Initial-Buffer-and-1}.
For queue 1 in Figure \ref{fig:Initial-Buffer-and-1}, the initial
buffer is empty since its most recent arrival has been processed.
The initial buffer of queue 2 is full. When the server switches to
queue 2, the packet in initial buffer 2 will be processed first.

This service scheme has the same PAoI as LCFS since under both schemes,
only the freshest packets result in age peaks. We can thus characterize
the PAoI of each queue by focusing on the initial buffer status. The
state of the initial buffer is either 0 or 1, and each period length
of state $0$ (when the initial buffer is empty), is equal to the
inter-arrival time $I_{i}$ between packets. We abuse our notation
by letting the time period of state $1$ (when the initial buffer
is full) be $W_{i}$, which we call the busy period of the initial
buffer. Using the analysis in Section \ref{subsec:Exact-Analysis-for},
the PAoI for queue $i$ is given as $\boldsymbol{E}[A_{i}]=\boldsymbol{E}[P_{i}]+\boldsymbol{E}[W_{i}]+\boldsymbol{E}[I_{i}]+\boldsymbol{E}[G_{i}],$
where $\boldsymbol{E}[P_{i}]=\frac{1}{\mu_{i}}$ is the expected service
time, $\boldsymbol{E}[W_{i}]$ is the expected length of period when
initial buffer is full, $\boldsymbol{E}[I_{i}]=\frac{1}{\lambda_{i}}$
is the expected inter-arrival time, and $\boldsymbol{E}[G_{i}]=\boldsymbol{E}[\frac{1}{\lambda_{i}}(1-e^{-\lambda_{i}W_{i}})]$
is the expected waiting time of the most recently arrived packet before
the buffer becomes empty, which is given in Lemma \ref{lem:1}. 

\begin{figure}[h]
\begin{centering}
\includegraphics[scale=0.35]{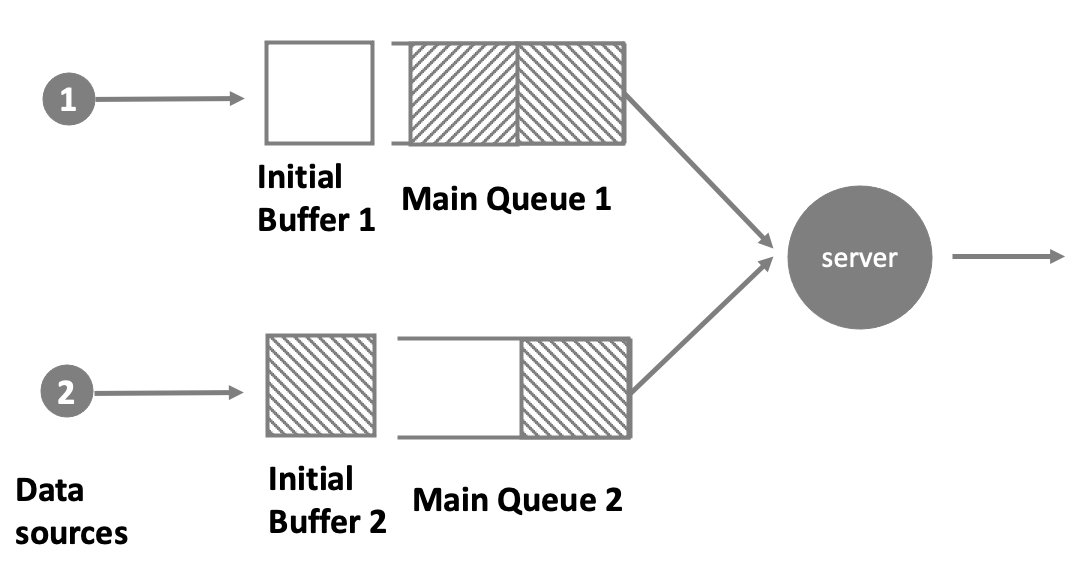}\caption{Initial Buffer and Main Queue for LCFS\label{fig:Initial-Buffer-and-1}}
\par\end{centering}
\end{figure}

Similar to what we did in Section \ref{subsec:Bounds-and-Approximation},
since in the system of M/G/1 type queues with LCFS, there is one initial
buffer in each queue, we can thus derive the PAoI for queue $i$ as:

\begin{eqnarray}
\boldsymbol{E}[A_{i}] & = & \boldsymbol{E}[P_{i}]+\boldsymbol{E}[W_{i}]+\boldsymbol{E}[I_{i}]+\boldsymbol{E}[G_{i}]\nonumber \\
 & = & \frac{1}{\mu_{i}}+\frac{p_{i}}{\lambda_{i}(1-p_{i})}+\frac{2}{\lambda_{i}}-\frac{1}{\lambda_{i}}\boldsymbol{E}[e^{-\lambda_{i}W_{i}}],\label{eq:9}
\end{eqnarray}

for all $i\in\{1,...,k\},$ where we abuse our notation here by letting
$p_{i}$ be the steady state probability that the initial buffer is
full (notice that in Section \ref{sec:Buffer-Size-One} we used it
as the probability that the buffer is full). Now we introduce the
method of finding $p_{i}$'s by providing Lemma \ref{lem:Single-Queue }
for the case of $k=1$ first.
\begin{lem}
\label{lem:Single-Queue }For the M/G/1 queue with LCFS of $k=1$,
the probability that the initial buffer is full is given by $p_{1}=\frac{\lambda_{1}}{\mu_{1}}-1+\psi_{1}(\lambda_{1})$,
where $\psi_{1}(u)$ is the LST of the service time.
\end{lem}
\begin{IEEEproof}
From Figure \ref{fig:Initial-Buffer-and} we find that a busy period
of the initial buffer always occurs when the server is processing,
and ends when the service is complete. Suppose the processing time
of a packet is $P_{1}=u$, then by Campbell's Theorem \cite{kulkarni2016modeling},
the busy period of the initial buffer $\hat{W}$ during time $P_{1}=u$
is given by 

\begin{eqnarray*}
\boldsymbol{E}[\hat{W}|P_{1}=u] & = & \sum_{m=1}^{\infty}\frac{m}{m+1}ue^{-\lambda_{1}u}\frac{(\lambda_{1}u)^{m}}{m!}\\
 & = & u-\frac{1}{\lambda_{1}}+\frac{1}{\lambda_{1}}e^{-\lambda_{1}u}.
\end{eqnarray*}

By unconditioning on $P_{1}=u$ we have $\boldsymbol{E}[\hat{W}]=\int_{0}^{\infty}(u-\frac{1}{\lambda_{1}}+\frac{1}{\lambda_{1}}e^{-\lambda_{1}u})dF_{1}(u)=\frac{1}{\mu_{1}}-\frac{1}{\lambda_{1}}+\frac{\psi_{1}(\lambda_{1})}{\lambda_{1}}$.
However, it is important to note that this $\boldsymbol{E}[\hat{W}]$
is the expected busy time of initial buffer during the processing
time of a packet. To obtain $p_{i}$, we need the following argument.
Suppose $n(t)$ packets have been served during $(0,t]$. Thus the
amount of time that the initial buffer being full during $(0,t]$
is $n(t)\boldsymbol{E}[\hat{W}]$. If the queue is stable, we have
$n(t)$ converging to $\lambda_{1}t$ as $t\rightarrow\infty.$ Therefore
\begin{eqnarray*}
p_{1}=\lim_{t\rightarrow\infty}\frac{n(t)\boldsymbol{E}[\hat{W}]}{t}=\lambda_{1}\boldsymbol{E}[\hat{W}],
\end{eqnarray*}

which is the stationary probability that the initial buffer is full.
Note that $\frac{\lambda_{1}}{\mu_{1}}-1+\psi_{1}(\lambda_{1})$ is
a legitimate probability as it always lies within $[0,1]$. To show
this, from the fact that $\psi_{1}(\lambda_{1})=\int_{0}^{\infty}e^{-\lambda x}dF_{1}(x)\leq1$,
we have $\frac{\lambda_{1}}{\mu_{1}}-1+\psi_{1}(\lambda_{1})\leq\frac{\lambda_{1}}{\mu_{1}}<1$
from stability assumption. Since $\psi_{1}(\lambda_{1})=\int_{0}^{\infty}e^{-\lambda x}dF_{1}(x)\geq\int_{0}^{\infty}(1-\lambda_{1}x)dF_{1}(x)=1-\frac{\lambda_{1}}{\mu_{1}}$,
we have $\frac{\lambda_{1}}{\mu_{1}}-1+\psi_{1}(\lambda_{1})\geq0$.
Thus $\frac{\lambda_{1}}{\mu_{1}}-1+\psi_{1}(\lambda_{1})$ is a legitimate
probability.
\end{IEEEproof}
\begin{figure}[h]
\includegraphics[scale=0.3]{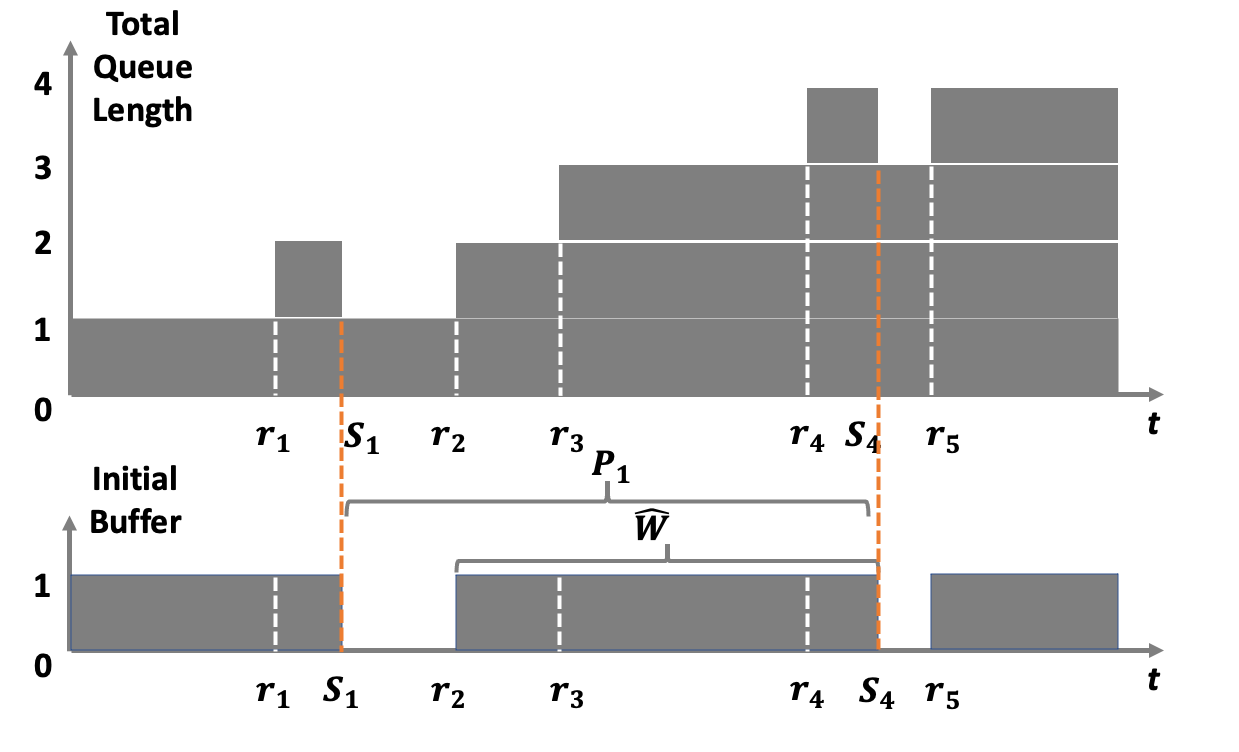}
\begin{centering}
\caption{Initial Buffer and Total Queue Length (Including Initial Buffer and
Main Queue) for LCFS \label{fig:Initial-Buffer-and}}
\par\end{centering}
\end{figure}

Now we discuss the case when $k\geq2$. Notice that for each packet
that is in service, if there is a new arrival from queue 1 occurring
during this service time, then the busy period for initial buffer
1 is from the arrival time of this new packet to the completion time
of the packet being processed. From Lemma \ref{lem:Single-Queue },
the busy period for initial buffer 1 if a type $i$ packet is being
processed when the busy period starts, is given by $\frac{1}{\mu_{i}}-\frac{1}{\lambda_{1}}+\frac{1}{\lambda_{1}}\psi_{i}(\lambda_{1}).$
Thus we have $p_{1}=\sum_{i=1}^{k}\lambda_{i}(\frac{1}{\mu_{i}}-\frac{1}{\lambda_{1}}+\frac{1}{\lambda_{1}}\psi_{i}(\lambda_{1}))$
.

To get the probability $p_{i}$ for queue $i\geq2$, we use the idea
introduced by Kella and Yechiali \cite{kella1988priorities}. We merge
the queues with priority higher than $i$ as one class $\mathcal{C}_{ai}$
and the other queues as another class $\mathcal{C}_{bi}$ by letting
$\lambda_{ai}=\sum_{j=1}^{i-1}\lambda_{j}$, $\lambda_{bi}=\sum_{j=i}^{k}\lambda_{j}$,
$\rho_{ai}=\sum_{j=1}^{i-1}\rho_{j}$ and $\rho_{bi}=\sum_{j=i}^{k}\rho_{j}$.
We also let $F_{ai}(x)=\sum_{j=1}^{i-1}\frac{\lambda_{j}}{\lambda_{ai}}F_{j}(x)$
be the service time distribution for packets from class $\mathcal{C}_{ai}$,
with mean $\boldsymbol{E}[P_{ai}]$ and $F_{bi}(x)=\sum_{j=i}^{k}\frac{\lambda_{j}}{\lambda_{bi}}F_{j}(x)$
be the service time distribution for packets from class $\mathcal{C}_{bi}$,
with mean $\boldsymbol{E}[P_{bi}].$ Notice that the busy period of
initial buffer $i$ ends only when there is no packet from class $\mathcal{C}_{ai}$.
We now classify the busy periods of server (the time period during
which the server is continuously serving packets) into two types.
One type of busy period $V_{ai}$ starts with processing a packet
from $\mathcal{C}_{ai}$, and ends when there is no packet of from
$\mathcal{C}_{ai}$ left in the system. The other type of busy period
$V_{bi}$ starts with processing a packet from $\mathcal{C}_{bi}$,
and also ends when there is no packet from $\mathcal{C}_{ai}$ left
in the system. If no packet from $\mathcal{C}_{ai}$ arrives during
processing the first packet in $V_{bi}$, then the length of $V_{bi}$
is just $P_{bi}$. If one packet from $\mathcal{C}_{ai}$ arrives
during processing the first packet in $V_{bi}$, then after the processing
the first packet, a busy period $V_{ai}$ is followed. Similar to
the analysis in \cite{kella1988priorities,conway2003theory} and what
we did in Section \ref{subsec:Exact-Analysis-for}, by conditioning
on service time of the first packet in a busy period, the LST of $V_{ai}$
and $V_{bi}$, denoted as $\tilde{V}_{ai}(s)$ and $\tilde{V}_{bi}(s)$,
are given as 
\begin{eqnarray}
\tilde{V}_{ai}(s) & = & \psi_{ai}(s+\lambda_{ai}-\lambda_{ai}\tilde{V}_{ai}(s))\label{eq:10}
\end{eqnarray}
 and 
\begin{eqnarray}
\tilde{V}_{bi}(s) & = & \psi_{bi}(s+\lambda_{ai}-\lambda_{ai}\tilde{V}_{ai}(s)),\label{eq:11}
\end{eqnarray}
where $\psi_{ai}(s)$ is the LST of $F_{ai}(x)$ and $\psi_{bi}(s)$
is the LST of $F_{bi}(x).$ By taking the derivative of $\tilde{V}_{ai}(s)$
and $\tilde{V}_{bi}(s)$ at $s=0$, the expected length of server's
busy periods can be given as
\begin{eqnarray*}
\boldsymbol{E}[V_{ai}] & = & \frac{\boldsymbol{E}[P_{ai}]}{1-\rho_{ai}},
\end{eqnarray*}
and 
\begin{eqnarray*}
\boldsymbol{E}[V_{bi}] & = & \frac{\boldsymbol{E}[P_{bi}]}{1-\rho_{ai}}.
\end{eqnarray*}

Note that busy period $V_{bi}$ always starts by serving a packet
from $\mathcal{C}_{bi}$, thus we know 

\begin{eqnarray*}
\boldsymbol{P}(system\ in\ V_{bi}) & = & \lambda_{bi}\frac{\boldsymbol{E}[P_{bi}]}{1-\rho_{ai}}.
\end{eqnarray*}

Since when the server is busy, it is either in busy period $V_{ai}$
or $V_{bi}$ , we have 
\begin{eqnarray*}
\boldsymbol{P}(system\ in\ V_{ai}) & = & \sum_{j=1}^{k}\rho_{j}-\lambda_{bi}\frac{\boldsymbol{E}[P_{bi}]}{1-\rho_{ai}}=\hat{\lambda}_{ai}\frac{\boldsymbol{E}[P_{ai}]}{1-\rho_{ai}},
\end{eqnarray*}
where 
\begin{eqnarray}
\hat{\lambda}_{ai} & = & \frac{\sum_{j=1}^{k}\rho_{j}-\lambda_{bi}\frac{\boldsymbol{E}[P_{bi}]}{1-\rho_{ai}}}{\frac{\boldsymbol{E}[P_{ai}]}{1-\rho_{ai}}}\label{eq:12}
\end{eqnarray}
is the ``arrival rate'' of busy period $V_{ai}.$ We now use $F_{ai}^{V}(x)$
and $F_{bi}^{V}(x)$ to denote the CDF of $V_{ai}$ and $V_{bi}$.
From Lemma \ref{lem:Single-Queue } we know that during busy period
$V_{ai}$, the time period of initial buffer being busy is given as
\begin{eqnarray}
\boldsymbol{E}[\hat{W}_{ai}] & = & \int_{0}^{\infty}(u-\frac{1}{\lambda_{i}}+\frac{1}{\lambda_{i}}e^{-\lambda_{i}u})dF_{ai}^{V}(u)\nonumber \\
 & = & \frac{\boldsymbol{E}[P_{ai}]}{1-\rho_{ai}}-\frac{1}{\lambda_{i}}+\frac{1}{\lambda_{i}}\tilde{V}_{ai}(\lambda_{i}).\label{eq:13}
\end{eqnarray}
Similarly, we have $\boldsymbol{E}[\hat{W}_{bi}]=\frac{\boldsymbol{E}[P_{bi}]}{1-\rho_{ai}}-\frac{1}{\lambda_{i}}+\frac{1}{\lambda_{i}}\tilde{V}_{bi}(\lambda_{i}).$
In many cases where $\tilde{V}_{ai}(\lambda_{i})$ and $\tilde{V}_{bi}(\lambda)$
cannot be solved analytically, numerical methods such as bisection
method or Newton's method (see \cite{cheney2012numerical}) can be
applied to find the roots numerically. From the same argument in Lemma
\ref{lem:Single-Queue }, we have 
\begin{eqnarray}
p_{i} & = & \hat{\lambda}_{ai}\boldsymbol{E}[\hat{W}_{ai}]+\lambda_{bi}\boldsymbol{E}[\hat{W}_{bi}].\label{eq:13_1}
\end{eqnarray}

Next we introduce the process of obtaining $\boldsymbol{E}[e^{-\lambda_{i}W_{i}}]$
for Equation (\ref{eq:9}). Notice that $\hat{W}_{ai}$ is the length
of initial buffer being full during period $V_{ai}$. Similar to the
argument in Lemma \ref{lem:Single-Queue }, we have 

\begin{eqnarray*}
 &  & \boldsymbol{E}[e^{-s\hat{W}_{ai}}|V_{ai}=t]\\
 & = & \int_{0}^{t}e^{-sx}\sum_{m=1}^{\infty}\frac{mx^{m-1}}{t^{m}}e^{-\lambda_{i}t}\frac{(\lambda_{i}t)^{m}}{m!}dx+e^{-\lambda_{i}t}\\
 & = & \int_{0}^{t}e^{-sx}e^{-\lambda_{i}t}\sum_{m=1}^{\infty}x^{m-1}\frac{(\lambda_{i})^{m}}{(m-1)!}dx+e^{-\lambda_{i}t}\\
 & = & e^{-\lambda_{i}t}\frac{\lambda_{i}}{\lambda_{i}-s}(e^{(\lambda_{i}-s)t}-1)+e^{-\lambda_{i}t}\\
 & = & \frac{\lambda_{i}}{\lambda_{i}-s}e^{-st}-\frac{s}{\lambda_{i}-s}e^{-\lambda_{i}t}.
\end{eqnarray*}

By unconditioning on \textbf{$V_{ai}=t$} we have 
\begin{eqnarray*}
\boldsymbol{E}[e^{-s\hat{W}_{ai}}] & = & \frac{\lambda_{i}}{\lambda_{i}-s}\tilde{V}_{ai}(s)-\frac{s}{\lambda_{i}-s}\tilde{V}_{ai}(\lambda_{i}),
\end{eqnarray*}
and 
\begin{eqnarray*}
\boldsymbol{E}[e^{-s\hat{W}_{bi}}] & = & \frac{\lambda_{i}}{\lambda_{i}-s}\tilde{V}_{bi}(s)-\frac{s}{\lambda_{i}-s}\tilde{V}_{bi}(\lambda_{i}).
\end{eqnarray*}
Using L'Hospital rule taking the limit $s\rightarrow\lambda_{i}$,
we have 
\begin{eqnarray}
\boldsymbol{E}[e^{-\lambda_{i}\hat{W}_{ai}}] & = & -\lambda_{i}\tilde{V}_{ai}^{'}(\lambda_{i})+\tilde{V}_{ai}(\lambda_{i}),\label{eq:14}
\end{eqnarray}
and 
\begin{eqnarray}
\boldsymbol{E}[e^{-\lambda_{i}\hat{W}_{bi}}] & = & -\lambda_{i}\tilde{V}_{bi}^{'}(\lambda_{i})+\tilde{V}_{bi}(\lambda_{i}).\label{eq:15}
\end{eqnarray}
From the formula of $\tilde{V}_{ai}(s)$ and $\tilde{V}_{bi}(s)$
given above we have 
\begin{eqnarray}
\tilde{V}_{ai}^{'}(\lambda_{i}) & = & \frac{\psi_{ai}^{'}(\lambda_{i}+\lambda_{ai}-\lambda_{ai}\tilde{V}_{ai}(\lambda_{i}))}{1+\lambda_{ai}\psi_{ai}^{'}(\lambda_{i}+\lambda_{ai}-\lambda_{ai}\tilde{V}_{ai}(\lambda_{i}))},\label{eq:16}
\end{eqnarray}
and
\begin{eqnarray}
 &  & \tilde{V}_{bi}^{'}(\lambda_{i})\nonumber \\
 & = & \psi_{bi}^{'}(\lambda_{i}+\lambda_{ai}-\lambda_{ai}\tilde{V}_{ai}(\lambda_{i}))(1-\lambda_{ai}\tilde{V}_{ai}^{'}(\lambda_{i})).\label{eq:17}
\end{eqnarray}

Notice that $W_{i}$ is the busy period of the initial buffer for
each age peak, and only ($1-p_{i}$) portion of arrivals in queue
$i$ incur age peaks, so the ``arrival rate'' for $W_{i}$ is $\lambda_{i}(1-p_{i})$.
Since $\hat{W}_{i}$ is the busy period of initial buffer during each
$V_{ai}$ and $V_{bi}$ with arrival rate $\hat{\lambda}_{ai}$ and
$\lambda_{bi}$ respectively, from the fact that $\lambda_{bi}=\sum_{j=i}^{k}\lambda_{j}\geq\lambda_{i}(1-p_{i})$,
we have the following relationship

\begin{eqnarray*}
 &  & \hat{\lambda}_{ai}\boldsymbol{E}[e^{-\lambda_{i}\hat{W}_{ai}}]+\lambda_{bi}\boldsymbol{E}[e^{-\lambda_{i}\hat{W}_{bi}}]\\
 & = & \lambda_{i}(1-p_{i})\boldsymbol{E}[e^{-\lambda_{i}W_{i}}]+(\hat{\lambda}_{ai}+\lambda_{bi}-\lambda_{i}(1-p_{i})).
\end{eqnarray*}

Therefore, 
\begin{eqnarray}
\boldsymbol{E}[e^{-\lambda_{i}W_{i}}] & = & \frac{1}{\lambda_{i}(1-p_{i})}\bigg[\hat{\lambda}_{ai}\boldsymbol{E}[e^{-\lambda_{i}\hat{W}_{ai}}]+\lambda_{bi}\boldsymbol{E}[e^{-\lambda_{i}\hat{W}_{bi}}]\nonumber \\
 &  & -(\hat{\lambda}_{ai}+\lambda_{bi}-\lambda_{i}(1-p_{i}))\bigg].\label{eq:18}
\end{eqnarray}

A closed-form formula of PAoI in M/G/1 type queues with LCFS is then
given in the following theorem.
\begin{thm}
The PAoI of queue $i$ in M/G/1 system with LCFS is given by $\boldsymbol{E}[A_{i}]=\frac{1}{\mu_{i}}+\frac{1}{\lambda_{i}(1-p_{i})}-\frac{1}{\lambda_{i}^{2}(1-p_{i})}\bigg[\hat{\lambda}_{ai}(-\lambda_{i}\tilde{V}_{ai}^{'}(\lambda_{i})+\tilde{V}_{ai}(\lambda_{i})-1)+\lambda_{bi}(-\lambda_{i}\tilde{V}_{bi}^{'}(\lambda_{i})+\tilde{V}_{bi}(\lambda_{i})-1)\bigg],$
where $\tilde{V}_{ai}(s)=\psi_{ai}(s+\lambda_{ai}-\lambda_{ai}\tilde{V}_{ai}(s))$,
$\tilde{V}_{bi}(s)=\psi_{bi}(s+\lambda_{ai}-\lambda_{ai}\tilde{V}_{ai}(s))$,
$p_{i}=\hat{\lambda}_{ai}\left[\frac{\boldsymbol{E}[P_{ai}]}{1-\rho_{ai}}-\frac{1}{\lambda_{i}}+\frac{1}{\lambda_{i}}\tilde{V}_{ai}(\lambda_{i})\right]+\lambda_{bi}\left[\frac{\boldsymbol{E}[P_{bi}]}{1-\rho_{ai}}-\frac{1}{\lambda_{i}}+\frac{1}{\lambda_{i}}\tilde{V}_{bi}(\lambda_{i})\right]$,
and $\hat{\lambda}_{ai}=\frac{\sum_{j=1}^{k}\rho_{j}-\lambda_{bi}\frac{\boldsymbol{E}[P_{bi}]}{1-\rho_{ai}}}{\frac{\boldsymbol{E}[P_{ai}]}{1-\rho_{ai}}}$.
\end{thm}
\begin{IEEEproof}
To obtain the exact PAoI for queue $i$ of M/G/1 system with LCFS,
we first solve $\tilde{V}_{ai}(\lambda_{i})$ and $\tilde{V}_{bi}(\lambda_{i})$
using Equation (\ref{eq:10}) and (\ref{eq:11}). And all the required
components for computing $\boldsymbol{E}[e^{-\lambda_{i}W_{i}}]$
in Equation (\ref{eq:18}) can be obtained from Equations (\ref{eq:12})
- (\ref{eq:17}). After that the PAoI of queue $i$ can be obtained
using Equation (\ref{eq:9}).
\end{IEEEproof}
Note here that this approach of calculating PAoI of queue $i$ in
M/G/1 type system LCFS can still be applied even when the number of
queues $k$ is large. The numerical test of this approach will be
provided in Section \ref{sec:Numerical-Study}.

\subsection{Discussion of the Single Queue Systems\label{subsec:Discussion-of-the}}

The LST of PAoI for single queue with LCFS was provided in \cite{inoue2019general},
however its expression is quite involved (see Equation (99) in \cite{inoue2019general}).
Here we use our approach introduced in Subsection \ref{subsec:M/G/1-LCFS}
to provide a concise expression for PAoI of M/G/1/LCFS queue in the
following corollaries. Since we only have one queue here, for simplicity
of the notation, in this subsection we remove the subscript of each
variable. 
\begin{cor}
\label{cor:The-PAoI-of}The PAoI of M/G/1/LCFS is given by $\boldsymbol{E}[A]=\frac{1}{\mu}+\frac{1}{\lambda}+\frac{\frac{1}{\mu}+\psi^{'}(\lambda)}{2-\frac{\lambda}{\mu}-\psi(\lambda)}$,
where $\psi(\lambda)$ is the LST of service time.
\end{cor}
\begin{IEEEproof}
Since we know $\boldsymbol{E}[e^{-\lambda\hat{W}}]=-\lambda\psi^{'}(\lambda)+\psi(\lambda)$
from Equation (\ref{eq:15}) and $\boldsymbol{E}[e^{-\lambda W}]=\frac{\boldsymbol{E}[e^{-\lambda\hat{W}}]-p}{1-p}$
from Equation (\ref{eq:18}), also from Lemma \ref{lem:Single-Queue }
we know $p=\frac{\lambda}{\mu}-1+\psi(\lambda)$, we have our corollary
proved.
\end{IEEEproof}
\begin{cor}
The PAoI of M/M/1/LCFS is given by $\boldsymbol{E}[A]=\frac{1}{\lambda}+\frac{1}{\mu}+\frac{2\lambda\mu+\lambda^{2}}{(\lambda+\mu)(\mu^{2}+\lambda\mu-\lambda^{2})}$.
\end{cor}
It is shown in \cite{costa2016age} that M/M/1/$2^{*}$ system has
a smaller PAoI than M/M/1/2 and M/M/1/1 systems. Interestingly, we
find that the PAoI in M/G/1/$2^{*}$ system is no greater than that
in M/G/1/LCFS system, as shown in the following theorem.
\begin{thm}
\textcolor{black}{The PAoI in }\textup{\textcolor{black}{M/G/1/}}\textcolor{black}{$2^{*}$
system is no greater than that in M/G/1/LCFS system, for any $\rho=\frac{\lambda}{\mu}<1$.}
\end{thm}
\begin{IEEEproof}
It is shown in \cite{xu2020vacations} that the PAoI in M/G/1/$2^{*}$
system is given by $\boldsymbol{E}[A^{M/G/1/2^{*}}]=\frac{2}{\mu}+\frac{1}{\lambda}+\psi^{'}(\lambda).$
Using the result of Corollary \ref{cor:The-PAoI-of}, we have 
\begin{eqnarray*}
 &  & \boldsymbol{E}[A^{M/G/1/LCFS}]-\boldsymbol{E}[A^{M/G/1/2^{*}}]\\
 & = & \frac{1}{\mu}+\frac{1}{\lambda}+\frac{\frac{1}{\mu}+\psi^{'}(\lambda)}{2-\frac{\lambda}{\mu}-\psi(\lambda)}-\left(\frac{2}{\mu}+\frac{1}{\lambda}+\psi^{'}(\lambda)\right)\\
 & = & \left(\frac{1}{\mu}+\psi^{'}(\lambda)\right)\frac{\frac{\lambda}{\mu}+\psi(\lambda)-1}{2-\frac{\lambda}{\mu}-\psi(\lambda)}.
\end{eqnarray*}

We first have $\frac{1}{\mu}+\psi^{'}(\lambda)=\frac{1}{\mu}-\boldsymbol{E}[Pe^{-\lambda P}]=\boldsymbol{E}[P]-\boldsymbol{E}[Pe^{-\lambda P}]\geq0$.
Then from the facts that 
\begin{eqnarray*}
\frac{\lambda}{\mu}+\psi(\lambda) & = & \lambda\boldsymbol{E}[P]+\boldsymbol{E}[e^{-\lambda P}]\\
 & \geq & \lambda\boldsymbol{E}[P]+\boldsymbol{E}[1-\lambda P]=1
\end{eqnarray*}
 and $\frac{\lambda}{\mu}+\psi(\lambda)<2$, we have $\boldsymbol{E}[A^{M/G/1/LCFS}]\geq\boldsymbol{E}[A^{M/G/1/2^{*}}].$
\end{IEEEproof}
We now show that LCFS is actually not the optimal service discipline
for minimizing PAoI among all the non-preemptive work-conserving service
disciplines. To do this, we simply consider the case of exponential
service. The PAoI of FCFS in this case is given by $\boldsymbol{E}[A^{FCFS}]=\frac{1}{\lambda}+\frac{1}{\mu}+\frac{\lambda}{\mu(\mu-\lambda)}$.
We then have 
\begin{eqnarray*}
 &  & \boldsymbol{E}[A^{FCFS}]-\boldsymbol{E}[A^{LCFS}]\\
 & = & \frac{\lambda}{\mu(\mu-\lambda)}-\frac{2\lambda\mu+\lambda^{2}}{(\lambda+\mu)(\mu^{2}+\lambda\mu-\lambda^{2})}\\
 & = & \frac{-\lambda^{4}+\lambda^{3}\mu+3\lambda^{2}\mu^{2}-\lambda\mu^{3}}{\mu(\mu-\lambda)(\lambda+\mu)(\mu^{2}+\lambda\mu-\lambda^{2})}.
\end{eqnarray*}

If we let $\mu=1$ in the formula above, we have $\boldsymbol{E}[A^{FCFS}]-\boldsymbol{E}[A^{LCFS}]=\frac{-\lambda^{4}+\lambda^{3}+3\lambda^{2}-\lambda}{(1-\lambda)(\lambda+1)(1+\lambda-\lambda^{2})}$.
By numerically solving it we know that $\boldsymbol{E}[A^{FCFS}]\leq\boldsymbol{E}[A^{LCFS}]$
when $0\leq\lambda<\lambda^{*}=0.3111,$ and $\boldsymbol{E}[A^{FCFS}]>\boldsymbol{E}[A^{LCFS}]$
when $1>\lambda>\lambda^{*}$, which is also shown in Figure \ref{fig:FCFS-for-M/M/1}(a).
Similarly, in Section \ref{sec:Numerical-Study}, we will show that
LCFS is not the optimal service discipline for the multi-queue case
either, when the objective is to minimize PAoI of each queue. 

This result contradicts to the conclusion in \cite{bedewy2019age}
since it was correctly proven in Theorem 3 of \cite{bedewy2019age},
that for an arbitrary sample path (arrival and service times realizations),
LCFS always results in a smaller age process $\Delta(t)$ for $t\in[0,\infty)$
than FCFS. The authors of \cite{bedewy2019age} then concluded that
PAoI under LCFS is smaller than that under FCFS, which is not accurate.
The reason is that the average age peaks (up to time $T$) is not
a non-decreasing functional (that is defined in \cite{bedewy2019age})
of the age process $\Delta(t)$ during $t\in[0,T)$. Figure \ref{fig:PAoI-under-FCFS}
shows the age functions of FCFS and LCFS under the same sample path,
where LCFS always has a smaller age $\Delta(t)$ than FCFS, but FCFS
can have smaller average age peaks than LCFS. In fact, PAoI is an
expected value conditioning on those packets that cause age peaks
(they are also called informative packets \cite{talak2018can,chen2016age}).
Under FCFS, all the data packets arriving into the system are informative,
but the informative packets under LCFS are only a subset of all the
data packets. As shown in Figure \ref{fig:PAoI-under-FCFS}, the LCFS
age function may not have peaks at all the time instances the FCFS
age function has peaks. This explains why FCFS can sometimes have
smaller PAoI than LCFS. However, the advantage of applying FCFS may
only come from the definition of metric PAoI. We will discuss more
about the metric PAoI in Subsection \ref{subsec:Discussion-on-Metric}. 

\begin{figure}
\subfloat[PAoI under FCFS. The first 6 packets result in 6 age peaks]{\includegraphics[scale=0.26]{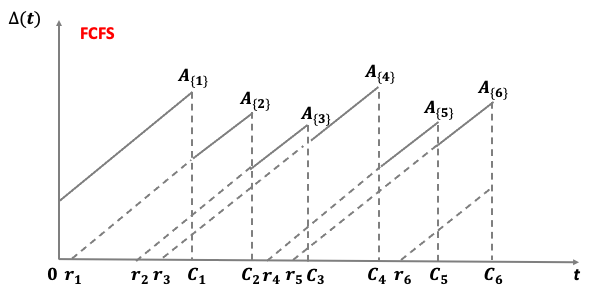}

}

\subfloat[PAoI under LCFS. The first 6 packets result in 4 age peaks ]{\includegraphics[scale=0.26]{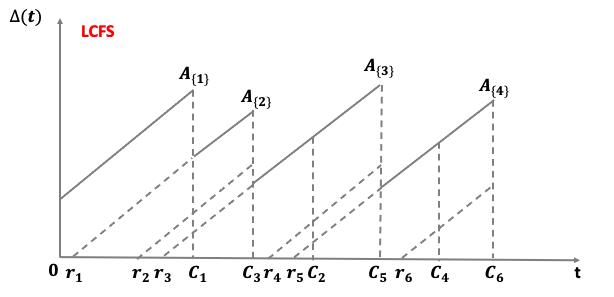}

}

\caption{PAoI under FCFS and LCFS \label{fig:PAoI-under-FCFS}}

\end{figure}

In fact, FCFS can also sometimes have a smaller PAoI than M/M/1/$2^{*}$
system, as shown in Figure \ref{fig:FCFS-for-M/M/1}(b). This is also
due the special property of the metric PAoI, since in M/M/1/$2^{*}$
system, not all the packets eventually result in age peaks. In Section
\ref{sec:Numerical-Study} we will show that for the multi-queue case,
having buffer with size one can sometimes result a larger PAoI than
FCFS.

\begin{figure}
\subfloat[FCFS versus LCFS]{\includegraphics[scale=0.5]{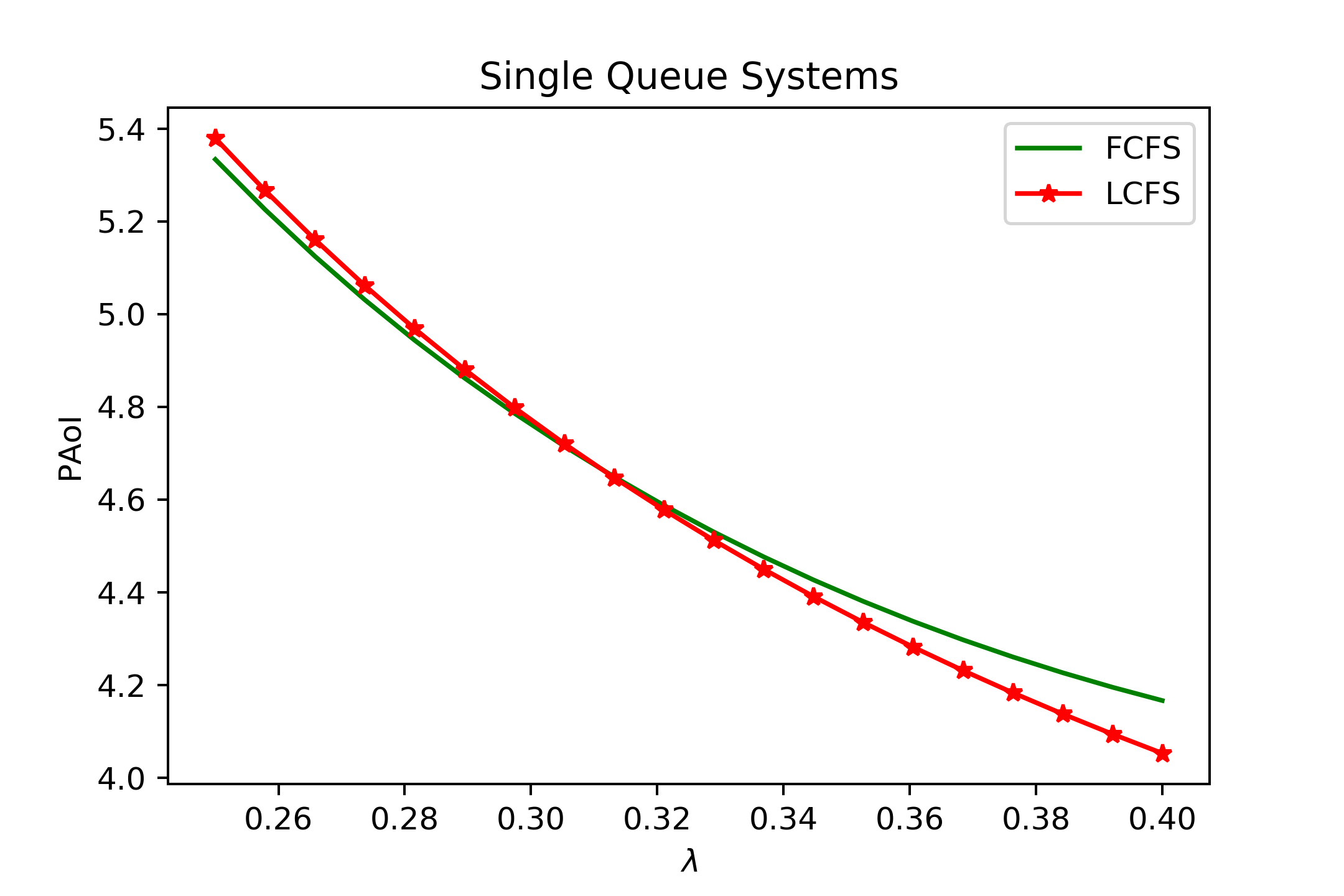}}

\subfloat[FCFS versus $M/M/1/2^{*}$]{\includegraphics[scale=0.5]{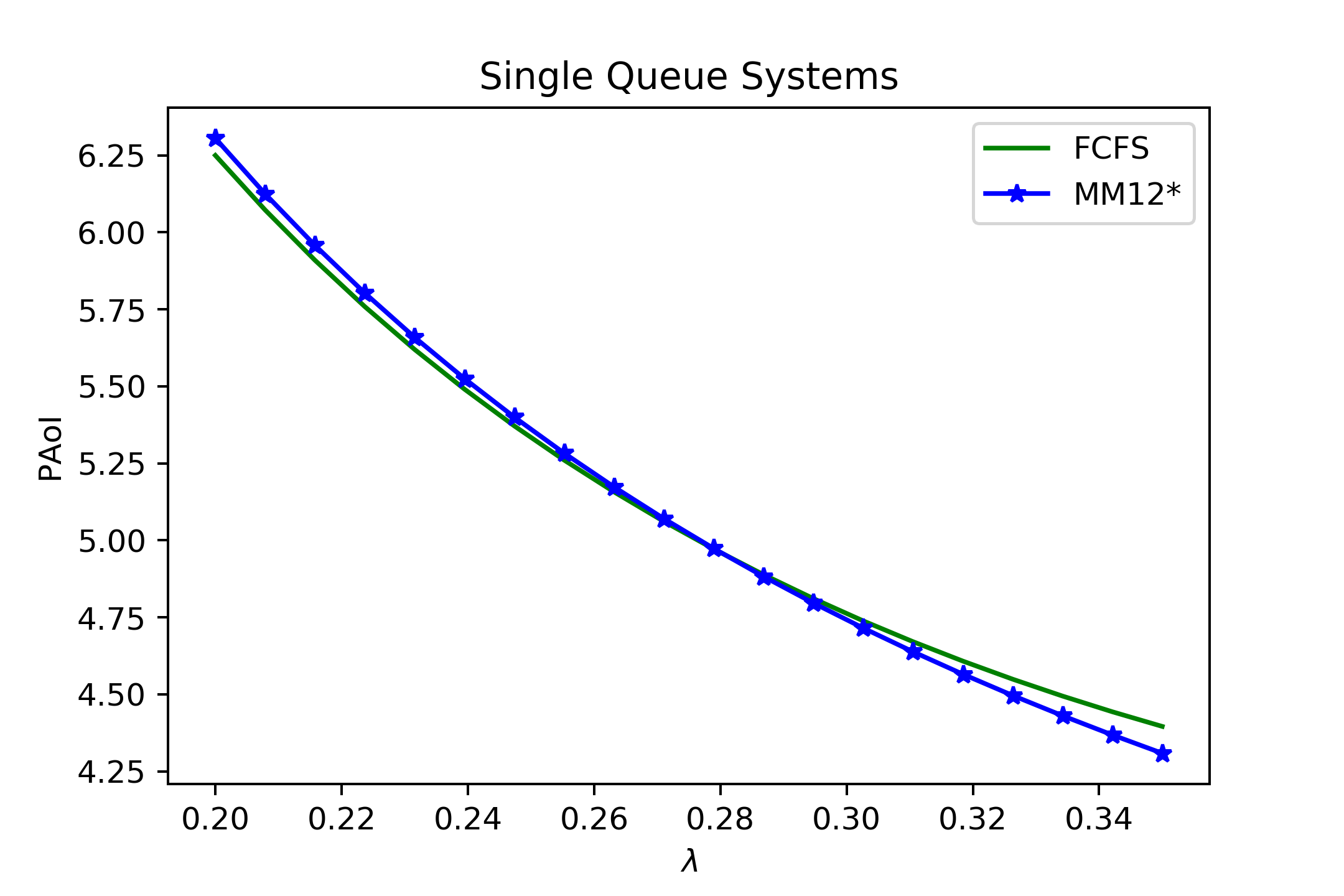}

}\caption{FCFS for M/M/1 Queue with $\mu=1$ \label{fig:FCFS-for-M/M/1}}
\end{figure}

\subsection{\label{subsec:Discussion-on-Metric}Discussion on the Metric PAoI }

We show in Subsection \ref{subsec:Discussion-of-the} that FCFS can
sometimes have a smaller PAoI than LCFS or the single buffer system.
However, any advantage of FCFS may only come from the special definition
of PAoI. FCFS may not have real advantage over LCFS, since LCFS always
results in a smaller age than FCFS for an arbitrary sample path all
the time (as correctly proven in \cite{bedewy2019age}). Considering
PAoI as the objective metric may lead us to the conclusion of LCFS
being non-optimal, which, to some extent, implies that PAoI may not
be a perfect metric to measure the information freshness.

In fact, the peak age (not its expected value PAoI) is worth of studying
in many senses. First of all, it can be used to characterize age violation
penalties (see \cite{costa2016age,ostman2019peak,devassy2019reliable}),
especially when the interest is in the worst case (see \cite{chiariotti2020peak}).
Secondly, peak age can be used to derive other age-related metrics
such as AoI (see \cite{xu2020vacations,inoue2019general}). The metric
PAoI however, as the expected value of peak age, is not able to fully
characterize the properties of peak age. The maximal peak age (see
\cite{he2016optimal}), peak age distribution (see \cite{devassy2019reliable,chiariotti2020peak}),
or the Laplace-Stieltjes Transform (LST) of peak age (see \cite{kosta2020cost,xu2020vacations,inoue2019general})
may also be interesting metrics to study.

On the other hand, PAoI is also a useful alternative metric to AoI,
especially when AoI is difficult to obtain analytically. In many systems,
PAoI has a simple form and some classic queueing results can be easily
applied to derive PAoI (see \cite{costa2016age,chen2016age,huang2015optimizing,kosta2019age,xu2020vacations,kosta2017age}).
In some cases, PAoI can also be used as a bound or an approximation
for AoI (see \cite{tripathi2019age,huang2015optimizing}). PAoI and
AoI also perform similarly when changing the traffic rate in many
systems (see \cite{costa2016age,xu2020vacations}). As a metric to
measure information freshness, a small PAoI means that the average
maximal age stays at a low level, which indicates that the system
receives frequent update in the long run. While our analysis reveals
the limitation of PAoI in some specific settings, more properties
of PAoI are worth exploring in the future. 

\section{Numerical Study \label{sec:Numerical-Study}}

In this section we will firstly use a numerical study to verify the
exact solutions for M/M/1+$\sum1^{*}$ that we provided in Section
\ref{subsec:Exact-Analysis-for}, and then test the bounds of M/G/1+$\sum1^{*}$
which we provided in Section \ref{subsec:Bounds-and-Approximation}.
We will then verify the exact solution for M/G/1 type queues with
LCFS. Besides, we will compare the performance of different service
disciplines, and develop our insights based on the numerical studies. 

We begin our discussion by comparing simulation results with exact
solutions for M/M/1+$\sum1^{*}$ system with $k=2$. The comparison
is done by changing one parameter from $\lambda_{1},\ \lambda_{2},\ \mu_{1}$
and $\mu_{2}$ while keeping the others fixed. The results are shown
in Figure \ref{fig:M/M/1-Queues-with}. From plots in Figure \ref{fig:M/M/1-Queues-with}
we can see that the simulation results match the exact solutions that
we provide in Section \ref{subsec:Exact-Analysis-for}, thus verifying
our results. Figure \ref{fig:M/M/1-Queues-with}(a) shows that when
we increase the arrival rate for the priority 1 queue, its PAoI is
drastically decreased, while PAoI for queue 2 increases constantly.
Figure \ref{fig:M/M/1-Queues-with}(b) shows that if we increase the
arrival rate of queue 2, its PAoI will decrease dramatically, while
PAoI of queue 1 increases slowly. Figure \ref{fig:M/M/1-Queues-with}(c)
and (d) show that when service rate increases, PAoI for both queues
are decreased. Interestingly, we find that when queue 1 has a low
service rate, PAoI for both queues will be large, while PAoI of queue
1 is not significantly affected by the service rate change of queue
2. We then test how the average PAoI across queues (i.e., $\frac{1}{k}\sum_{i=1}^{k}\boldsymbol{E}[A_{i}]$)
is affected by parameters, which we show in Figure \ref{fig:Average-PAoI-of}.
From Figure \ref{fig:Average-PAoI-of}(a) we see that by increasing
the service rate of either queue, the average PAoI across queues will
be reduced, and increasing the service rate of queue 1 makes this
reduction more significant. Figure \ref{fig:Average-PAoI-of}(b) shows
that by increasing the arrival rate of queue 2, the average PAoI across
queues is decreased. This is because the PAoI for queue 1 is not sensitive
to the arrival rate of queue 2, as we also show in Theorem \ref{thm:infty bound}.
However, when we increase the arrival rate of queue 1, the average
PAoI will decrease drastically at the beginning, and increase afterwards.
This is because the PAoI of queue 2 increases constantly when we increase
$\lambda_{1}$, which we also see from Figure \ref{fig:M/M/1-Queues-with}(a).
Note that although we only discuss the optimization problem of minimizing
average PAoI across queues here, since we have the exact solution
for PAoI, we could also formulate and solve optimization problems
such as minimizing average weighted PAoI (similar to \cite{talak2019optimizing})
and minimizing the maximum PAoI (similar to \cite{huang2015optimizing}).

\begin{figure}[tph]
\subfloat[$\mu_{1}=\frac{1}{10},\lambda_{2}=\frac{1}{10},\mu_{2}=\frac{1}{10}$]{\includegraphics[scale=0.5]{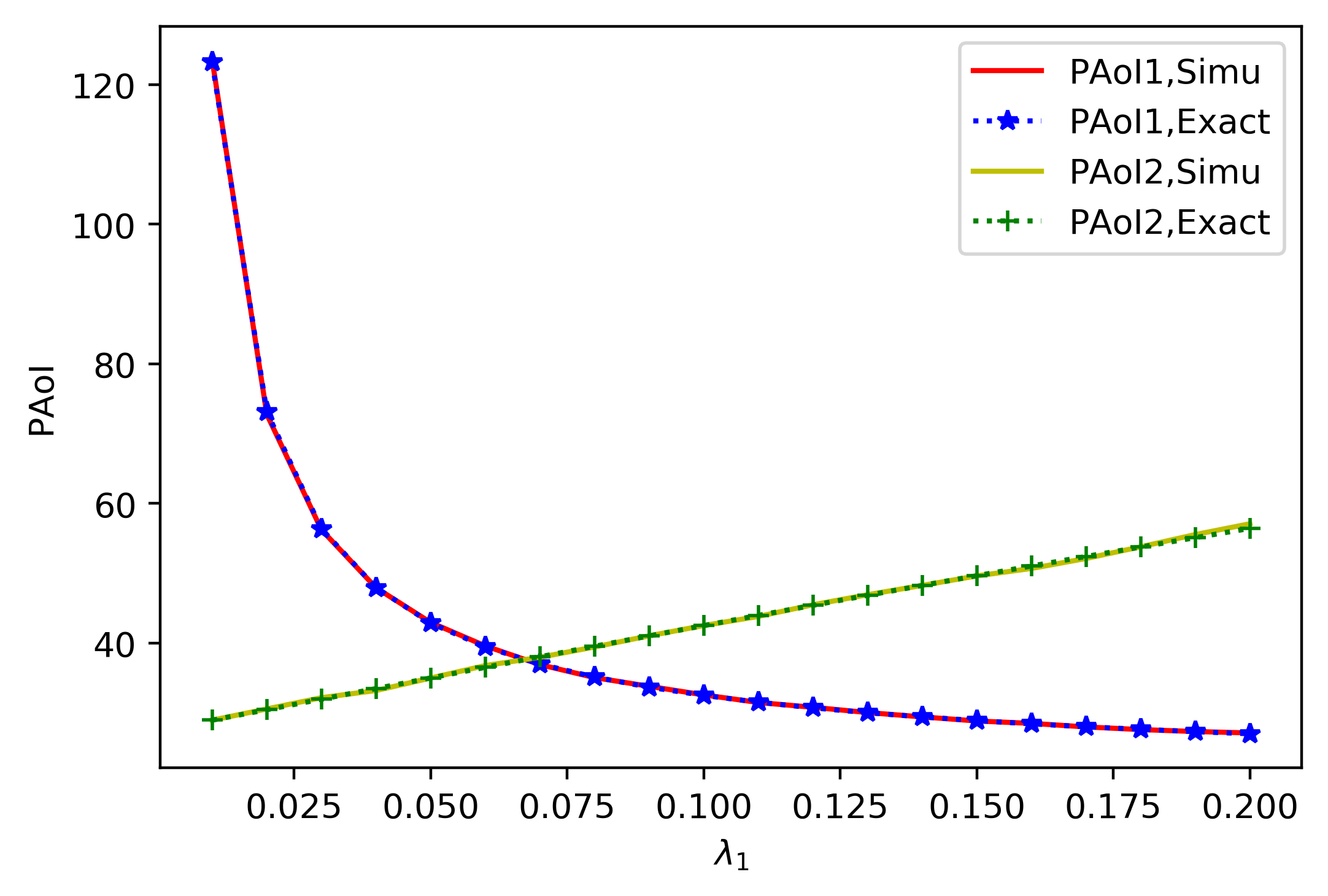}

}

\subfloat[$\lambda_{1}=\frac{1}{10},\mu_{1}=\frac{1}{10},\mu_{2}=\frac{1}{10}$]{\includegraphics[scale=0.5]{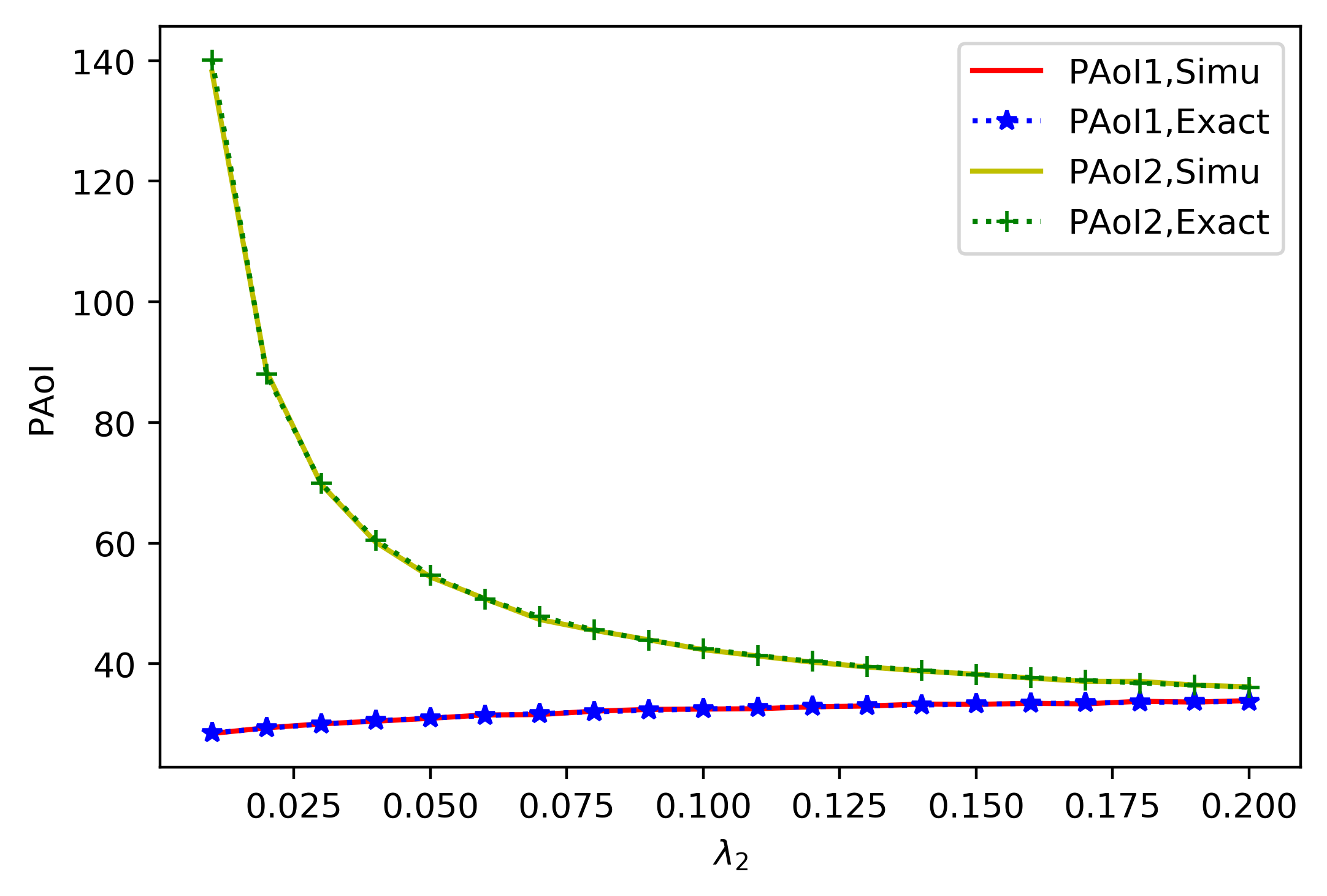}

}

\subfloat[$\lambda_{1}=\frac{1}{10},\lambda_{2}=\frac{1}{10},\mu_{2}=\frac{1}{10}$]{\includegraphics[scale=0.5]{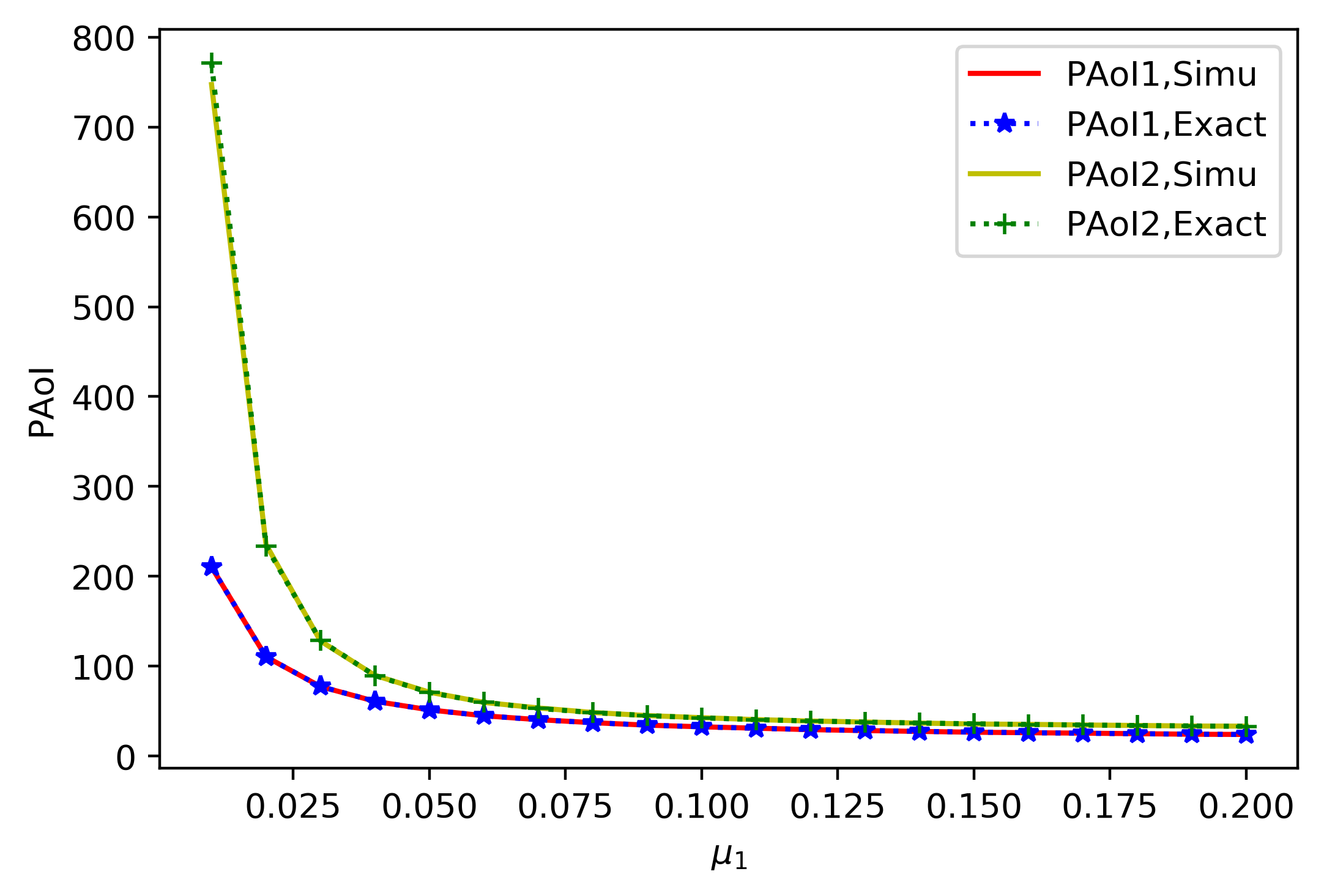}

}

\subfloat[$\lambda_{1}=\frac{1}{10},\mu_{1}=\frac{1}{10},\lambda_{2}=\frac{1}{10}$]{\includegraphics[scale=0.5]{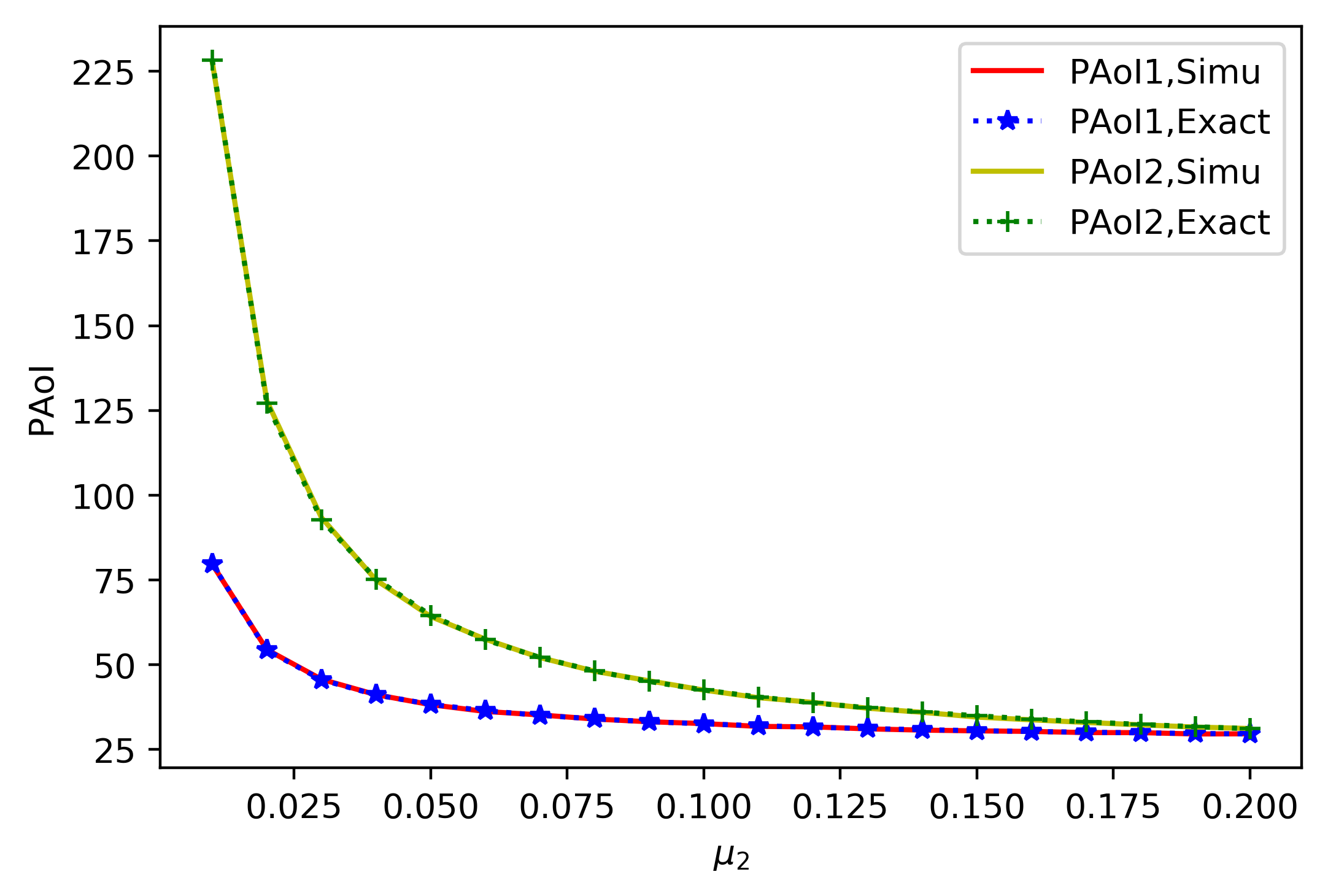}

}\caption{M/M/1+$\sum1^{*}$ Type Queues with Buffer Size One \label{fig:M/M/1-Queues-with}}
\end{figure}

\begin{figure}[h]
\subfloat[$\lambda_{1}=\frac{1}{10},\lambda_{2}=\frac{1}{10}$]{\includegraphics[scale=0.5]{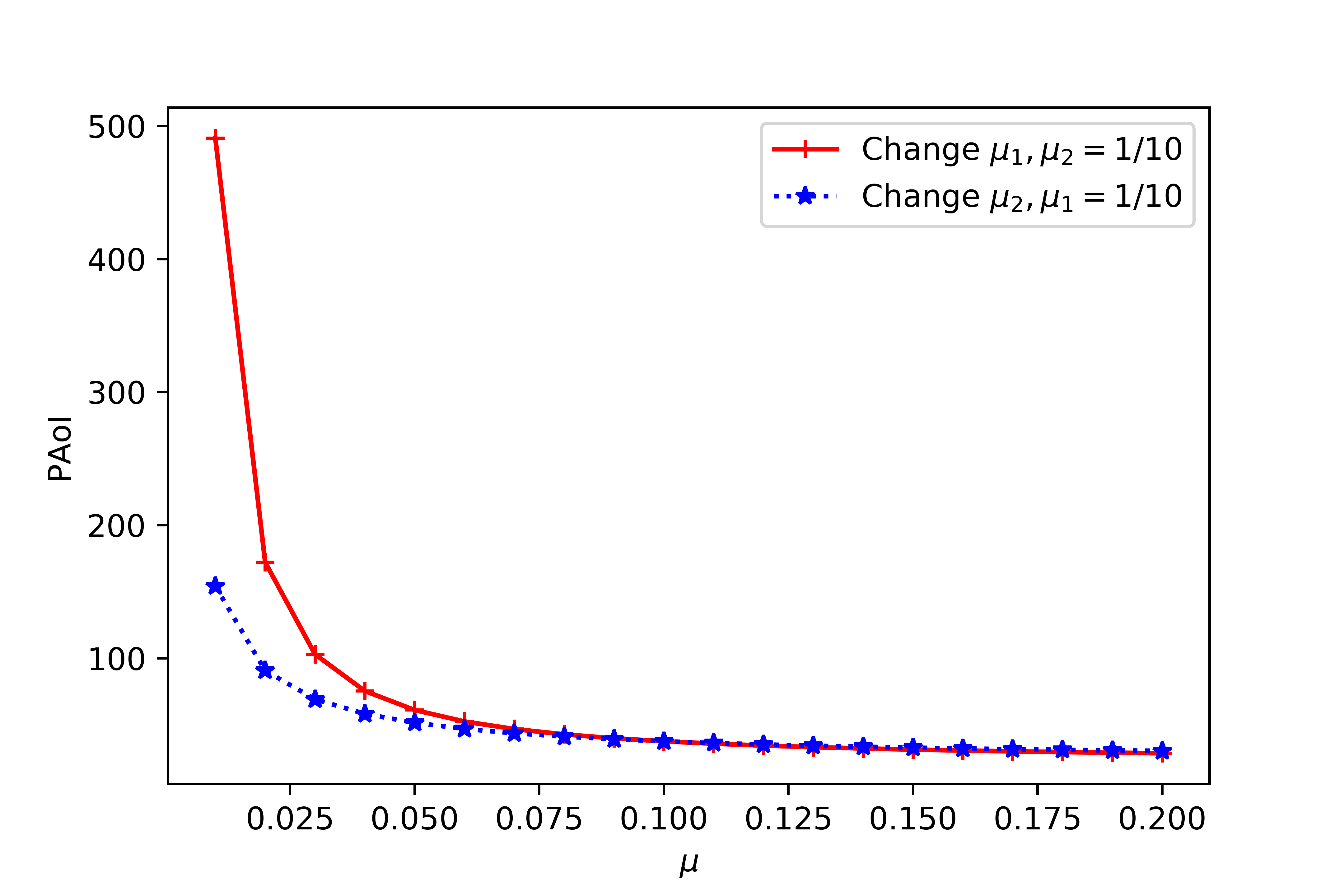}

}

\subfloat[$\mu_{1}=\frac{1}{10},\mu_{2}=\frac{1}{10}$]{\includegraphics[scale=0.5]{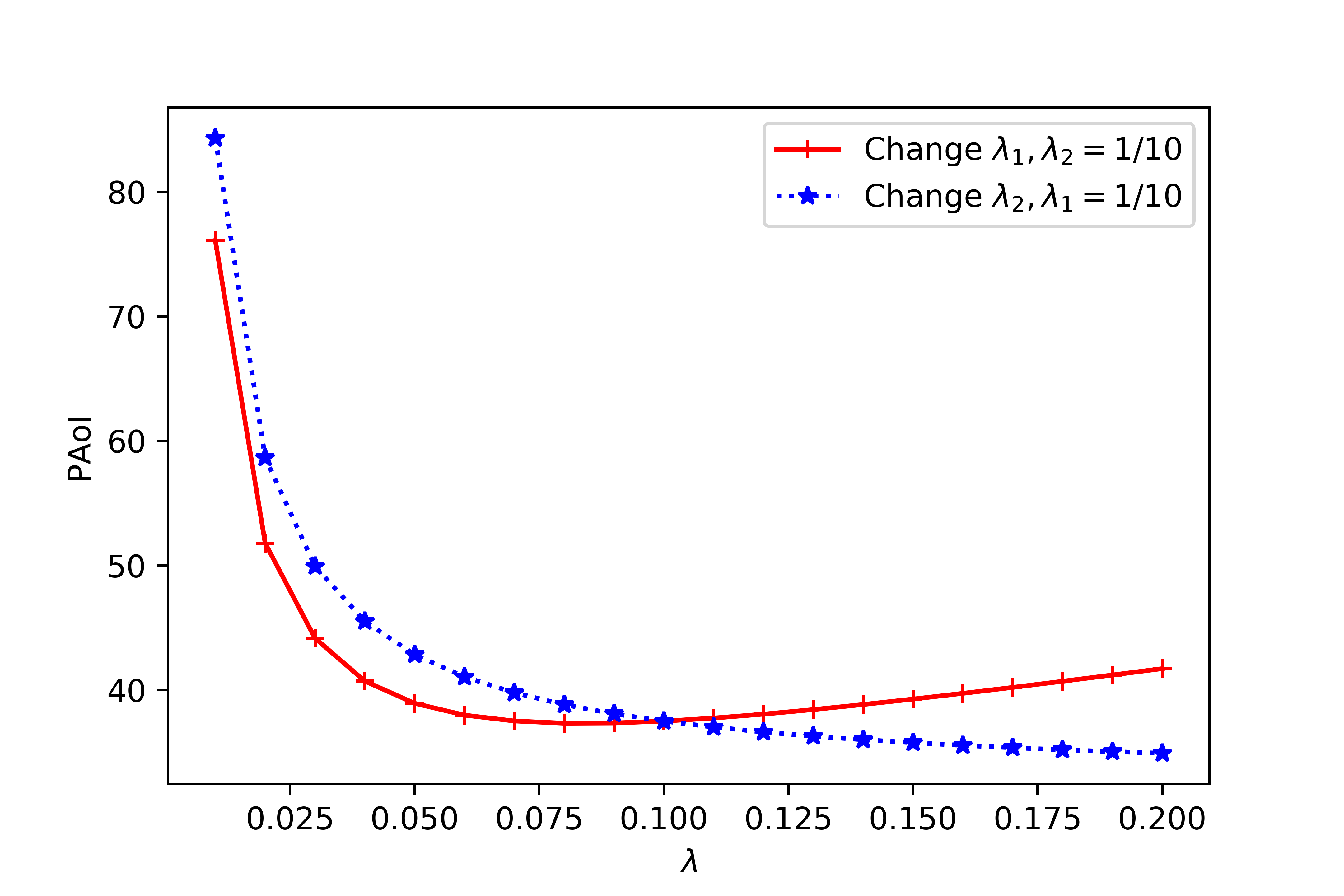}

}\caption{Average PAoI of M/M/1+$\sum1^{*}$ Queues with Buffer Size One\label{fig:Average-PAoI-of}}

\end{figure}

Next we consider queues with general service times. The bounds for
M/G/1+$\sum1^{*}$ type queues with buffer size one and $k=3$ are
shown in Figure \ref{fig:Bounds-for-M/G/1}, where the bounds are
provided by Inequality (\ref{eq:5}). We test the bounds by letting
service time follow exponential, uniform and gamma distributions.
Note that in Figure \ref{fig:Bounds-for-M/G/1} we provide the bounds
for exponential service case too, although we have the exact solution
for PAoI when service times are exponential. We find from Figure \ref{fig:Bounds-for-M/G/1}
that Inequality (\ref{eq:5}) serves as a decent approximation for
the actual PAoI since the bounds and simulation curves for all queues
are close. The three service distributions in Figure \ref{fig:Bounds-for-M/G/1}
have the same mean but the LST of these distributions vary from each
other. From our discussion in Section \ref{sec:Buffer-Size-One} we
find that the probability $p_{i}$ is related to the LST of service
time. Therefore, different service time distributions result in different
probability $p_{i},$ and further result in different PAoI. 

Then we consider queues with infinite buffer size. The exact PAoI
and simulation results for M/G/1 type queues with LCFS are shown in
Figure \ref{fig:Bounds-for-M/G/1-1}. We also test the cases for exponential,
uniform and gamma distributed service times. In Figure \ref{fig:Bounds-for-M/G/1-1}
we see that the exact PAoI that we provide in Subsection \ref{subsec:M/G/1-LCFS}
match the simulation results. We also find that in M/G/1 type queues
with LCFS, by increasing the arrival rate of queue 1, PAoI of queue
1 is significantly reduced, and PAoI for lower priority queues is
increased at the same time. We do not present the numerical test for
M/G/1 queues with FCFS here, as its analysis is exact and also straightforward. 

\begin{figure}[h]
\subfloat[$\lambda_{2}=\lambda_{3}=\frac{1}{30}$, $P_{1},P_{2},P_{3}\sim exp(\frac{1}{10})$]{\includegraphics[scale=0.5]{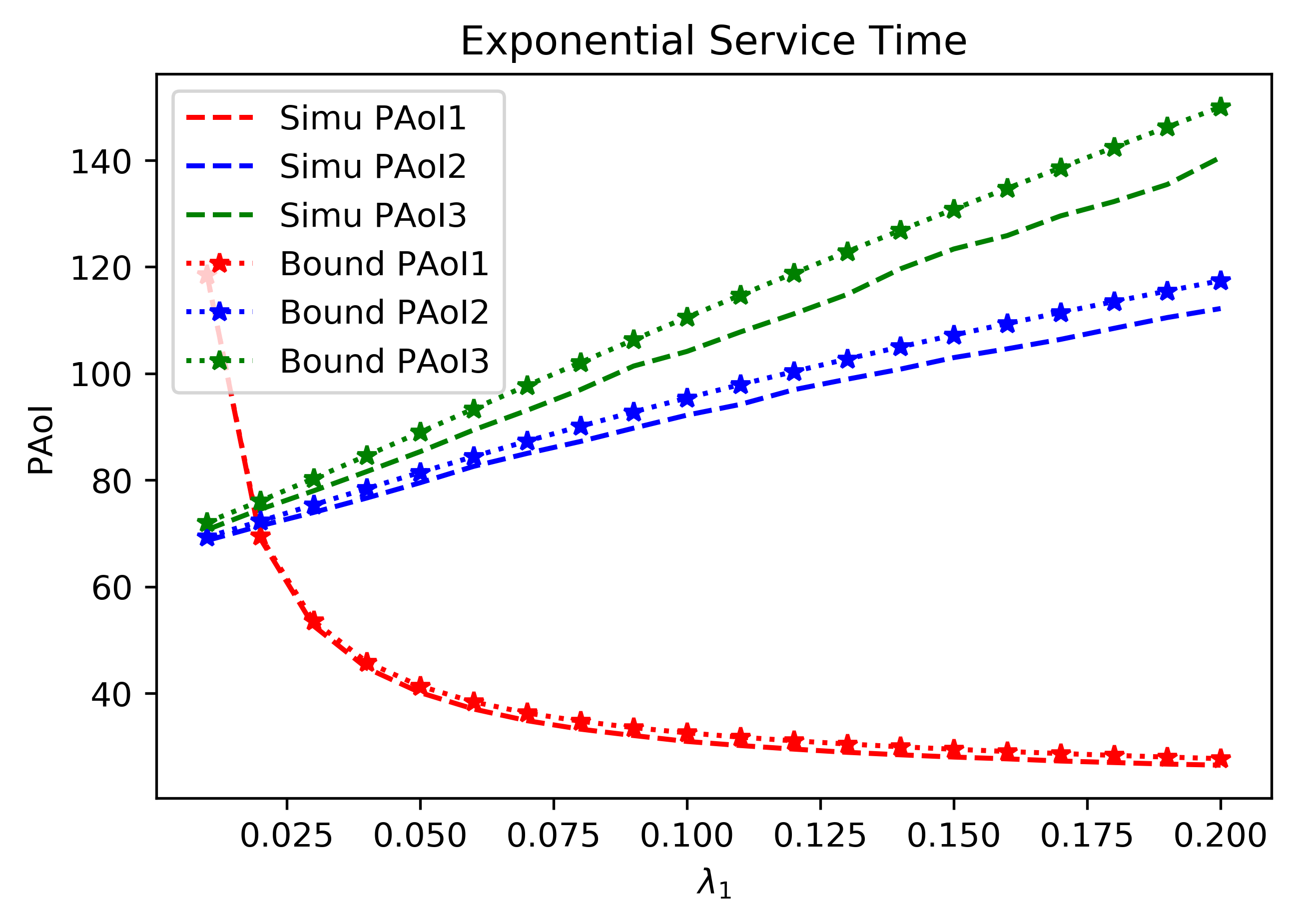}}

\subfloat[$\lambda_{2}=\lambda_{3}=\frac{1}{30},$ $P_{1},P_{2},P_{3}\sim Unif(0,20)$]{\includegraphics[scale=0.5]{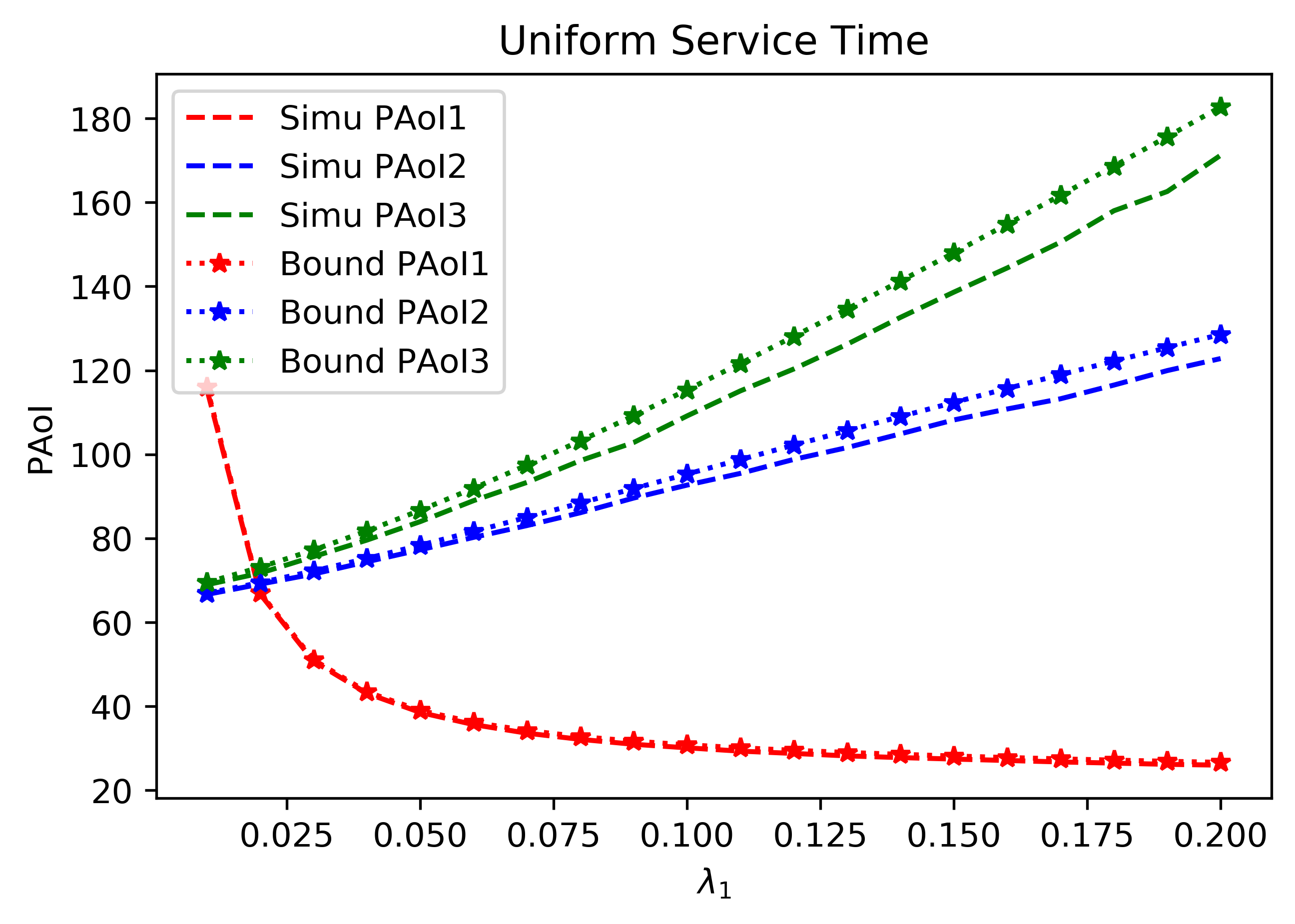}}

\subfloat[$\lambda_{2}=\lambda_{3}=\frac{1}{30},$ $P_{1},P_{2},P_{3}\sim Gamma(10,1)$]{\includegraphics[scale=0.5]{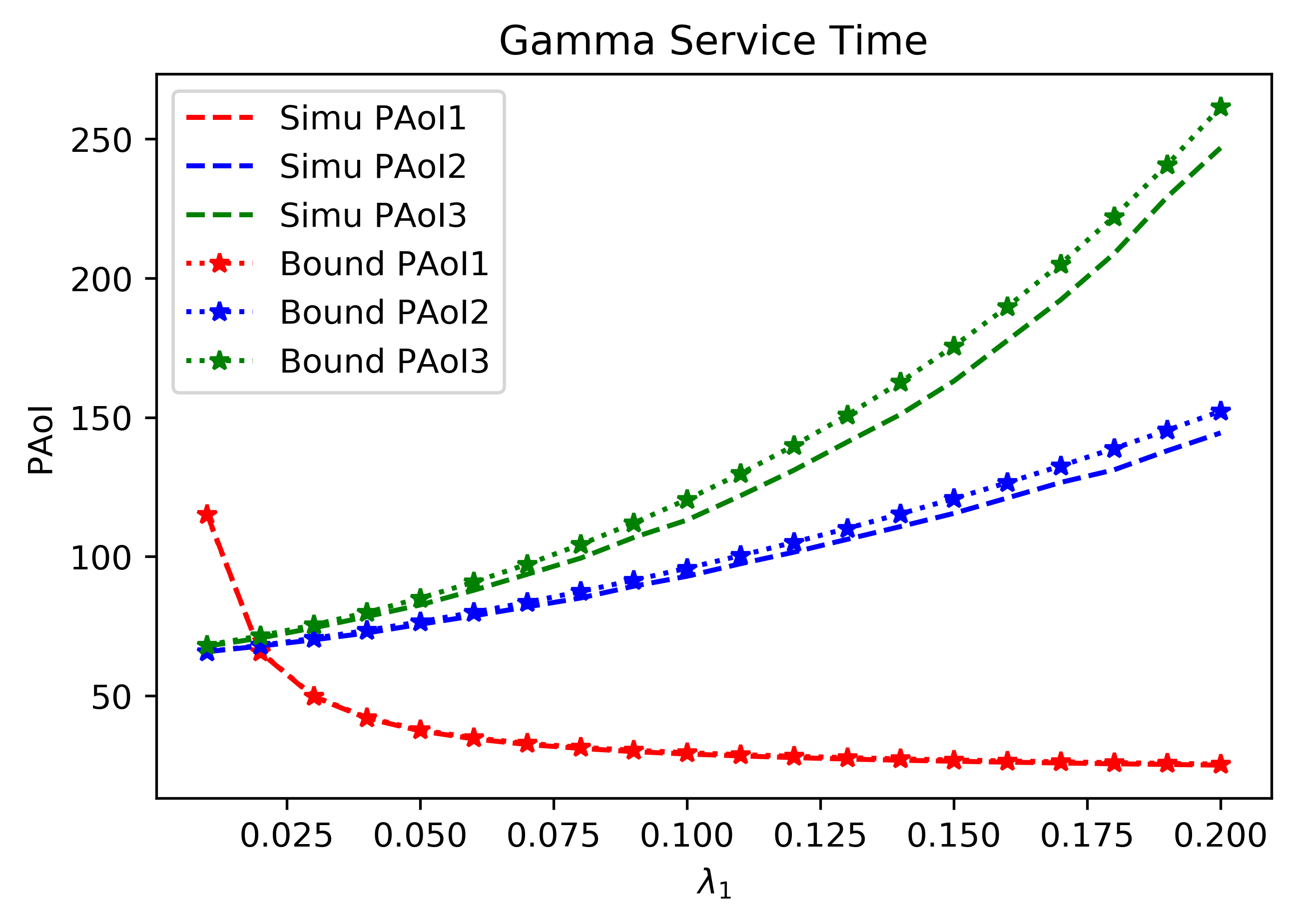}}

\caption{Bounds for M/G/1+$\sum1^{*}$ Type Queues with Buffer Size One\label{fig:Bounds-for-M/G/1}}
\end{figure}

\begin{figure}[h]
\subfloat[$\lambda_{2}=\lambda_{3}=\frac{1}{50},$ $P_{1},P_{2},P_{3}\sim exp(\frac{1}{10})$]{\includegraphics[scale=0.5]{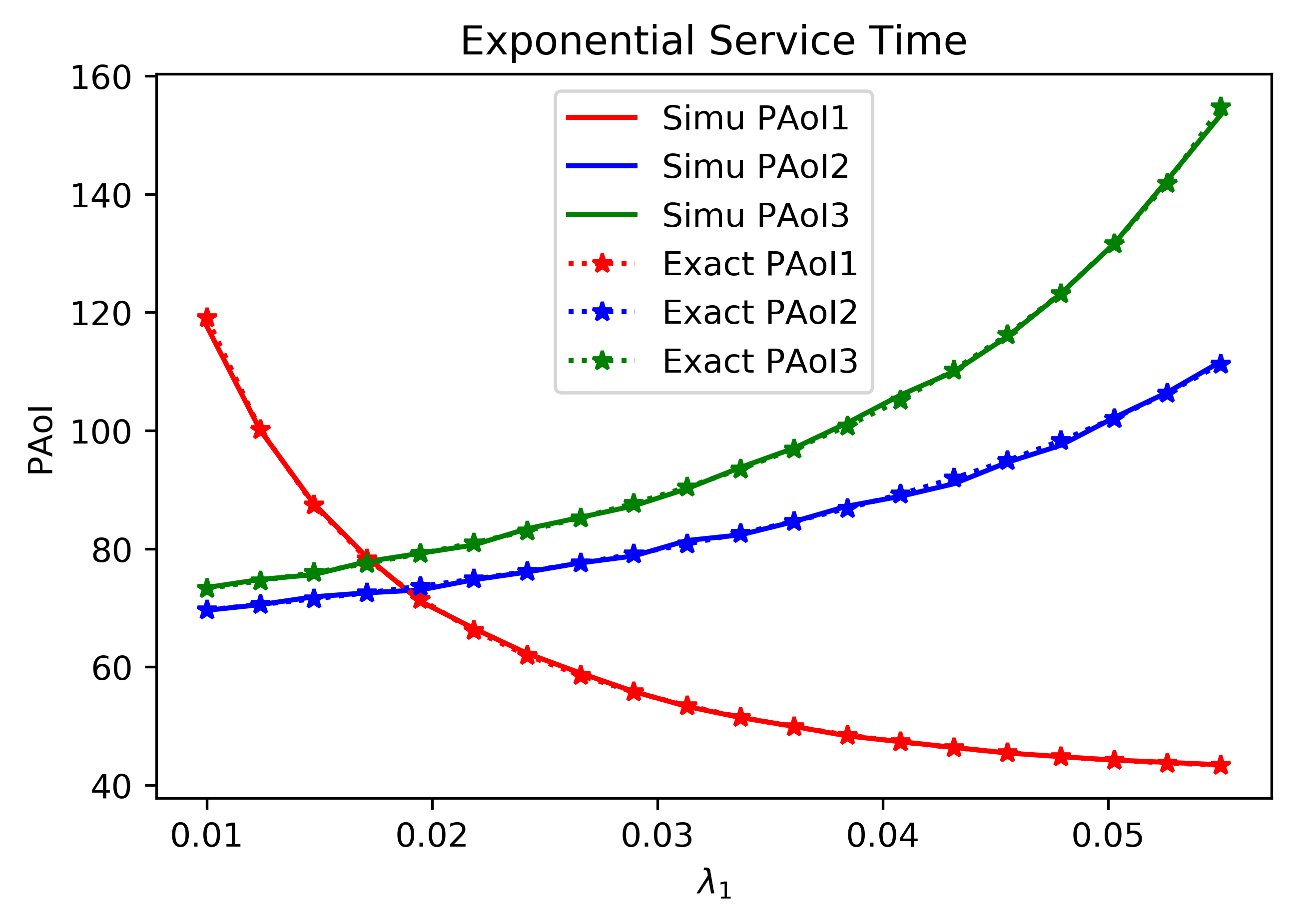}

}

\subfloat[$\lambda_{2}=\lambda_{3}=\frac{1}{50},$$P_{1},P_{2},P_{3}\sim Unif(0,20)$]{

\includegraphics[scale=0.5]{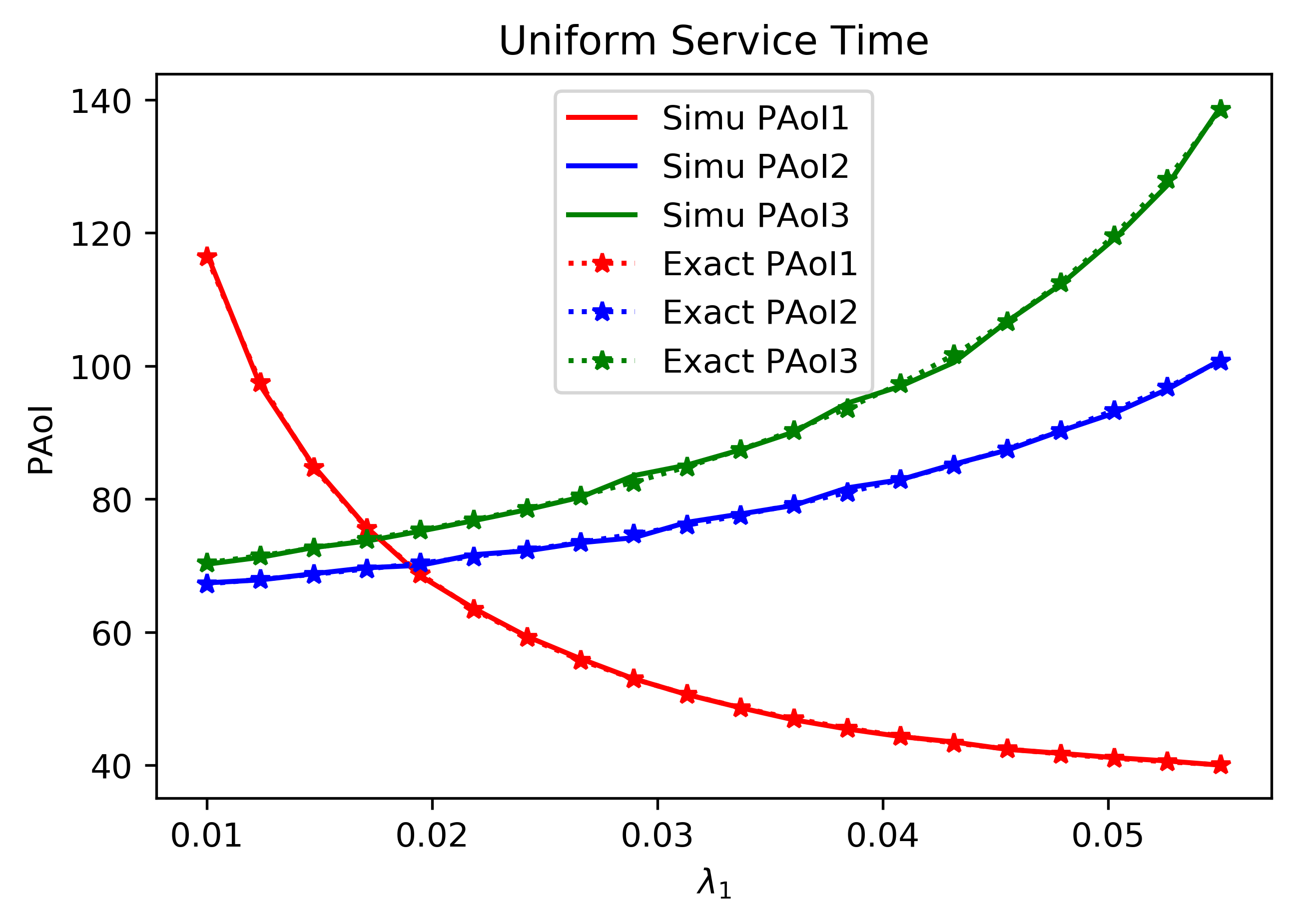}}

\subfloat[$\lambda_{2}=\lambda_{3}=\frac{1}{50},$$P_{1},P_{2},P_{3}\sim Gamma(10,1)$]{\includegraphics[scale=0.5]{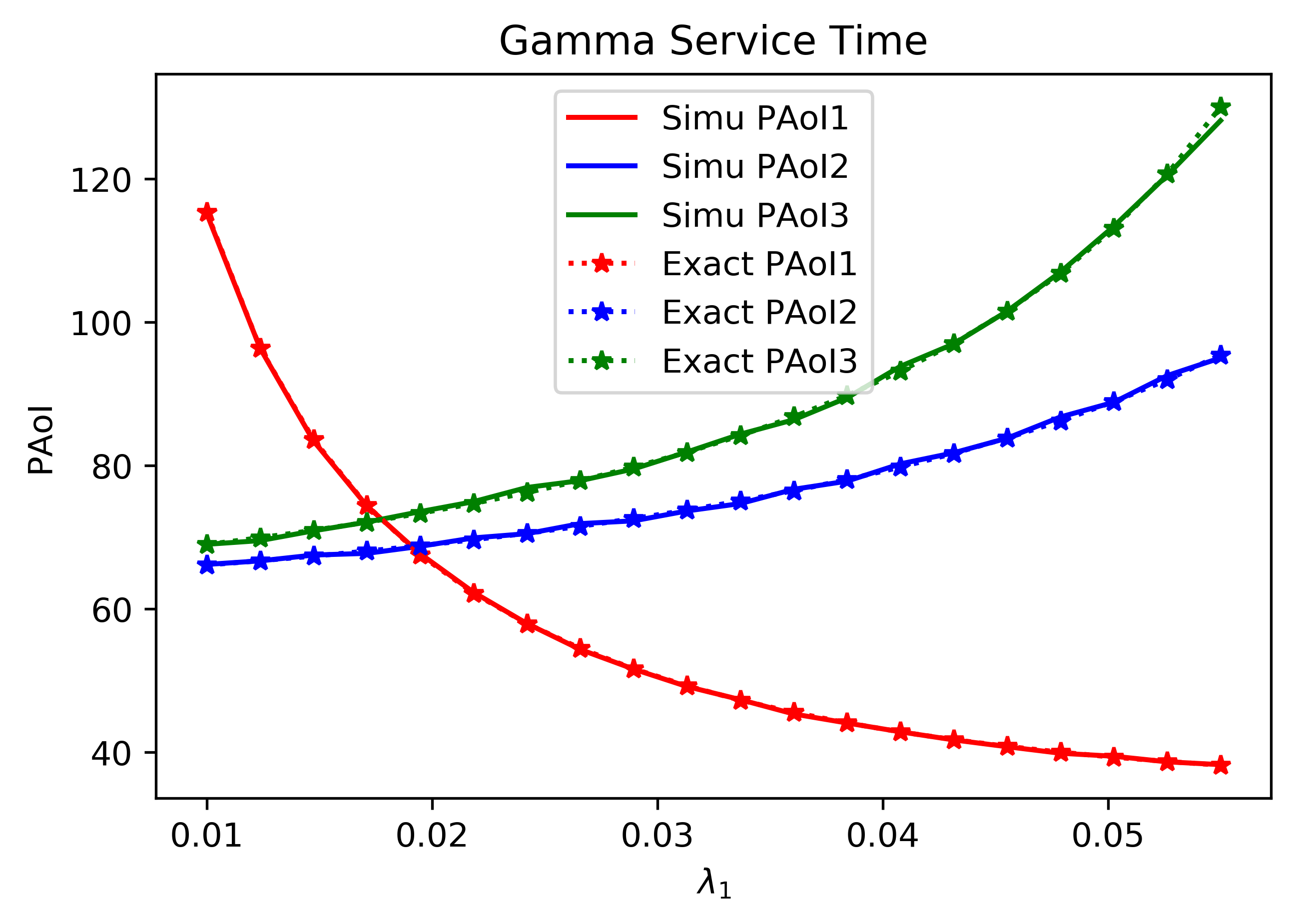}

}\caption{M/G/1 Type Queues with LCFS\label{fig:Bounds-for-M/G/1-1} }
\end{figure}

\begin{figure}[h]
\subfloat[$\lambda_{2}=\frac{1}{100},\mu_{1}=\mu_{2}=\frac{1}{10}$]{\includegraphics[scale=0.5]{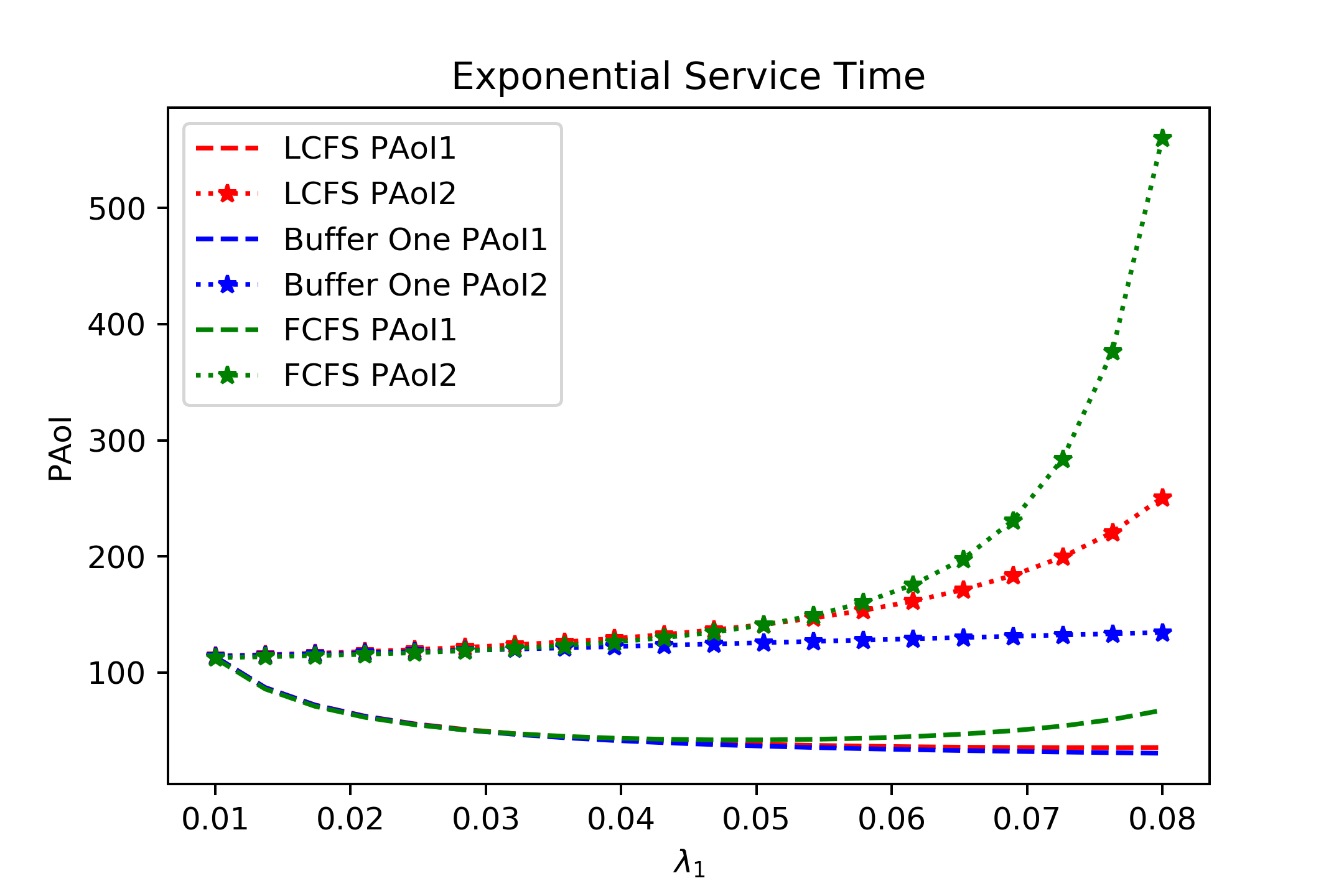}

}

\subfloat[$\lambda_{1}=\frac{1}{100},\mu_{1}=\mu_{2}=\frac{1}{10}$]{\includegraphics[scale=0.5]{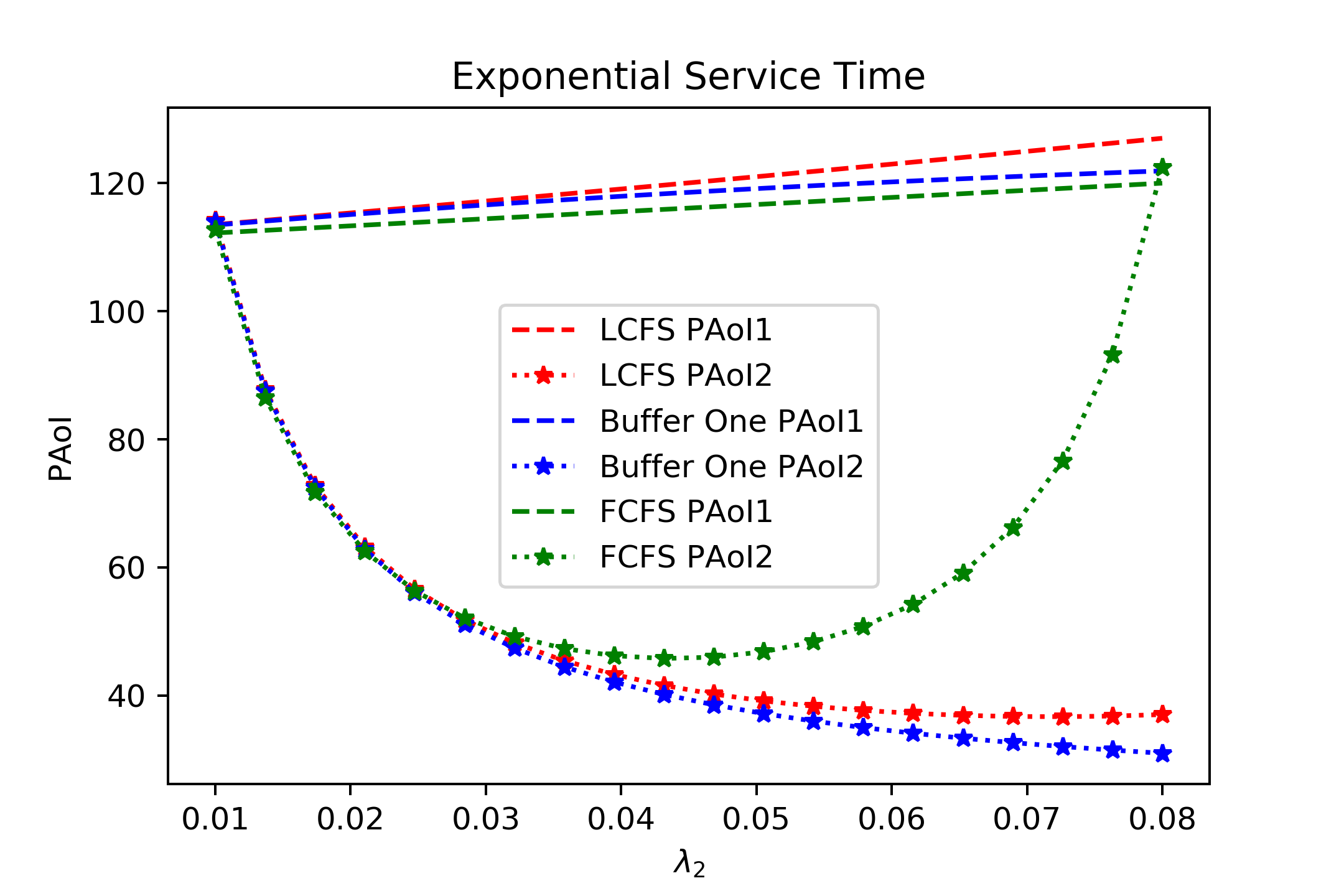}

}

\caption{PAoI under Different Service Disciplines\label{fig:PAoI-under-Different}}
\end{figure}

\begin{figure}[h]
\subfloat[$\lambda_{2}=\frac{1}{100},\mu_{1}=\mu_{2}=\frac{1}{10}$]{\includegraphics[scale=0.5]{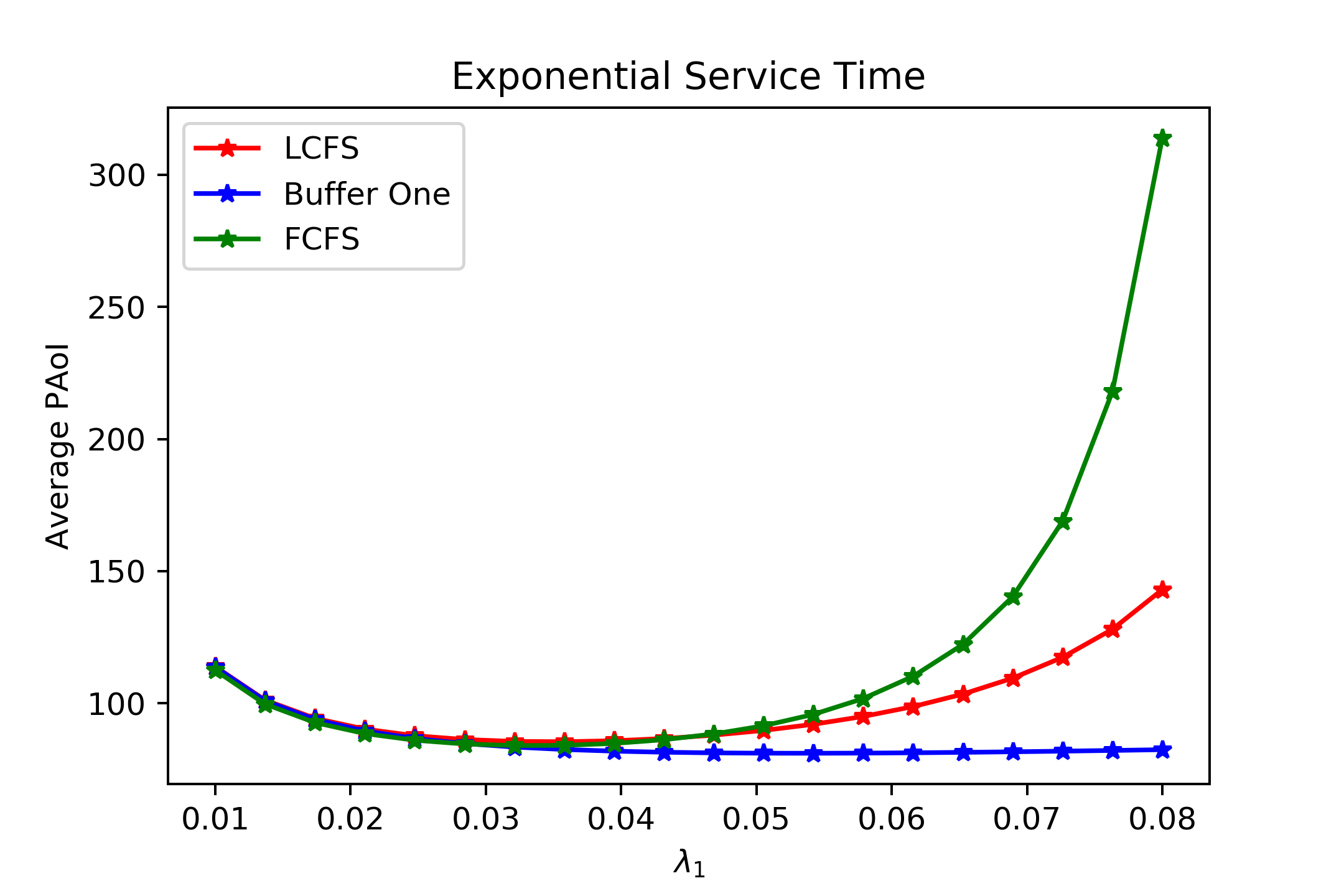}

}

\subfloat[$\lambda_{1}=\frac{1}{100},\mu_{1}=\mu_{2}=\frac{1}{10}$]{\includegraphics[scale=0.5]{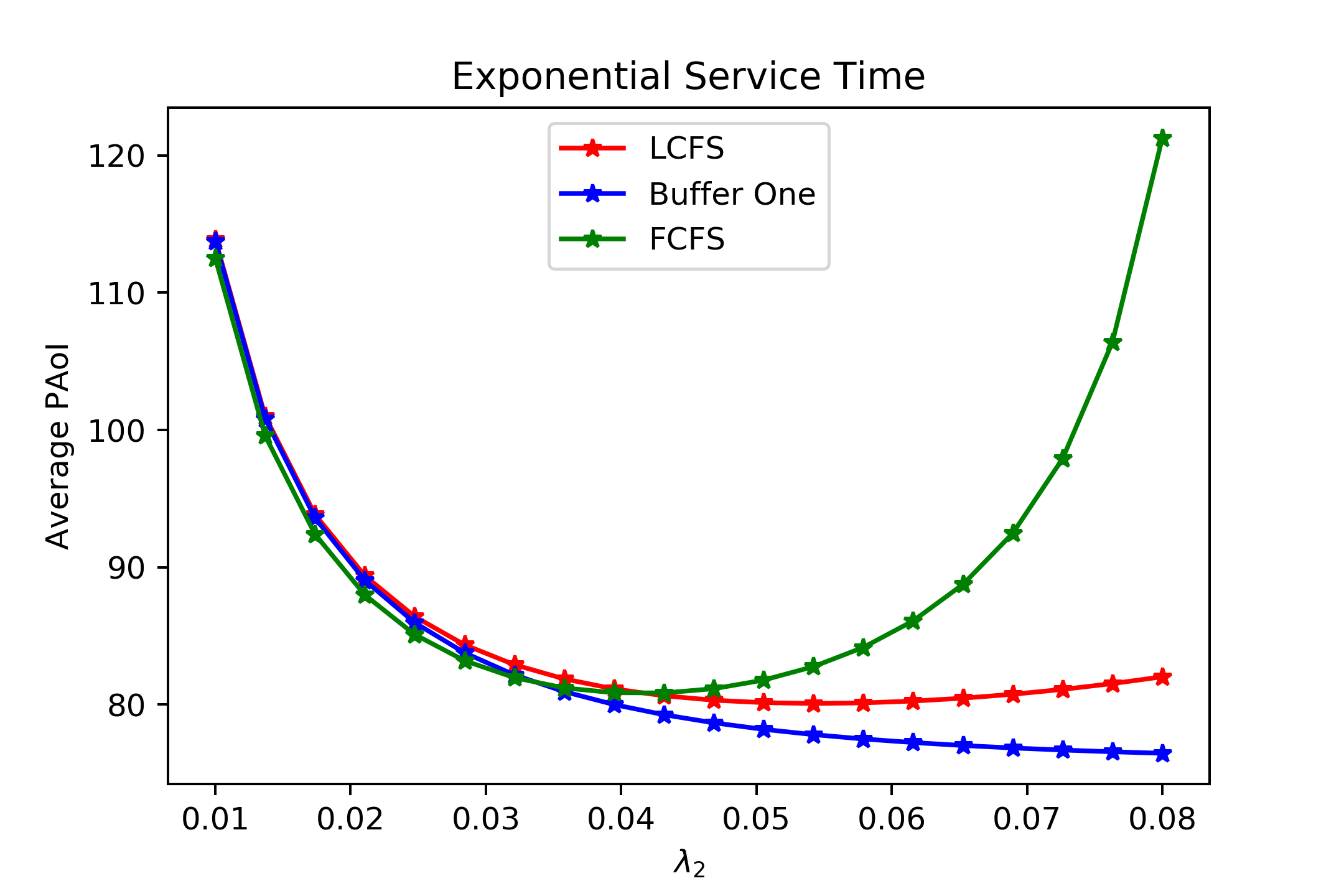}

}\caption{Average PAoI Across Queues under Different Service Disciplines \label{fig:Average-PAoI-under}}

\end{figure}

Next we address PAoI by comparing the single buffer size case against
infinite buffer size cases under FCFS and LCFS. In fact, since in
the M/G/1 system with infinite buffer size, if we keep replacing the
packets in buffers with new arrivals, then there is at most one packet
waiting in each queue, therefore the system will act exactly the same
as M/G/1+$\sum1^{*}$ system. So here we consider the PAoI under M/G/1+$\sum1^{*}$and
M/G/1 with FCFS and LCFS altogether. In Figure \ref{fig:PAoI-under-Different}
we plot the PAoI for each queue in the case of $k=2$, and in Figure
\ref{fig:Average-PAoI-under} we plot the average PAoI across all
queues ($\frac{1}{2}\sum_{i=1}^{2}\boldsymbol{E}[A_{i}]$). From Figure
\ref{fig:PAoI-under-Different} we see that under FCFS, LCFS and M/M/1+$\sum1^{*}$,
PAoI of queue 2 is sensitive to the change of $\lambda_{1}$, however
PAoI of queue 1 is less sensitive to $\lambda_{2}$. This is because
the PAoI for queue 2 highly depends on the busy time of queue 1. For
FCFS, the PAoI increases greatly when arrival rate becomes large.
This is because under FCFS, every packet that arrive to the system
needs to be processed, and increasing arrival rate enlarges the average
queue size, causing packets to wait a longer time. From the average
PAoI across queues shown in Figure \ref{fig:Average-PAoI-under},
we can see that increasing the arrival rate for the high priority
queue enlarges the PAoI much faster than increasing $\lambda_{2}$.
It indicates that when designing the priority for queues to minimize
average PAoI across queues, the one with the lowest traffic intensity
should be allocated with the highest priority. We also proved this
result in Section \ref{sec:Infinite-Buffer-Size} for M/G/1 queues
with FCFS. 

Also, it is interesting to observe that having a single-sized buffer
at each queue is not always the optimal strategy to minimize PAoI,
as we observe from Figure \ref{fig:PAoI-under-Different} and Figure
\ref{fig:Average-PAoI-under}. This fact can be seen more clearly
in Figure \ref{fig:PAoI-under-Different}(b) and Figure \ref{fig:Average-PAoI-under}(b)
when the traffic intensity of queue 2 is higher. In Figure \ref{fig:PAoI-under-Different}(b),
the PAoI of queue 1 under FCFS is lower than that under the other
two policies. In Figure \ref{fig:Average-PAoI-under}(b), FCFS results
in lower average PAoI than the other two policies when the traffic
intensity is low. This result also indicates that LCFS is not the
optimal policy among all work-conserving non-preemptive policies when
minimizing PAoI. However, as we mentioned in Section \ref{sec:Infinite-Buffer-Size},
the advantage of FCFS may only come from the special property of the
metric PAoI. 

We now compare the simulation performance of the priority queue policy
(which we introduce in this paper), Age-based Max-Weight Policy from
\cite{kadota2019minimizing}) and Randomized Policy with Weights from
\cite{talak2019optimizing} in M/M/1/$1+\sum1^{*}$ system. Under
the Age-based Max-Weight Policy, at every time $t$ the server becomes
available, it selects the non-empty queue with the highest weighted
age reduction, i.e., $\max_{i}\beta_{i}(\Delta_{i}(t)-W_{i}(t))$,
where $\beta_{i}$ is a constant weight, $\Delta_{i}(t)$ is the age
for queue $i$, and $W_{i}(t)$ is the waiting time of the packet
in buffer $i$. Under the Randomized Policy with Weights, the server
would pick queue $i$ with probability $\frac{\beta_{i}}{\sum_{i\in\mathcal{E}(t)}\beta_{i}}$,
where $\mathcal{E}(t)$ is the set of non-empty queues at time $t$.
The simulation results for these three policies in a two-queue system
with $\beta_{1}=0.8$ and $\beta_{2}=0.2$ are provided in Figure
\ref{fig:Policy-Comparison}. As we see from Figure \ref{fig:Policy-Comparison}(a),
when the traffic intensity of queue 1 is large, the priority queue
policy has similar performance to the other two policies on queue
1, but the priority queue policy causes a higher PAoI for queue 2.
This is because the priority queue policy has a stronger preference
on queue 1 than the other two policies. Under the priority queue policy,
queue 2 can only be served when queue 1 is empty, whereas queue 2
can still be served when both queues are non-empty under the other
two policies. When traffic intensity of queue 2 is large, the priority
queue policy results in a slightly smaller PAoI for queue 1 than the
Randomized Policy with Weights, as we can see from Figure \ref{fig:Policy-Comparison}(b).
It is because in this two-queue case, the priority queue policy can
be regarded as the Randomized Policy with $\beta_{1}=1$ and $\beta_{2}=0$.
Therefore under the priority queue policy, queue 1 is given higher
preference than that under the Randomized Policy with $\beta_{1}=0.8$
and $\beta_{2}=0.2$. Like we mentioned in Section \ref{sec:Introduction},
there could be data sources which have information more important
or time-sensitive than the other data sources. The priority queue
policy actually helps reduce the PAoI of these data sources by assigning
them with high static priorities. When considering the average PAoI
across queues, as shown in Figure \ref{fig:Policy-Comparison}(c)
and (d), priority queue policy can result in a slightly smaller average
PAoI than the Randomized Policy with $\beta_{1}=0.8$ and $\beta_{2}=0.2$,
if high priority is assigned to the queue with low traffic intensity.
This also implies that in a system where traffic intensities of queues
are distinct, high priority is recommended to be given to the low
traffic queue to reduce the average PAoI over queues. 

\begin{figure*}
\hfill{}\subfloat[$\lambda_{2}=\frac{1}{100},\mu_{1}=\mu_{2}=\frac{1}{10}$]{\includegraphics[scale=0.5]{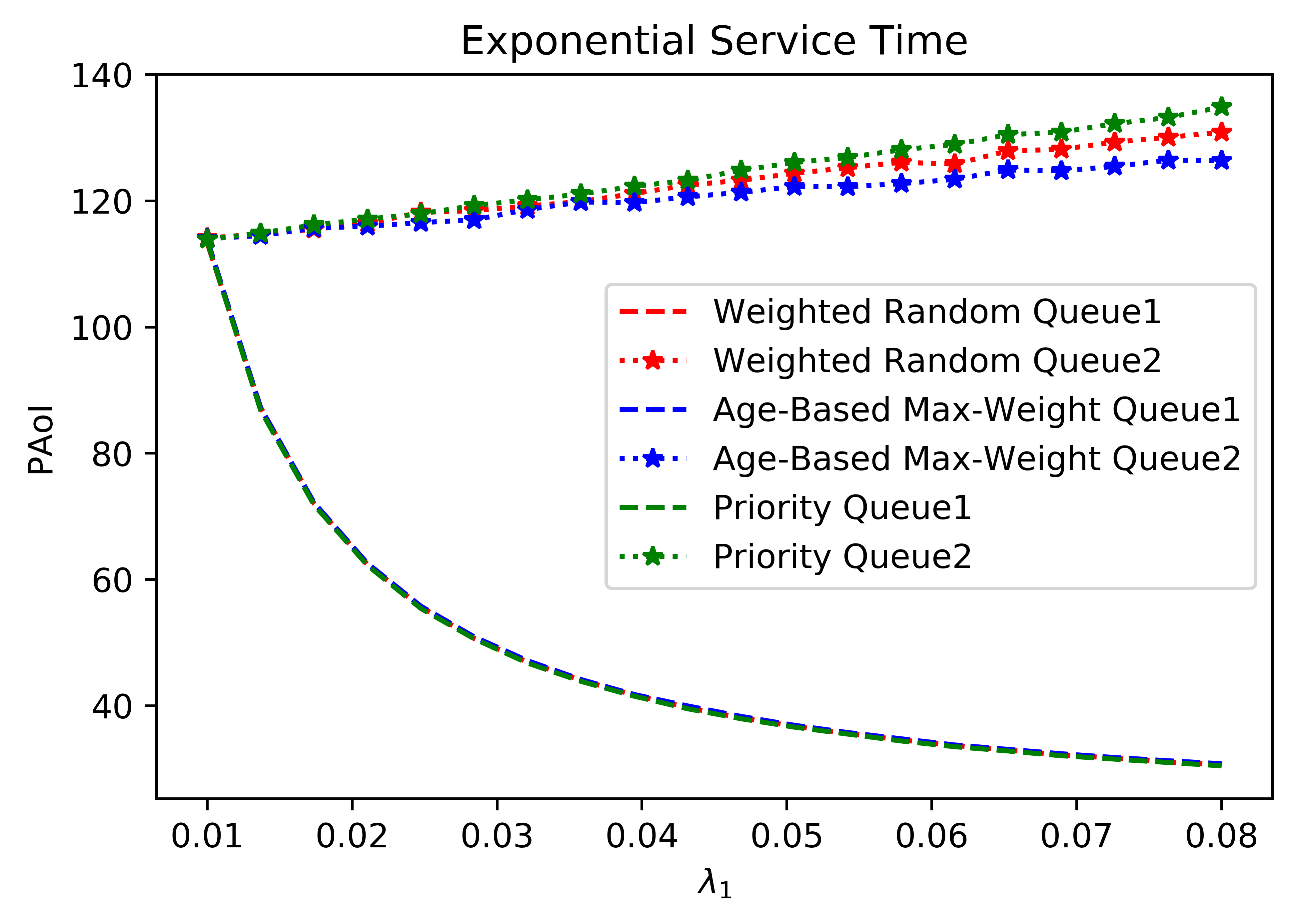}

}\hfill{}\subfloat[$\lambda_{1}=\frac{1}{100},\mu_{1}=\mu_{2}=\frac{1}{10}$]{\includegraphics[scale=0.5]{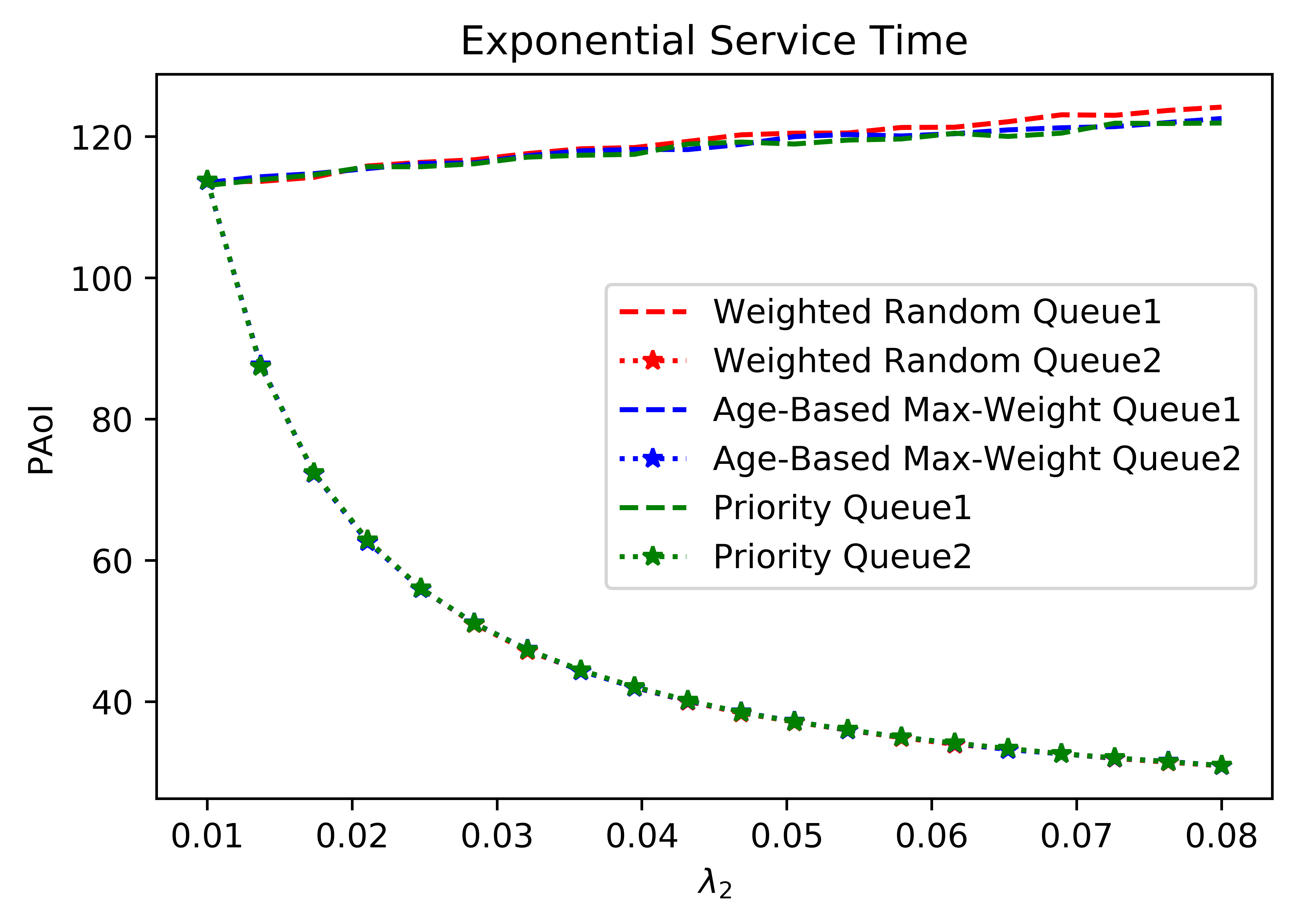}

}\hfill{}

\hfill{}\subfloat[$\lambda_{2}=\frac{1}{100},\mu_{1}=\mu_{2}=\frac{1}{10}$]{\includegraphics[scale=0.5]{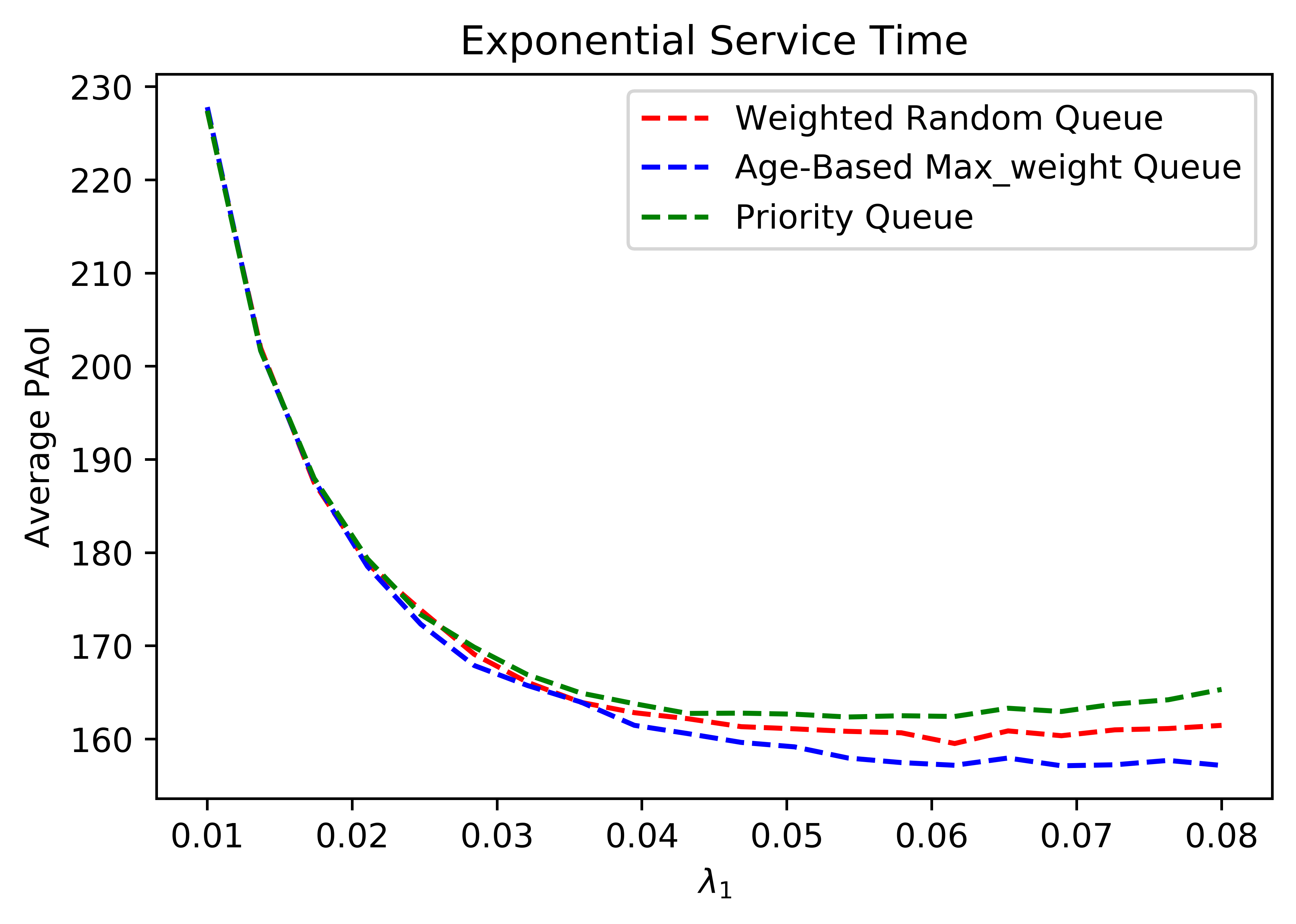}

}\hfill{}\subfloat[$\lambda_{1}=\frac{1}{100},\mu_{1}=\mu_{2}=\frac{1}{10}$]{\includegraphics[scale=0.5]{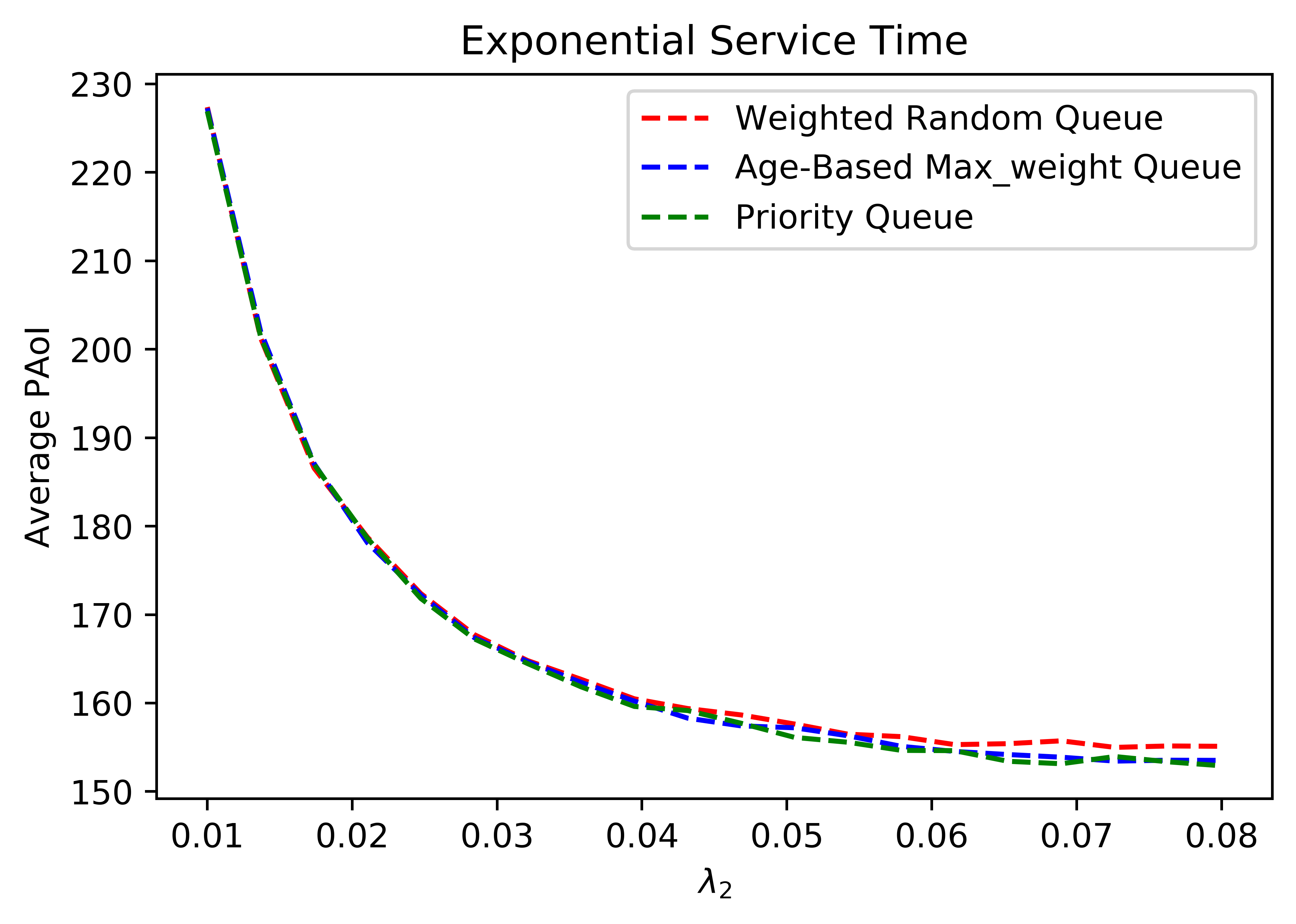}

}\hfill{}

\caption{Policy Comparison \label{fig:Policy-Comparison}}

\end{figure*}

\section{\label{sec:Conclusions-and-Future}Concluding Remarks and Future
Work}

In this research we considered a multi-class queueing system where
each class of data source generates packets according to a Poisson
process and a single processor uses a static priority scheme to serve
packets. We characterized the PAoI for such a system under two situations:
(i) when the buffer size for each queue is one; (ii) when the buffer
size for each queue is infinite and service disciplines within each
queue can be FCFS or LCFS. We obtained exact expressions for PAoI
in case (i) when the service times are exponential, and bounds (which
serve as excellent approximations) for case (i) when service times
are generally distributed. The method of obtaining PAoI for case (i)
becomes cumbersome when the number of queues is large. For case (ii)
with general service times, we provided the exact methods for calculating
PAoI, and this method can be applied when the number of queues is
large. 

Using PAoI results, we made a few observations that are useful in
determining priorities, service disciplines and sampling rates. We
found that LCFS is not the optimal service discipline in minimizing
PAoI, and we also found that systems with buffer size one at each queue
does not always provide smaller PAoI than the systems with infinite
buffer size. These are due to the special definition of the metric
PAoI. We further discussed the merits and limitations of the metric
PAoI.

From both analytical and numerical results, we showed that for minimizing
the average PAoI across queues, it is beneficial to give higher priorities
to queues with lower traffic intensities. Besides, we found that the
PAoI of queues with low priorities are more sensitive to the packet
arrival rate of high priority queues. Increasing the arrival rate
for one queue, while reducing the PAoI for this certain data source,
would significantly increase the PAoI of queues with lower priorities. 

Since in this paper we mainly focus on static queue priorities, in
our future work we will consider a system with dynamic priorities.
Besides, in smart manufacturing systems where the status of machines
changes over time, sampling with a time-varying rate is also possible
and it is interesting to consider the PAoI with time-varying arrival
rates. Moreover, the variance of PAoI is also useful in measuring
the data freshness in real-time systems, and the distribution of PAoI
is also of interest. Thus, there are numerous opportunities for research
in the area of PAoI for multi-priority queues.

\bibliographystyle{ieeetr}
\bibliography{Jin_AoI_1}

\begin{thebibliography}{10}

\bibitem{kaul2012real}
S.~Kaul, R.~Yates, and M.~Gruteser, ``Real-time status: How often should one
  update?,'' in {\em INFOCOM, 2012 Proceedings IEEE}, pp.~2731--2735, IEEE,
  2012.

\bibitem{huang2015optimizing}
L.~Huang and E.~Modiano, ``Optimizing age-of-information in a multi-class
  queueing system,'' in {\em 2015 IEEE International Symposium on Information
  Theory (ISIT)}, pp.~1681--1685, IEEE, 2015.

\bibitem{inoue2019general}
Y.~Inoue, H.~Masuyama, T.~Takine, and T.~Tanaka, ``A general formula for the
  stationary distribution of the age of information and its application to
  single-server queues,'' {\em IEEE Transactions on Information Theory},
  vol.~65, no.~12, pp.~8305--8324, 2019.

\bibitem{masry1978poisson}
E.~Masry, ``Poisson sampling and spectral estimation of continuous-time
  processes,'' {\em IEEE Transactions on Information Theory}, vol.~24, no.~2,
  pp.~173--183, 1978.

\bibitem{najm2019content}
E.~Najm, R.~Nasser, and E.~Telatar, ``Content based status updates,'' {\em IEEE
  Transactions on Information Theory}, 2019.

\bibitem{maatouk2019age}
A.~Maatouk, M.~Assaad, and A.~Ephremides, ``Age of information with prioritized
  streams: When to buffer preempted packets?,'' in {\em 2019 IEEE International
  Symposium on Information Theory (ISIT)}, pp.~325--329, IEEE, 2019.

\bibitem{costa2016age}
M.~Costa, M.~Codreanu, and A.~Ephremides, ``On the age of information in status
  update systems with packet management,'' {\em IEEE Transactions on
  Information Theory}, vol.~62, no.~4, pp.~1897--1910, 2016.

\bibitem{zou2019benefis}
P.~{Zou}, O.~{Ozel}, and S.~{Subramaniam}, ``Waiting before serving: A
  companion to packet management in status update systems,'' {\em IEEE
  Transactions on Information Theory}, vol.~66, no.~6, pp.~3864--3877, 2020.

\bibitem{theodoratos1999data}
D.~Theodoratos and M.~Bouzeghoub, ``Data currency quality factors in data
  warehouse design.,'' in {\em DMDW}, p.~15, 1999.

\bibitem{bouzeghoub2004framework}
M.~Bouzeghoub, ``A framework for analysis of data freshness,'' in {\em
  Proceedings of the 2004 international workshop on Information quality in
  information systems}, pp.~59--67, ACM, 2004.

\bibitem{al2010hedera}
M.~Al-Fares, S.~Radhakrishnan, B.~Raghavan, N.~Huang, and A.~Vahdat, ``Hedera:
  dynamic flow scheduling for data center networks.,'' in {\em Nsdi}, vol.~10,
  pp.~89--92, 2010.

\bibitem{vaquero2014finding}
L.~M. Vaquero and L.~Rodero-Merino, ``Finding your way in the fog: Towards a
  comprehensive definition of fog computing,'' {\em ACM SIGCOMM Computer
  Communication Review}, vol.~44, no.~5, pp.~27--32, 2014.

\bibitem{najm2016age}
E.~Najm and R.~Nasser, ``Age of information: The gamma awakening,'' in {\em
  2016 IEEE International Symposium on Information Theory (ISIT)},
  pp.~2574--2578, IEEE, 2016.

\bibitem{soysal2018age}
A.~Soysal and S.~Ulukus, ``Age of information in {G/G/1/1} systems,'' {\em
  arXiv preprint arXiv:1805.12586}, 2018.

\bibitem{kosta2019queue}
A.~Kosta, N.~Pappas, A.~Ephremides, and V.~Angelakis, ``Queue management for
  age sensitive status updates,'' in {\em 2019 IEEE International Symposium on
  Information Theory (ISIT)}, pp.~330--334, IEEE, 2019.

\bibitem{najm2018status}
E.~Najm and E.~Telatar, ``Status updates in a multi-stream {M/G/1/1} preemptive
  queue,'' in {\em IEEE Infocom 2018-IEEE Conference On Computer Communications
  Workshops (Infocom Wkshps)}, pp.~124--129, IEEE, 2018.

\bibitem{kosta2019age}
A.~Kosta, N.~Pappas, A.~Ephremides, and V.~Angelakis, ``Age of information
  performance of multiaccess strategies with packet management,'' {\em Journal
  of Communications and Networks}, vol.~21, no.~3, pp.~244--255, 2019.

\bibitem{jiang2018can}
Z.~Jiang, B.~Krishnamachari, S.~Zhou, and Z.~Niu, ``Can decentralized status
  update achieve universally near-optimal age-of-information in wireless
  multiaccess channels?,'' in {\em 2018 30th International Teletraffic Congress
  (ITC 30)}, vol.~1, pp.~144--152, IEEE, 2018.

\bibitem{maatouk2020optimality}
A.~Maatouk, S.~Kriouile, M.~Assaad, and A.~Ephremides, ``On the optimality of
  the whittle's index policy for minimizing the age of information,'' {\em
  arXiv preprint arXiv:2001.03096}, 2020.

\bibitem{kadota2019scheduling}
I.~Kadota, A.~Sinha, and E.~Modiano, ``Scheduling algorithms for optimizing age
  of information in wireless networks with throughput constraints,'' {\em
  IEEE/ACM Transactions on Networking}, vol.~27, no.~4, pp.~1359--1372, 2019.

\bibitem{talak2019optimizing}
R.~{Talak}, S.~{Karaman}, and E.~{Modiano}, ``Optimizing information freshness
  in wireless networks under general interference constraints,'' {\em IEEE/ACM
  Transactions on Networking}, vol.~28, no.~1, pp.~15--28, 2020.

\bibitem{he2017optimal}
Q.~He, D.~Yuan, and A.~Ephremides, ``Optimal link scheduling for age
  minimization in wireless systems,'' {\em IEEE Transactions on Information
  Theory}, vol.~64, no.~7, pp.~5381--5394, 2017.

\bibitem{hsu2017age}
Y.-P. Hsu, E.~Modiano, and L.~Duan, ``Age of information: Design and analysis
  of optimal scheduling algorithms,'' in {\em 2017 IEEE International Symposium
  on Information Theory (ISIT)}, pp.~561--565, IEEE, 2017.

\bibitem{jiang2018timely}
Z.~Jiang, B.~Krishnamachari, X.~Zheng, S.~Zhou, and Z.~Niu, ``Timely status
  update in massive iot systems: Decentralized scheduling for wireless
  uplinks,'' {\em arXiv preprint arXiv:1801.03975}, 2018.

\bibitem{jiang2019timely}
Z.~Jiang, B.~Krishnamachari, X.~Zheng, S.~Zhou, and Z.~Niu, ``Timely status
  update in wireless uplinks: Analytical solutions with asymptotic
  optimality,'' {\em IEEE Internet of Things Journal}, vol.~6, no.~2,
  pp.~3885--3898, 2019.

\bibitem{kadota2018scheduling}
I.~Kadota, A.~Sinha, E.~Uysal-Biyikoglu, R.~Singh, and E.~Modiano, ``Scheduling
  policies for minimizing age of information in broadcast wireless networks,''
  {\em IEEE/ACM Transactions on Networking}, vol.~26, no.~6, pp.~2637--2650,
  2018.

\bibitem{yates2019age}
R.~D. Yates and S.~K. Kaul, ``The age of information: Real-time status updating
  by multiple sources,'' {\em IEEE Transactions on Information Theory},
  vol.~65, no.~3, pp.~1807--1827, 2019.

\bibitem{jaiswal1968priority}
N.~K. Jaiswal, {\em Priority queues}.
\newblock Elsevier, 1968.

\bibitem{adan2001queueing}
I.~Adan, O.~J. Boxma, and J.~A.~C. Resing, ``Queueing models with multiple
  waiting lines,'' {\em Queueing Systems}, vol.~37, no.~1-3, pp.~65--98, 2001.

\bibitem{kaul2018age}
S.~K. Kaul and R.~D. Yates, ``Age of information: Updates with priority,'' in
  {\em 2018 IEEE International Symposium on Information Theory (ISIT)},
  pp.~2644--2648, IEEE, 2018.

\bibitem{little2011or}
J.~D.~C. Little, ``Or forum-little's law as viewed on its 50th anniversary,''
  {\em Operations Research}, vol.~59, no.~3, pp.~536--549, 2011.

\bibitem{kulkarni2016modeling}
V.~G. Kulkarni, {\em Modeling and analysis of stochastic systems}.
\newblock Chapman and Hall/CRC, 2016.

\bibitem{conway2003theory}
R.~W. Conway, W.~L. Maxwell, and L.~W. Miller, {\em Theory of scheduling}.
\newblock Courier Corporation, 2003.

\bibitem{takenaka1989analysis}
T.~Takenaka, ``Analysis of a nonpreemptive {$\Sigma M_i/G/1 (\Sigma N_i)$}
  system,'' {\em Electronics and Communications in Japan (Part I:
  Communications)}, vol.~72, no.~3, pp.~75--84, 1989.

\bibitem{takenaka1984buffer}
T.~Takenaka, ``Buffer management schemes for a heterogeneous packet switching
  system,'' {\em Electronics and Communications in Japan (Part I:
  Communications)}, vol.~67, no.~11, pp.~46--54, 1984.

\bibitem{takenaka1989characteristics}
T.~Takenaka, T.~Akaike, and K.~Takami, ``Characteristics and approximation
  methods of a nonpreemptive {$\Sigma M_i/G/1 (\Sigma N_i) $} system,'' {\em
  Electronics and Communications in Japan (Part I: Communications)}, vol.~72,
  no.~3, pp.~85--94, 1989.

\bibitem{gautam2012analysis}
N.~Gautam, {\em Analysis of queues: methods and applications}.
\newblock CRC Press, 2012.

\bibitem{bedewy2018age}
A.~M. Bedewy, Y.~Sun, S.~Kompella, and N.~B. Shroff, ``Age-optimal sampling and
  transmission scheduling in multi-source systems,'' {\em arXiv preprint
  arXiv:1812.09463}, 2018.

\bibitem{kella1988priorities}
O.~Kella and U.~Yechiali, ``Priorities in {M/G/1} queue with server
  vacations,'' {\em Naval Research Logistics (NRL)}, vol.~35, no.~1,
  pp.~23--34, 1988.

\bibitem{cheney2012numerical}
E.~W. Cheney and D.~R. Kincaid, {\em Numerical mathematics and computing}.
\newblock Cengage Learning, 2012.

\bibitem{xu2020vacations}
J.~Xu, I.-H. Hou, and N.~Gautam, ``Age of information for single buffer systems
  with vacation server,'' {\em arXiv preprint arXiv:2004.11847}, 2020.

\bibitem{bedewy2019age}
A.~M. Bedewy, Y.~Sun, and N.~B. Shroff, ``The age of information in multihop
  networks,'' {\em IEEE/ACM Transactions on Networking}, vol.~27, no.~3,
  pp.~1248--1257, 2019.

\bibitem{talak2018can}
R.~Talak, S.~Karaman, and E.~Modiano, ``Can determinacy minimize age of
  information?,'' {\em arXiv preprint arXiv:1810.04371}, 2018.

\bibitem{chen2016age}
K.~Chen and L.~Huang, ``Age-of-information in the presence of error,'' in {\em
  2016 IEEE International Symposium on Information Theory (ISIT)},
  pp.~2579--2583, IEEE, 2016.

\bibitem{ostman2019peak}
J.~{\"O}stman, R.~Devassy, G.~Durisi, and E.~Uysal, ``Peak-age violation
  guarantees for the transmission of short packets over fading channels,'' in
  {\em IEEE INFOCOM 2019-IEEE Conference on Computer Communications Workshops
  (INFOCOM WKSHPS)}, pp.~109--114, IEEE, 2019.

\bibitem{devassy2019reliable}
R.~Devassy, G.~Durisi, G.~C. Ferrante, O.~Simeone, and E.~Uysal, ``Reliable
  transmission of short packets through queues and noisy channels under latency
  and peak-age violation guarantees,'' {\em IEEE Journal on Selected Areas in
  Communications}, vol.~37, no.~4, pp.~721--734, 2019.

\bibitem{chiariotti2020peak}
F.~Chiariotti, O.~Vikhrova, B.~Soret, and P.~Popovski, ``Peak age of
  information distribution in tandem queue systems,'' {\em arXiv preprint
  arXiv:2004.05088}, 2020.

\bibitem{he2016optimal}
Q.~He, D.~Yuan, and A.~Ephremides, ``On optimal link scheduling with min-max
  peak age of information in wireless systems,'' in {\em 2016 IEEE
  International Conference on Communications (ICC)}, pp.~1--7, IEEE, 2016.

\bibitem{kosta2020cost}
A.~Kosta, N.~Pappas, A.~Ephremides, and V.~Angelakis, ``The cost of delay in
  status updates and their value: Non-linear ageing,'' {\em IEEE Transactions
  on Communications}, 2020.

\bibitem{kosta2017age}
A.~Kosta, N.~Pappas, and V.~Angelakis, ``Age of information: A new concept,
  metric, and tool,'' {\em Foundations and Trends in Networking}, vol.~12,
  no.~3, pp.~162--259, 2017.

\bibitem{tripathi2019age}
V.~Tripathi, R.~Talak, and E.~Modiano, ``Age optimal information gathering and
  dissemination on graphs,'' in {\em IEEE INFOCOM 2019-IEEE Conference on
  Computer Communications}, pp.~2422--2430, IEEE, 2019.

\bibitem{kadota2019minimizing}
I.~Kadota and E.~Modiano, ``Minimizing the age of information in wireless
  networks with stochastic arrivals,'' {\em IEEE Transactions on Mobile
  Computing}, 2019.

\end{thebibliography}

\end{document}